\documentclass[3p]{elsarticle} %3p,draft
\makeatletter
\def\ps@pprintTitle{%
 \let\@oddhead\@empty
 \let\@evenhead\@empty
 \def\@oddfoot{\centerline{\thepage}}%
 \let\@evenfoot\@oddfoot}
\makeatother

\date{}
\def\texpsfig#1#2#3{\vbox{\kern #3\hbox{\includegraphics{#1}\kern #2}}\typeout{(#1)}}

\usepackage{url,hyperref}
\usepackage{bbm}
\usepackage{eurosym}                           % Symbol for Euro
\usepackage[latin1]{inputenc}                  % Input encoding for Western Europe
\usepackage{url}                               % Display URLs in typewriter font
\usepackage{longtable}                         % Tables running of several pages
\usepackage{array}                             % more table options
\usepackage{graphicx,color}
\usepackage{amsthm}
\usepackage{amsbsy}

\usepackage{amssymb}
\usepackage{amsmath}
\usepackage{enumerate}
\usepackage{graphicx,color}
\usepackage{epsfig}
\usepackage{multirow,bigdelim}
\usepackage[table]{xcolor}% http://ctan.org/pkg/xcolor
\usepackage{natbib}
\theoremstyle{plain}
\newtheorem{thm}{Theorem}[section]

\newtheorem*{rem}{Remark}
\theoremstyle{remark}

\theoremstyle{plain}
\newtheorem{lem}[thm]{Lemma}

\theoremstyle{definition}

\newcommand{\e}{{\rm e}}        % "e" number
\def\R{\mathbb{ R}}             % Real number
\def\E{\mathbb{ E}}             % Expectation
\def\Q{\mathbb{ Q}}             % Measure Q

\def\P{\mathbb{ P}}             % Measure P
\def\F{\mathcal{F}}             % Filtration
\def\var{\mathbb{V}\text{ar}}   % Variance
\def\Var{\mathbb{V}\text{ar}}   % Variance
\renewcommand{\d}{{\rm d}}      % straight "d" in in integration and ODEs, \int_a^b f(x)\d x
\def\dW{{\rm d}W}               % dW in SDEs- Brownian noise.
\def\dt{{\rm d}t}

\def\dx{{\rm d}x}

\def\T{{\rm T}}
\def\1{{\mathbbm{1}}}            % Indicator function

\theoremstyle{plain}% default

\usepackage[margin=1cm]{caption}
\usepackage{algorithm}
\usepackage[noend]{algpseudocode}

\numberwithin{equation}{section}	     %Equation numbering per section
\title{Sparse Grid Method for Highly Efficient Computation of Exposures for xVA}
\begin{document}
\author[1,2]{Lech A.~Grzelak\corref{cor1}}
\ead{L.A.Grzelak@tudelft.nl}
\cortext[cor1]{Corresponding author at Delft Institute of Applied Mathematics, TU Delft, Delft, the Netherlands.}
\address[1]{Delft Institute of Applied Mathematics, Delft University of Technology, Delft, the Netherlands}
\address[2]{Rabobank, Utrecht, the Netherlands}

\begin{abstract}
    \noindent Every ``x''-adjustment in the so-called xVA financial risk management framework relies on the computation of exposures.  Considering thousands of Monte Carlo paths and tens of simulation steps, a financial portfolio needs to be evaluated numerous times during the lifetime of the underlying assets. This  is the bottleneck of every simulation of xVA.

In this article, we explore numerical techniques for improving the simulation of exposures. We aim to decimate the number of portfolio evaluations, particularly for large portfolios involving multiple, correlated risk factors. The usage of the Stochastic Collocation (SC) method~\cite{grzelak2015stochastic}, together with Smolyak's~\cite{Smol63,JUDD2014} sparse grid extension, allows for a significant reduction in the number of portfolio evaluations, even when dealing with many risk factors. The proposed model can be easily applied to any portfolio and size.We report that for a realistic portfolio comprising linear and non-linear derivatives, the expected reduction in the portfolio evaluations may exceed 6000 times, depending on the dimensionality and the required accuracy. We give illustrative examples and examine the method with realistic multi-currency portfolios consisting of interest rate swaps and swaptions.
\end{abstract}

\maketitle

\section{Introduction}
\label{sec:introduction}
%{\let\thefootnote\relax\footnotetext{The views expressed in this paper are the personal views of the authors and do not necessarily reflect the views or policies of their current or past employers.}}

Since the Basel Committee introduced their requirements for the credit exposures computation~\cite{Besel}, many banks needed to reform their pricing infrastructure.
Calculation of risk indicators like expected exposure (EE) or potential future exposures (PFE) is crucial to assess the safety of a bank's positions against market movements in the future. From the pricing perspective, however, these calculations require substantial computational effort. Both indicators need hypothetical scenarios using models calibrated to the financial market. These scenarios represent potential movements of the risk factors in the future. Therefore, to ``measure'' the bank's exposure, the portfolios need to be evaluated for many future scenarios~\cite{GregoryBook,OosterleeGrzelakBook}, which is a highly intensive task, especially for large portfolios involving thousands of trades that depend on different risk factors.

For large financial institutions, like banks and hedge/pension funds, most of the products are linear, non-exotic. This is particularly true since illiquid derivatives are often capital intensive due to regulatory requirements. Although a portfolio of linear products is straightforward to value, because of the volume of trades (often exceeding hundreds of thousands of trades), the complete xVA computation may take multiple hours. Especially if a portfolio involves many long-dated swaps with daily compounding or averaging~\footnote{Daily compounding of rates can be seen, among others, in SOFR, ESTR, FF, Brazilian market.}.

Evaluation of xVA poses a significant challenge to researchers and practitioners to improve the required  computational time. Although the problem exists for more than a decade already, no genuine progress in  efficiency has been achieved. The common approach to enhance the speed focuses on reducing the number of Monte Carlo paths or exposure (monitoring) dates. Either of these choices is at the cost of quality and stability of the ultimate results. An alternative way to improving computational speed is significant investments in hardware, like multi-core CPUs/GPUs.~\cite{GreenXVA,CrepeyGPUXVA}.

This article focuses on the efficient computation of exposures, where many risk factors rule out traditional PDE-based techniques. Alternative approaches to improve the exposure computation are also known. Deep learning, for example, is discussed in~\cite{AnderssonOosterlee}. A combination of deep learning, GPUs, and forward/backward SDEs can be found in~\cite{CrepeyBalanceSheet}. The approach proposed in this article provides a novel addition to the methods known in the literature.

Despite all these efforts to improve computation speed, fast evaluation of exposures for a portfolio is still open, especially when dealing with computationally intensive calculations of the related sensitivities. This article aims to take a step forward in reducing the computational effort for simulating exposures. The technique presented is also well-suited for parallelization and thus for GPU computation.

%=============== Something about the collocation and the background ===========
We will develop a highly accurate and fast numerical scheme. For this, we will employ the Stochastic Collocation (SC) method~\footnote{Generally, the term {\it collocation} denotes techniques that estimate deterministic or stochastic variables by finding a linear predictor from a finite set of observations}, developed in~\cite{grzelak2015stochastic}, as an efficient approach for approximating distribution functions. The distribution function of interest is then expanded as a polynomial in terms of a random variable that is cheap to sample from at given collocation points, and interpolation occurs between these points. Stochastic collocation points have a specific meaning, i.e., they represent critical features of the probability distribution of interest. The SC method enables us to generate samples from a complex distribution by mainly using interpolation efficiently.

Although in the 1D (one risk factor) case, the number of collocation points, $n_1$, is typically small~\footnote{Typically it varies from 3 to 6 to guarantee high-quality approximation (see~\cite{grzelak2015stochastic} for more details).}, an extensive system of, let us say, $d$ risk factors would require $n_1^d$ collocation points. Thus, it is subject to the so-called curse of dimensionality. Therefore, for large systems of SDEs, as often seen in the xVA context, this is not desirable.
However, this number can be reduced using the sparse grid approach, introduced by Smolyak in~\cite{Smol63}. The sparse grid algorithm in Smolyak's work is concentrated on multidimensional integration and high-dimensional interpolation. The significant advantage of the algorithm is that the number of grid points does not grow exponentially with the dimension, but only polynomially, meaning that the Smolyak method is not subject to the curse of dimensionality. The algorithm  has been associated with sparse grids, hyperbolic cross approximation, sparse tensor products, and various applications~\cite{Smol63_2}. In particular, variants of Smolyak's algorithm have been employed in the computation of high-dimensional integrals, in the numerical solution of PDEs and SDEs, and uncertainty quantification~\cite{UncertaintyQuantHeston}.
%======================================================================

It is also worth mentioning that, recently, the reduction of the number of option valuations using polynomial interpolation has been addressed with Chebyshev interpolation in~\cite{Gau2018}. The collocation method with Lagrange interpolations for arbitrage-free option pricing is discussed in~\cite{GrzelakArbitrage:2016}. Practical aspects of fast portfolio evaluation using interpolation for xVA are covered in~\cite{Ruiz2018}, and in~\cite{GlauPachonPotz2020} the application of Chebyshev interpolation for exposure calculation for the 1D case of Bermudan interest rate swaptions has been presented. Moreover, alternatives to the sparse grid approach for dealing with multi-dimensionality problems in the context of derivatives pricing exist. Recently, a low-rank tensor approximation gained increasing interest in applications to derivative pricing~\cite{Ruiz2021,Glau2020_LowRank}. The comparative study of the two methods in the context of uncertainty quantification is presented in~\cite{HighDimensionalUncertainty}. It was shown that the SC method on sparse grids appears to be computationally more efficient, at the cost of accuracy.

The present article is organized as follows: in Section~\ref{sec:2}, we introduce the Stochastic Collocation (SC) method and discuss exposure computation for portfolios depending on single and multi-factor models together with Smolyak's sparse grid algorithm. Section~\ref{sec:3} is a numerical section where the developed method is applied to realistic portfolios of interest rate swaps in multi-currencies. Section~\ref{sec:4} focuses on implementation details and improvements. Discussion of error analysis and convergence is covered in Section~\ref{sec:Error}. Concluding remarks are in Section~\ref{sec:conclusions}.

%%%%%%%%%%%%%%%%%%%%%%%%%%%%%%%%%%%%%%%%%%%%%%%%%%%%%%%%%%%%%%%%%%%%%%%
\section{Computations of Exposures with the SC Method}
\label{sec:2}
%%%%%%%%%%%%%%%%%%%%%%%%%%%%%%%%%%%%%%%%%%%%%%%%%%%%%%%%%%%%%%%%%%%%%%%
Mathematically, the positive and negative exposures, $E^+(t,{\bf X}(t))$, $E^-(t,{\bf X}(t))$, are defined as,
\begin{eqnarray}
\label{eqn:exposure}E^+(t,{\bf X}(t))&:=&\max(V(t,{\bf X}(t)),0),\;\;\;\;E^-(t,{\bf X}(t)):=\max(-V(t,{\bf X}(t)),0),
\end{eqnarray}
with
\begin{eqnarray}
\label{eqn:V}V(t,{\bf X}(t))&=&\E^\Q\left[\sum_{j=1}^{L}\frac{M(t)}{M(T_j)}H(T_j,{ {\bf X}}(T_j))\Big|\F(t)\right],\;\;T_j>t,
\end{eqnarray}
where $V(t,{{\bf X}}(t))$ represents the discounted value of a payoff $H(T_j,{{\bf X}}(T_j))$ at time $t$. The payments are taking place at $T_j$, $j=1,\dots,L$, with $T_j>t$, thus only outstanding payments are considered in the exposure computation. ${\bf X}(t)$ indicates a {\it risk factor} on which the derivative $H(\cdot)$ depends, and $M(t)$ stands for the money-savings account.

In practice, the exposures are computed per netting set, often involving hundreds of trades, thus with a netted portfolio involving $\bar{M}$ trades and $d$ different risk factors, the value $V(\cdot)$ is given by:
\begin{eqnarray}
V(t,{\bf X}(t)):=\sum_{i=1}^{\bar M}V_i(t,{\bf X}(t)),\;\;\;{\bf X}(t)=[X_1(t),\dots,X_d(t)]^\T.
\end{eqnarray}
Each of the risk factors in ${\bf X}(t)$ represents a stochastic process that influences the value of the portfolio. These can be interest rates in different currencies, stocks, inflation, foreign exchange, or commodities. The number of risk factors typically varies in time and depends on the portfolio composition.

Every component of xVA will depend on the exposure computation for any exposure date $T_i$. In a general setting, it can be represented as follows:
\begin{eqnarray}
\text{xVA}(t_0) &=& \int_{t_0}^T\E^\Q\left[\frac{M(t_0)}{M(t)}\chi(t,V(t,{\bf X}(t)))\Big|\F(t_0)\right]\dt\nonumber\\&\approx& \sum_{k=1}^{N_T}\E^\Q\left[\frac{M(t_0)}{M(T_k)}\chi(T_k,V(T_k,{\bf X}(T_k)))\Big|\F(t_0)\right]\Delta t,
\end{eqnarray}
with some generic function of exposures $\chi(t,x)$ and a discretization grid, $T_1,\dots,T_{N_T}$. In the case of CVA, for example~\footnote{In these illustrative examples we assume no Wrong-Way-Risk, however, the methodology would stay intact even when these assumptions were relaxed.}, $\chi(t,x)$ reads $\chi(t,x)=(1-R_c)E^+(t,x)f_D(t),$ with $E^+(t,x)$ defined in~(\ref{eqn:exposure}) and $f_D(t)$ being the default probability~\footnote{Default probability $f_D(t)$ depends on a particular ``x'' in xVA and may involve multiple counterparties. Moreover, its values will depend on the whole discretized interval $[t,t+\Delta t].$}, $R_c$ the recovery rate.

%In the case of FBA, for example, we have $\chi(t,x)=E^-(t)s_B(t)f_I(t)f_C(t)$ with $s_B(t)$ being the funding spread and $f_I(t)$, $f_C(t)$ are the probability of default of institution and a counterparty respectively.

In order to estimate the value of the portfolio $V(t,{\bf X}(t))$ with $N_p$ simulated Monte Carlo paths, the portfolio needs to be evaluated $N_p$-times at each exposure date $T_k$. For a portfolio involving multiple risk factors, the number of paths can be in the range of tens of thousands. From a computational perspective, a low number of portfolio evaluations is desired.

In this article, we propose a method that focuses on a drastic reduction of the number of portfolio evaluations. The proposed method relies on the approximation of the multi-dimensional portfolio $V(t,{\bf X}(t))$ by an approximating function, $\widetilde g(\{V\}_{i_1,\dots,i_d},{\bf X}(t)),$ where $\{V\}_{i_1,\dots,i_d}:=V(t,\{{\bf X}\}_{i_1,\dots,i_d})$. Function $\widetilde g(\cdot)$ is built based on only a few evaluations of the portfolio $V(t,\{{\bf X}\}_{i_1,\dots,i_d})$ on the set of so-called ``collocation points'', $\{{\bf X}\}_{i_1,\dots,i_d}$, obtained from the SC method (SC)~\cite{grzelak2015stochastic}. Intuitively, these collocation points can be understood as {\it optimal} quadrature points that describe the underlying random variable.  The idea behind the collocation method is, given the uncertain factors ${\bf X}(t)$, to determine the collocation points, $\{{\bf X}\}_{i_1,\dots,i_d}$, being the zeros of an orthogonal polynomial based on variable ${\bf X}(t).$

Once the approximating function $\widetilde g(\cdot)$ is determined then the computation of xVA is done as follows:
\begin{eqnarray}
\label{eqnxVA}
\text{xVA}(t_0)\approx \sum_{k=1}^{N_T}\E^\Q\left[\frac{M(t_0)}{M(T_k)}\chi(T_k,\widetilde g(\{V\}_{i_1,\dots,i_d},{\bf X}(T_k)))\Big|\F(t_0)\right]\Delta t,
\end{eqnarray}
with $\{V\}_{i_1,\dots,i_d}:=V(T_k,\{{\bf X}\}_{i_1,\dots,i_d})$. Function $\widetilde g(\cdot)$ does not require portfolio evaluations for every Monte Carlo path. The portfolio evaluation takes place only at the ``optimal'' points $\{{\bf X}\}_{i_1,\dots,i_d}$ that are determined based on the SC method. Once the approximating function $\widetilde g(\cdot,{\bf X}(t))$ is established, it is evaluated for all Monte Carlo paths. This computation, then, is extremely cheap, as function $\widetilde g(\cdot)$ would typically have a polynomial form constructed by the SC method.

%It is expected that these evaluations are very cheap and fast, it is due to the fact that the approximating function is constructed based on only a few portfolio evaluations.

%Once the function $g(\cdot)$ ``learns'' the portfolio profile for all the risk factors, it is evaluated for all the paths ${\bf X}(T_k).$ This computation however is extremely cheap as function $g(\cdot)$ would typically be in a form of a polynomial form.

In order to measure the quality of the approximations, we will consider discounted expected (positive) exposures defined for the exposure date, $T_k$, as:
\begin{equation}
\label{eqn:EE}
EE(t_0,T_k)=\E^\Q\left[\frac{M(t_0)}{M(T_k)}E^+(T_k,{\bf X}(T_k))\big|\F(t_0)\right],
\end{equation}
with $M(t)$ is the money-savings account where positive exposures $E^+(T_k,{\bf X}(T_k))$ are defined in~(\ref{eqn:exposure}) which, for the approximating function $\widetilde g(\cdot)$, reads:
\begin{eqnarray*}
E^+(T_k,{\bf X}(T_k)):=\max(V(T_k,{\bf X}(T_k)),0)&\approx& \max(\widetilde g(V(T_k,V(T_k,\{{\bf X}\}_{i_1,\dots,i_d}),{\bf X}(T_k)),0).
\end{eqnarray*} Another metric to measure the quality is the so-called {\it Potential Future Exposure} (PFE) which measures the maximum credit exposure calculated at some confidence level. The measure can be associated with measuring the quality of the approximating function $\widetilde g(\cdot)$ in the tails of the distribution of $V(t,{\bf X}(t))$, at $t=T_k$. PFE, at time $t$, i.e. $\text{PFE}(t_0,t)$, is defined as a quantile of the positive exposure $E^+(t,{\bf X}(t))$,
\begin{equation}
\label{eqn:PFE}
\text{PFE}(t_0,t) = \inf\{x\in\R:p\leq F_{E^+(t,{\bf X}(t))}(x)\},
\end{equation}
where $F_{E^+(t,{\bf X}(t))}(x)$ is the CDF of positive exposures observed at time $t$. Coefficient $p$ represents the  certainty level, i.e., the quantile level.

The SC method is considered an accurate simulation technique to approximate these quantities; while only a few portfolio evaluations are required to get sufficient accuracy. Typically, the number of points is low and varies from 3 to 6 per risk factor, independently of the number of trades in the underlying portfolio. Once the number of risk factors increases, it is no longer beneficial to apply the standard SC method~\footnote{The standard SC method is associated with a Cartesian grid built based on the collocation points constructed based on the corresponding basis for orthogonal polynomials.}, but we have to switch to a sparse-grid approach, where the number of portfolio evaluations will not suffer from the curse of dimensionality. In the follow-up section, we will provide details on the computational cost associated with exposure computation under the SC method.

%they do not require portfolio evaluations for each Monte Carlo path.

\begin{rem}[Machine Learning and Exposure]
One may consider building the approximating function $g(\cdot)$ as a process of supervised learning of the portfolio value $V(t,{\bf X}(t))$. In the standard machine learning methods, supervised learning involves expensive (portfolio) evaluations to determine the landscape of the objective function in the off-line stage. The method presented in this article may significantly improve the computations in such an expensive off-line stage. A similar idea of a so-called Compression-Decompression technique has been recently presented in~\cite{Grzelak:7L} in the context of Monte Carlo simulation using deep learning.
\end{rem}

\subsection{Optimal Points in Low and High Dimensions}
\label{sec:2_1}
Let us start with some background on the collocation method.
The proposed method was used to approximate an {\it expensive } random variable $Y$ utilizing a {\it cheap} variable
$X$. An approximation is made based on the inversion of the CDF of $Y$ at only a small set of collocation points, being the zeros of an orthogonal polynomial based on variable $X$.

Since any CDF is uniformly distributed, we have $F_Y(Y)\stackrel{\d}=F_X(X).$  From this representation, realizations of $Y$, $y_n$, and $X$,
$x_n$, are connected via the following inversion relation,
\begin{equation}
\label{eqn:CDF_inversion2} y_n=F_Y^{-1}(F_X(x_n))=:q_{Y}(F_X(x_n)),
\end{equation}
with $q_{Y}(p)=\min\{y \in\R:F_Y(y)\geq p\}$ indicating a quantile function.
The target is to determine an alternative relationship that does not require many of the ``expensive'' inversions $F_Y^{-1}(\cdot)$ for all samples of $X$.
The task is thus to find an approximation for function $g(\cdot)= F_Y^{-1}(F_X(\cdot)) $
such that
\[F_X(x)=F_Y(g(x)) \;\;\;\text{and}\;\;\; Y\stackrel{\d}=g(X),\]
where evaluations of function $g(\cdot)$ do not require the inversions $F_Y^{-1}(\cdot)$. With a mapping $g(\cdot)$
determined, the CDFs $F_X(x)$ and $F_Y(g(x))$ are not only equal in the distributional sense but also element-wise~\cite{grzelak2015stochastic}.

Sampling from random variable $Y$ can be decomposed into sampling from a
cheap random variable $X$ and a transformation to $Y$ via $g(\cdot)$, i.e., $y_n=g(x_n)$.
It is important to choose $g(\cdot)$ to be a basic function. To find a proper mapping function, we need to extract some information from $Y$.

An efficient method for sampling from variable $Y$ in terms of variable $X$ is obtained
by defining $g(\cdot)$ to be a {\em polynomial expansion}, i.e.
\begin{equation}
\label{eqn:lagrange_1} y_n\approx
\widetilde{g}(x_n)=\sum_{i=1}^{n_1}y_i\psi_i(x_n),\;\;\;\psi_i(x_n)=\prod_{j=1,\\i\neq j}^{n_1}\frac{x_n- x_j}{ x_i- x_j},
\end{equation}
where $x_n$ is a sample from $X$ and $ x_i$ are so-called {\it
collocation points}, $y_i$ is the exact evaluation at collocation
point $x_i$ in~(\ref{eqn:CDF_inversion2}), i.e.,
$y_i = F_Y^{-1}(F_X( x_i))$ in~(\ref{eqn:lagrange_1}) and $\psi_i(x_n)$ represent unidimensional basis functions. As mentioned earlier, $y_i$, is simply the $F_X(x_i)-$quantile thus, in essence, Equation~(\ref{eqn:lagrange_1}) represents an approximation of a random variable $Y$ through interpolation of its quantiles.
There are multiple ways of choosing the collocation points $x_i$. In the standard SC~\cite{grzelak2015stochastic} method the ``optimal'' points are the quadrature points of random variable $X$. One can also consider these points to be determined by equally spaced quantiles, $p_1,\dots,p_{n_1}$ with $x_i:= q_{X}(p_i)=\min\{x \in\R:F_X(x)\geq p_i\}.$

Knowing the principles behind the SC method, we need to adapt it to our problem at hand; namely, the approximation of the portfolio given a particular stochastic driver $X$. A direct application of the SC method, in Equation~(\ref{eqn:CDF_inversion2}), would require the availability of the portfolio distribution. Unfortunately, such a distribution is not available and can only be determined based on simulation- thus, precisely what we wish to avoid. Therefore, the analogy between cheap and expensive variables in the SC method is not applicable in reducing portfolio evaluations. On the other hand, we know that the value of a portfolio $V(t,X)$ is a function of the stochastic variable $X$; therefore, we can utilize the SC method to interpolate the value of a portfolio $V(t,x_i)$, $i=1,\dots,n_1$, where $x_i$'s are the collocation points, determined based on stochastic quantity, $X(t)$.\\
%Knowing the principles behind the SC method, we need to adapt it to our particular problem, namely, the approximation of the portfolio given a particular stochastic driver $X$. A direct application of the SC method, in Equation~(\ref{eqn:CDF_inversion2}), would require the availability of the portfolio distribution. Unfortunately, such a distribution is not available and can only be determined based on simulation- thus, precisely what we wish to avoid. Therefore, the analogy between cheap and expensive variables in the SC method is not applicable in reducing portfolio evaluations. On the other hand, we know that the value of a portfolio $V(t,X)$ is a function of stochastic variable $X$; therefore, we can utilize the SC method to interpolate the value of a portfolio $V(t,x_i)$, $i=1,\dots,n_1$, where $x_i$'s are the collocation points, determined based on stochastic quantity, $X(t)$.
The SC method can be extended to multiple inputs using a tensor product of quadrature points, i.e. for $d$-dimensions we have:
\begin{equation}
\label{eqn:g_N_d}
%\mu_{u_i}(\mathbf{x}, t) = \sum_{j_1=1}^N\cdots\sum_{j_{n_\xi}=1}^N u_{j_1,\ldots,j_{n_\xi}}(\mathbf{x},t)^iz_{j_1}\cdots z_{j_{n_\xi}}.
\widetilde g(x_{1},\dots, x_{d}) =\sum_{j_1=1}^{n_1}\cdots\sum_{j_{d}=1}^{n_d} y_{j_1,\dots,j_d}\psi_{j_1,\dots,j_d}(x_{1},\dots,x_{d}),
\end{equation}
where $x_{1},\dots,x_{d}$ are the collocation for each dimension, $1,\dots,d$. Although the representation~(\ref{eqn:g_N_d}) constitutes a natural extension of the 1D case,  once the number of uncertain parameters increases, approximations based on these tensor product grids become inefficient because the number of collocation points in a tensor grid grows exponentially in dimension, i.e., $n_1{\cdot}\ldots \cdot n_d$. In such a scenario, one may consider sparse tensor product spaces as first proposed by Smolyak~\cite{Smol63}. More precisely, the Smolyak sparse grid
SC method for approximating statistical quantities is used here to reduce the exponential increase of the number of tensor product quadrature points in $d$ dimensions. A linear combination of tensor product operators $p^{|{\bf w}|}$, with  $|{\bf w}|=w_1+\dots+w_d$, see~\cite{Xiu05,JUDD2014}, is given by:
\begin{eqnarray}
\label{eqn:SmolyakGrid}
 \widetilde g(x_{1},\dots, x_{d})= \sum_{\max(d,\mu+1)\leq\mid {\bf w}\mid\leq d+\mu}(-1)^{d+\mu-|\mathbf{w}|}{{d-1}\choose{d+\mu-|\mathbf{w}|}}p^{|\bf w|}(x_{1},\dots,x_{d}),
\end{eqnarray}
with $\mu$ the sparse grid level often referred to as the {\it deepness} or {\it density} parameter, and where the tensor product operator is defined by:
\begin{equation}
p^{|{\bf w}|}(x_{1},\dots,x_{d})=\sum_{|{\bf w}|=w_1+\dots+w_d}p^{w_1,w_2,\dots,w_d}(x_{1},\dots,x_{d}),
\end{equation}
with
\begin{eqnarray}
p^{w_1,w_2,\dots,w_d}(x_{1},\dots,x_{d})=\sum_{j_1=1}^{m(w_1)}\dots\sum_{j_d=1}^{m(w_d)}b_{j_1,\dots,j_d}\psi_{j_1}(x_1)\cdots\psi_{j_d}(x_d),
\end{eqnarray}
where $m(j_k)=2^{j_k}+1$ for $j_k\geq 2$ and $m(1)\equiv 1$ representing the number of basis functions in dimension $j_k$. Functions $\psi_{j_k}(x_k)$ are the $j_k$'th unidimensional basis functions and  $b_{j_1,\dots,j_d}$ corresponds to polynomial coefficients, where $x_j\in[-1,1]^d$ are the collocation points for dimension $j$. These coefficients are constructed in such a way that the approximating polynomial $\widetilde g(\cdot)$ exactly matches the original function $g(\cdot)$ in all the points of the constructed grid.

As noted in~\cite{JUDD2014}, the classical Smolyak representation in~(\ref{eqn:SmolyakGrid}) involves nested sets, i.e., many elements of the summation contain repeated elements that increase in dimension $d$ and level $\mu$. To further improve the computational efficiency, an alternative formulation, based on disjoint sets, has been proposed~\cite{JUDD2014}:
\begin{eqnarray}
\label{eqn:SmolyakGrid2}
 \widetilde g(x_{1},\dots, x_{d})= \sum_{d\leq\mid {\bf w}\mid\leq d+\mu}q^{|\bf w|}(x_{1},\dots,x_{d}),
\end{eqnarray}
with a tensor product operator:
\begin{equation}
q^{|{\bf w}|}(x_{1},\dots,x_{d})=\sum_{|{\bf w}|=w_1+\dots+w_d}q^{w_1,w_2,\dots,w_d}(x_{1},\dots,x_{d}),
\end{equation}
and where
\begin{eqnarray}
q^{w_1,w_2,\dots,w_d}(x_{1},\dots,x_{d})=\sum_{j_1=m(w_1-1)+1}^{m(w_1)}\dots\sum_{j_d=m(w_d-1)+1}^{m(w_d)}b_{j_1,\dots,j_d}\psi_{j_1}(x_1)\cdots\psi_{j_d}(x_d),
\end{eqnarray}
with $\psi_{j_1}(x_1)\cdots\psi_{j_d}(x_d)$ representing a product of unidimensional basis functions and other settings are as for representation in~(\ref{eqn:SmolyakGrid}). The benefit of the presentation above is that there are no repeated terms across the different functions $p^{w_1,w_2,\dots,w_d}(x_{1},\dots,x_{d})$. Coefficients $b_{j_1,\dots,j_d}$ must be constructed so that the approximating polynomial $\widetilde{g}(\cdot)$ exactly matches the true function $g(\cdot)$ at the collocation points. As discussed in~\cite{JUDD2014}, there are essentially two  choices for the computation of the coefficients $b_{j_1,\dots,j_d}$, analytically involving nested sums or numerically by a numerical solution of a Lagrange interpolation problem given by solving the following algebraic system of equations:
\begin{equation}
\left[\begin{array}{c}
g(x_1)\\g(x_2)\\\dots\\g(x_d)\end{array}\right]=\left[\begin{array}{c}
\widetilde g(x_1)\\\widetilde g(x_2)\\\dots\\\widetilde g(x_d)\end{array}\right]=\left[\begin{array}{ccc}
\psi_1(x_1)&\dots&\psi_d(x_1)\\
\psi_1(x_2)&\dots&\psi_d(x_d)\\
\dots&\ddots&\dots\\\
\psi_1(x_d)&\dots&\psi_d(x_d)
\end{array}\right]\cdot\left[\begin{array}{c}
b_1\\b_2\\\dots\\b_d\end{array}\right],
\end{equation}
where $\psi_k:[-1,1]^d\rightarrow \R$ for $k=1,\dots,d$ are $d-$dimensional basis functions. The representation above is very flexible as it allows for different choices of collocation points $x_1,\dots,x_d$ and basis functions $\psi_k.$
%As we see the sparse grid construction in~(\ref{eqn:SmolyakGrid}) is not trivial in implementation, however open-source libraries exist~[???].
The sparse grid construction is based on a unidimensional basis function.
A popular choice for the construction of these basis functions is found in the Chebyshev polynomials~\cite{Judd1998} and their extrema~\footnote{The article relies on the Smolyak grid, where the computed interpolation coefficients are used on a universal Lagrange interpolation technique. We use families of orthogonal basis functions to ensure the numerical stability of a solution to the Lagrange inverse problem. Although we use the Chebyshev polynomials in the article, it is not a limitation.}, also known as the Clenshaw-Curtis points~\cite{ClenshawCurtis1960},  defined by a recursive relation:
\begin{eqnarray*}
T_n(x)=2xT_{n-1}(x)-T_{n-2}(x),\;\;\;n\geq2,
\end{eqnarray*}
where $T_0(x)=1$, and $T_1(x)=x$. The Chebyshev polynomial of degree $N-1$ has $N$ extrema given by: $x_i=-\cos(\pi(i-1)/N-1)$ for $i=1,\dots,N.$ The approximation error of the Chebyshev is known to be polynomial/exponentially
convergent under the supreme norm for Lipschitz/analytic functions\cite{HighDimensionalUncertainty}.

The sparse tensor product grids for high dimensions are built upon Clenshaw-Curtis abscissas as they are particularly efficient (since the resulting sparse grids are nested). This hierarchical sampling property allows for reusing the samples when increasing the order when a more accurate model's response is required. In the original literature, we can find works that report the higher accuracy of Clenshaw-Curtis points than of the corresponding number of Gauss quadrature points~\cite{Tref08}.

An illustration of the collocation points for a 2D and 3D case is presented in Figures~\ref{fig:2d} and~\ref{fig:3d}. As we can see, the distribution of the points strongly depends on the level $\mu$. Under the sparse grid method, the number of grid points does not grow exponentially as it does for the tensor product but grows polynomially with dimension $d$.
In a later section, we will investigate its impact on the quality of the results~\footnote{Although the formulation may seem somewhat involved, open-source computed code libraries exist. An open-source library for efficient Smolyak grid construction can be found, for example, at \url{https://github.com/EconForge/Smolyak}.}.

\begin{figure}[h!]
  \centering
    \includegraphics[width=0.325\textwidth]{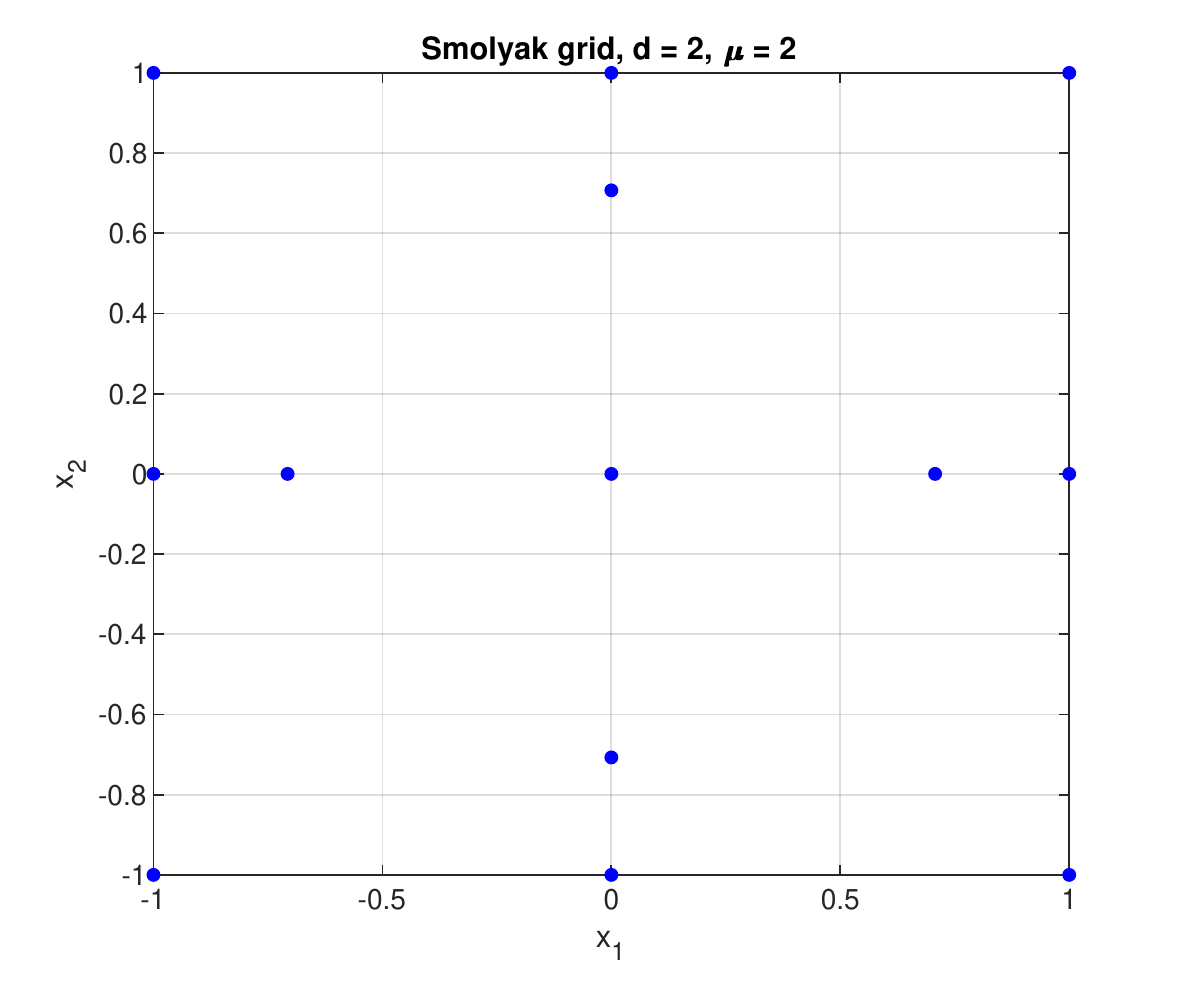}
    \includegraphics[width=0.325\textwidth]{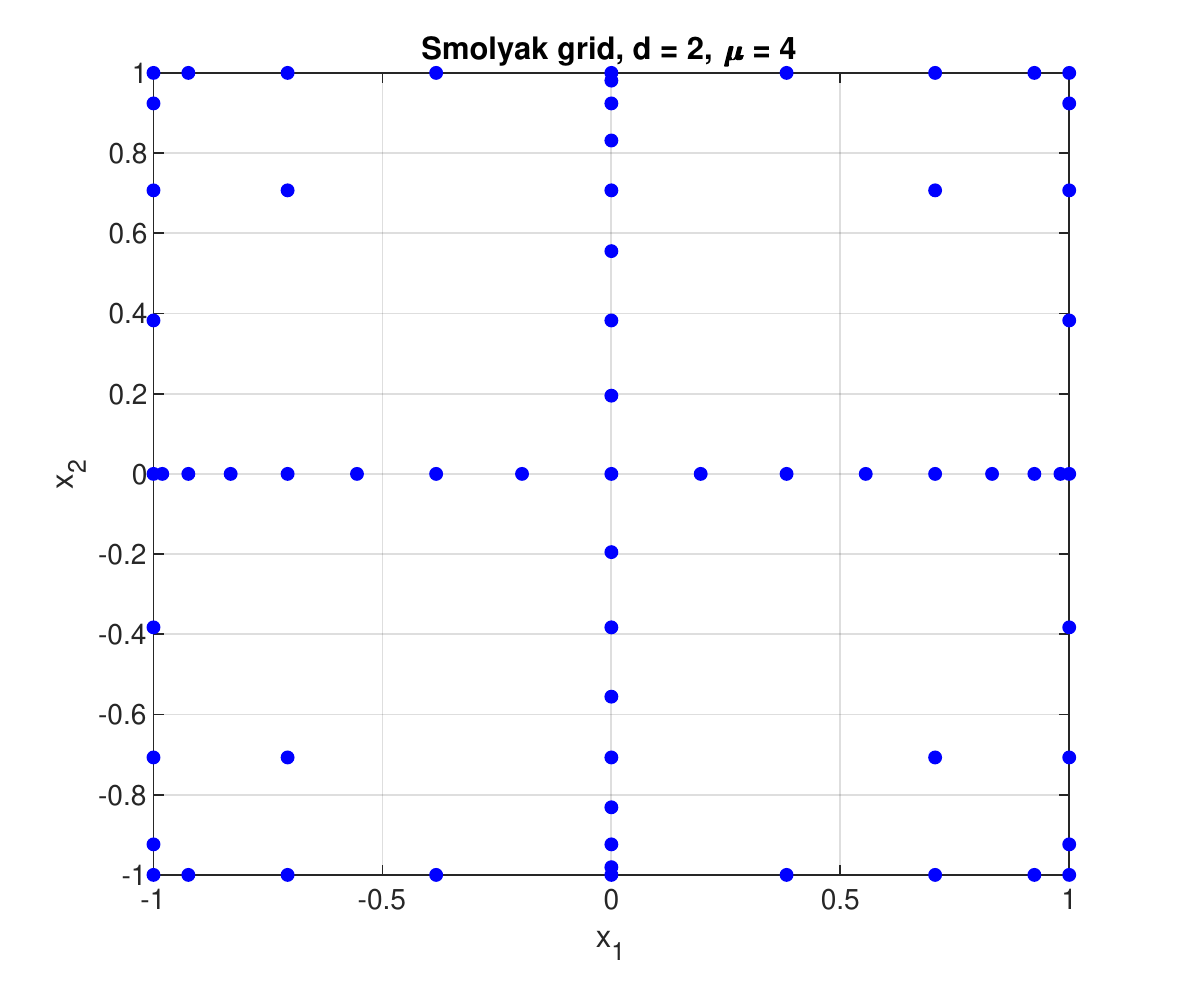}
    \includegraphics[width=0.325\textwidth]{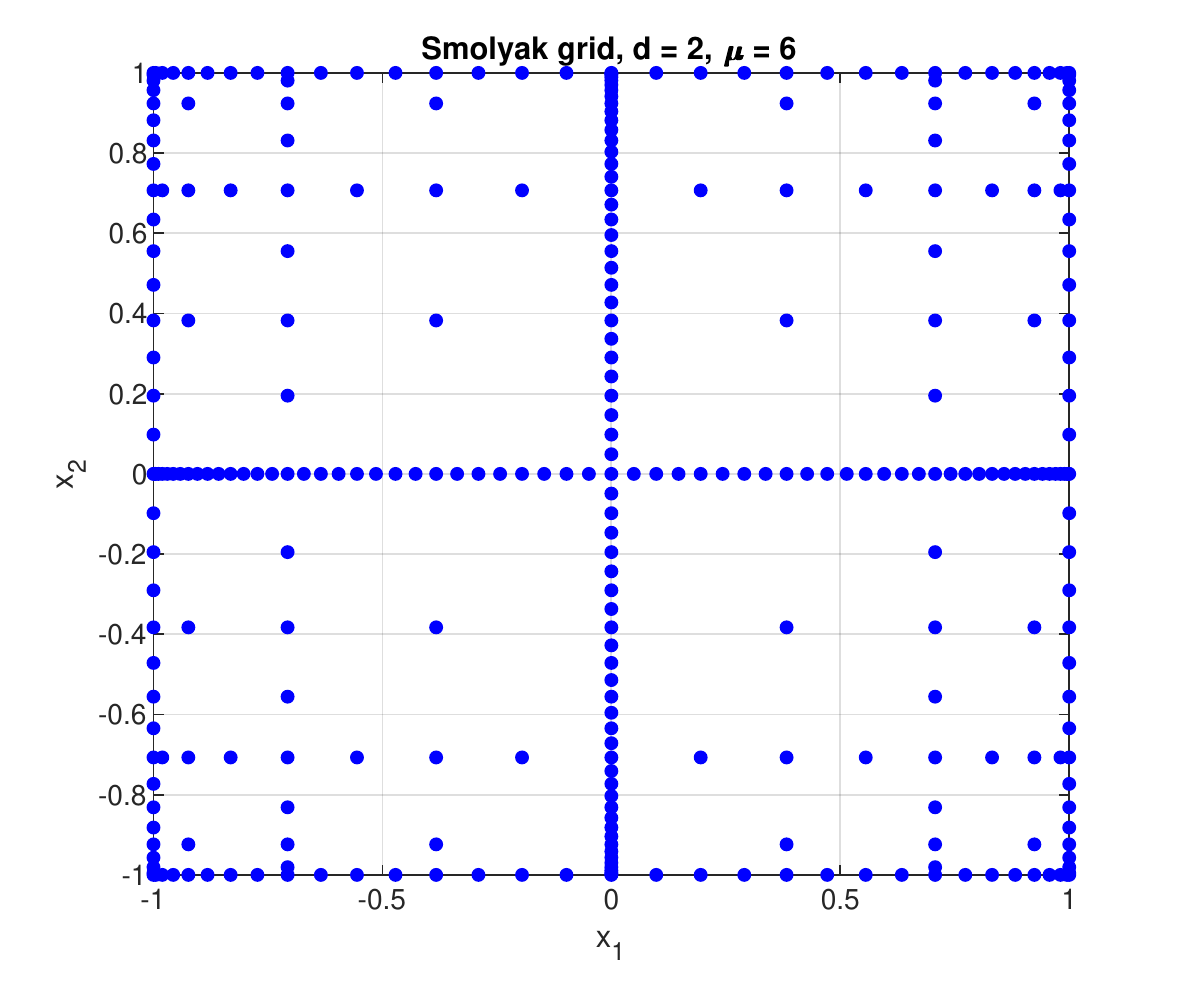}
      \caption{Collocation points for the $d=2$ case with varying level parameter $\mu=2,4,6$}
      \label{fig:2d}
\end{figure}

\begin{figure}[h!]
  \centering
    \includegraphics[width=0.325\textwidth]{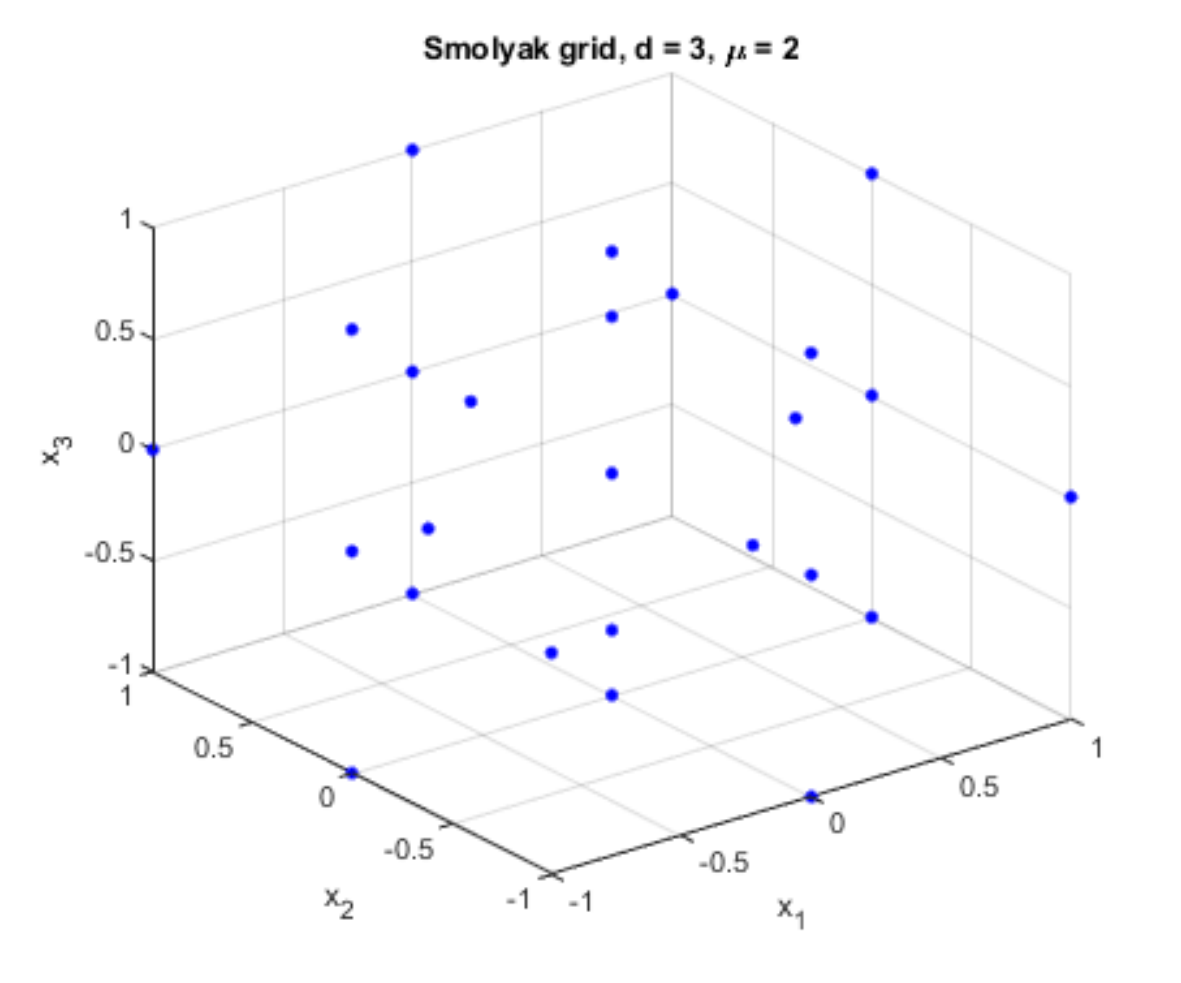}
    \includegraphics[width=0.325\textwidth]{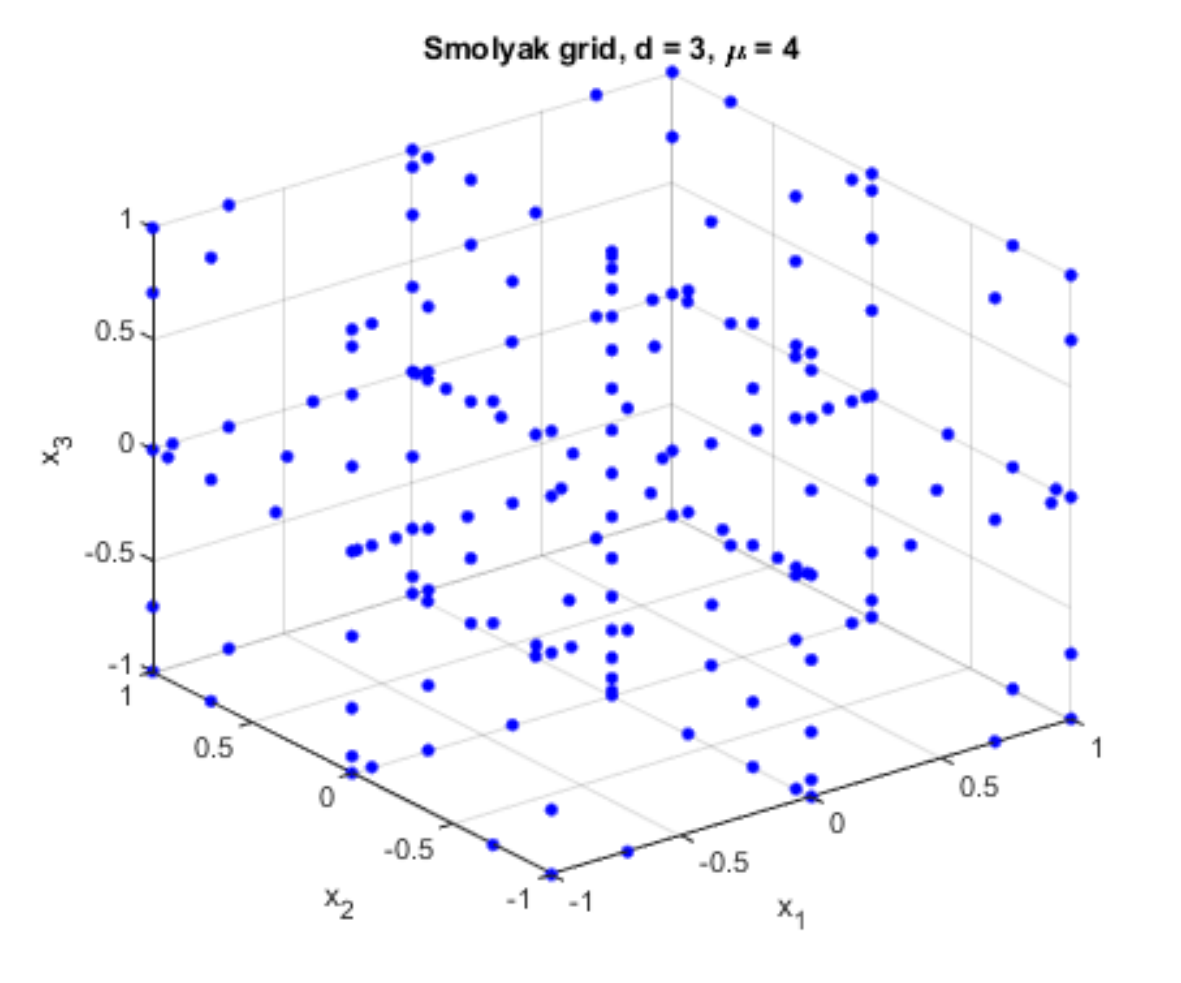}
    \includegraphics[width=0.325\textwidth]{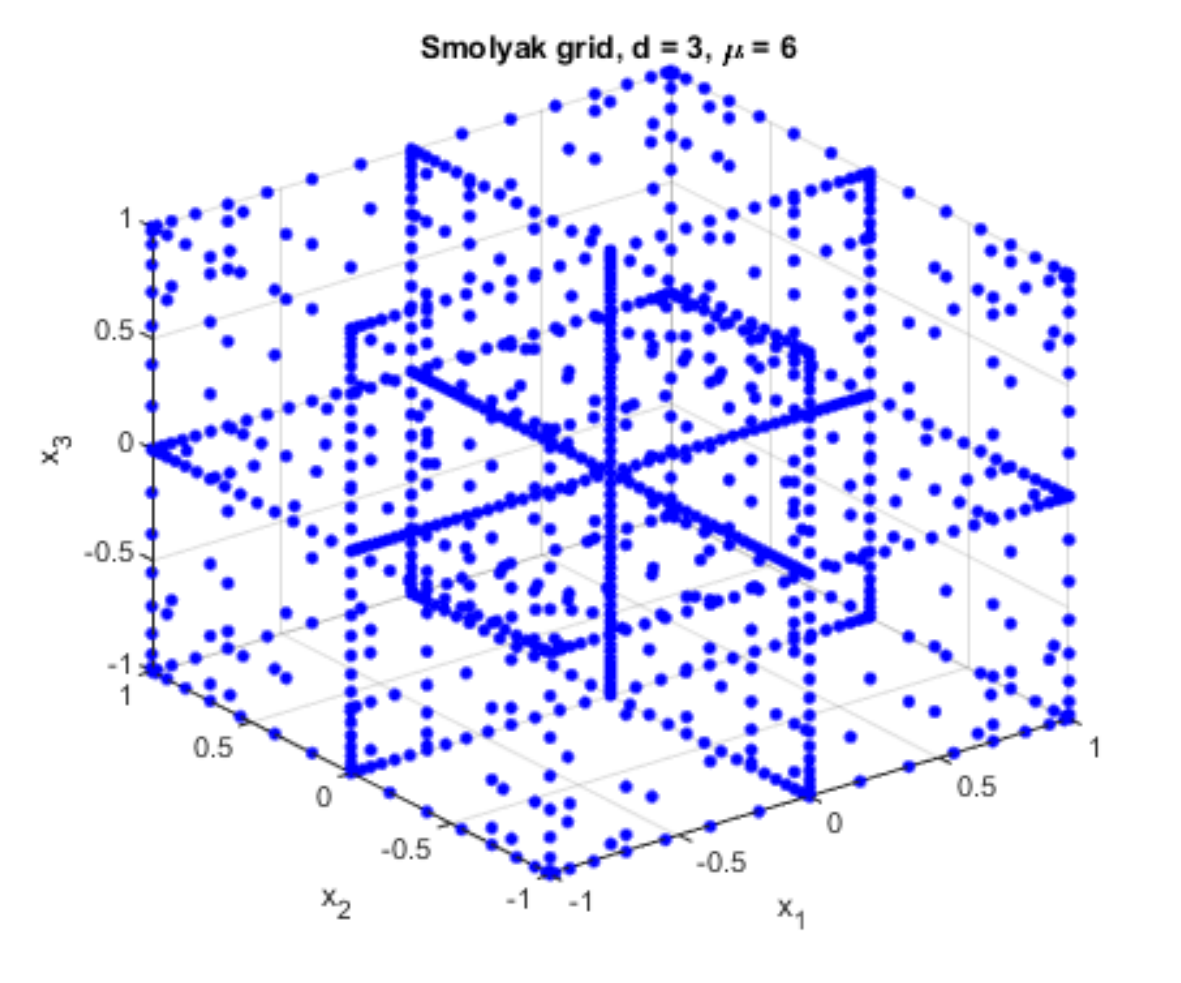}
      \caption{Collocation points for the $d=3$ case with varying level parameter $\mu=2,4,6.$}
      \label{fig:3d}
\end{figure}

\begin{rem}[High and Low Dimension Methods]
The main difference between the low- and high-dimensional methods is that in the low-dimensional case, the collocation points may be chosen based on the underlying risk factors, i.e., the SC method determines the optimal quadrature points based on the distribution, while in a high-dimensional case the collocation points are solely determined based on the sparse grid algorithm. Once the grid points in the $d-$dimensional space $[-1,1]^d$ are determined, the grid is appropriately scaled to match the domain of the risk factors ${\bf X}(t).$ Details on the scaling will be given in Section~\ref{sec:domainScaling}.
\end{rem}

In Table~\ref{Tab:diemsionality}, an overview of the required number of portfolio evaluations depending on dimensionality, $d$, and grid level, $\mu$, is presented. The results are very encouraging. For example, in the 2D case, we only need between 5 (for $\mu=1$) and 29 (for $\mu=3$) portfolio evaluations while in, for example, the 5D case, we only require 11 or 241 evaluations for $\mu=1$ and $\mu=3,$ respectively. Considering that the ``industrial'' practice for xVA is to run between 4,000 to 25,000 scenarios, the method makes a significant improvement, especially for a portfolio that does not depend on a very  high number of risk factors. It is desirable to keep the level parameter, $\mu$, as low as possible from the computational perspective. As it will be presented, the SC model with Smolyak's sparse grid performs accurately in estimating EEs and PFEs, already for $\mu = 2$ or $\mu = 3$.

\begin{table}[htb!]
\centering\footnotesize
\caption{\footnotesize The required number of portfolio evaluations per exposure date $T_k$, depending on the approximating algorithm, model configuration, tensor product and sparse grid: comparison of grid points as a function of dimension. $n_i$ indicates the number of collocation points in the $i$'th dimension and $d$ the number of risk factors that need to be simulated for exposure computation. $\mu$ represents the deepness of the grid.}
\begin{tabular}{c|c|c|c||c|c|c|c}
%&&&$n=3$\\
$d$&tensor product&tensor product&tensor product& \multicolumn{4}{c}{Smolyak grid}\\
&for $n_i=3$&for $n_i=4$&$n_i=5$&$\mu = 1$&$\mu=2$&$\mu=3$&$\mu=4$\\\hline\hline
1&3&4&5&3&5&9&-\\
2&9&16&25&5&13&29&65\\
3&27&64&125&7&25&69&177\\
4&81&256&625&9&41&137&401\\
5&243&1,024&3,125&11&61&241&801\\
6&729&4,096&15,625&13&85&389&1,457\\
7&2,187&16,384&78,125&15&113&589&2,465\\
8&6,561&65,536&390,625&17&145&849&3,937\\
\end{tabular}
\label{Tab:diemsionality}
\end{table}

As a rule of thumb, we distinguish two variants for determining the collocation points: with low dimensionality, $d\leq 2$, the collocation points will be established based on the simulated risk factor ${\bf X}(t)$ (standard collocation method). In the high-dimensional case, $d>2$, these points are established based on the sparse grid method of Smolyak.

\begin{rem}
Table~\ref{Tab:diemsionality} illustrates the gain in the portfolio evaluations; however, it is also worth mentioning that the evaluation of the approximating function $\widetilde{g}$, although swift, is not cost-free, i.e., it requires the evaluation of the interpolating function for all the simulated paths. This, however, is independent of the size of the portfolio considered.
\end{rem}

%Efficient implementation of the Smolyak methodology is wpsi-described in~\cite{JUDD2014}. Moreover, open-source %Python libraries for efficient computation are also available: ???

%%%%%%%%%%%%%%%%%%%%%%%%%%%%%%%%%%%%%%%%%%%%%%%%%%%%%%%%%%%%%%%%%%%%%%%%%%%%%%%%%%%%
%\subsection{Evaluation Algorithm}
\subsection{Construction of the Approximating Portfolio}
\label{sec:generalConstruction}
%%%%%%%%%%%%%%%%%%%%%%%%%%%%%%%%%%%%%%%%%%%%%%%%%%%%%%%%%%%%%%%%%%%%%%%%%%%%%%%%%%%%
Given the method described in the previous section, we specify the approximating portfolio construction. Every exposure computation relies on stochastic risk factors that need to be simulated, typically using a big time-step Monte Carlo simulation. In the multi-dimensional case, the risk factor ${\bf X}(t)$ is described by the SDEs:
\begin{eqnarray}
\label{SDE}
\d {\bf X}(t)=\alpha(t,{\bf X}(t))\dt + \sigma(t,{\bf X}(t))\d{\bf W}(t),
\end{eqnarray}
with $\alpha(t,{\bf X}(t))$ being a drift, $\sigma(t,{\bf X}(t))$ is the volatility and ${\bf W}(t)$ represents multi-dimensional correlated Brownian motion.  From the computational perspective, this step of the exposure simulation is just a fraction of the total portfolio evaluation. Discussion on efficient simulation of multidimensional SDEs can be found in~\cite{OosterleeGrzelakBook}, for example.

Once we have simulated the risk factors, the next step is to determine the approximation of the portfolio, i.e.: the value of a portfolio, at any exposure date $T_k$, which is approximated as:
\begin{eqnarray}
\label{eqn:proxySC}
V(T_k,{\bf X}(T_k))&\approx& \widetilde g\big(\{V\}_{i_1,\dots,i_d},{\bf X}(T_k)\big),
\end{eqnarray}
where $V(T_k,{\bf X}(T_k))$ stands for the original portfolio and  $\widetilde g\big(\{V\}_{i_1,\dots,i_d},{\bf X}(T_k)\big)$ is the approximation of the portfolio based on the portfolio values evaluated only at the collocation points $\{V\}_{i_1,\dots,i_d}:=V(t,\{{\bf X}\}_{i_1,\dots,i_d})$ with $\{{\bf X}\}_{i_1,\dots,i_d}=\{x_{1,1},x_{1,2},\dots x_{2,1},x_{2,2},\dots\}$, and where $x_{i,j}$ is the $j$'th collocation point for the $i$'th risk factor.

Note that the approach presented above slightly differs from the original SCMC method in~\cite{grzelak2015stochastic}, i.e., a direct application of the SCMC model would require the construction of the approximating polynomial based on quantiles of an {\it expensive} random variable, which in the current setting would be represented by the portfolio priced at the exposure date, $T_k$. Such quantiles are not available as they would require evaluation of the portfolio for all Monte Carlo paths- which we wish to avoid. Instead, we build an approximating grid based solely on the underlying risk factors.

Finally, once the approximating function $\widetilde g\big(\{V\}_{i_1,\dots,i_d},{\bf X}(T_k)\big)$ is built, we evaluate it for all  Monte Carlo paths defined in~(\ref{SDE}). This function evaluation does not require pricing of a portfolio, so the computational cost is low.

%Details regarding the computational costs for evaluation of function $\widetilde g(\cdot)$ will be discussed later in this article.
%The procedure to determine the optimal SC points varies depend on the dimensionality and is in detail described in Section~\ref{sec:2_1}.

In Figure~\ref{fig:pathsValue}, an illustrative example is presented. In the experiment, we consider a portfolio comprising  20 interest rate swaps and a stock. Thus, a case  with two risk factors, ${\bf X}(T_k)=[r(T_k),S(T_k)]^\T$ observed at a time $T_k$. In the LHS figure marked with black dots, we see the set of collocation points, $\{V\}_{i_1,i_2}$, that is determined based on the Smolyak sparse grid. Based on these grid points, an interpolated surface of the portfolio values $\widetilde g(\{V\}_{i_1,i_2})$ is built with the Lagrange interpolation. The RHS figure illustrates the evaluation of the approximating portfolio using the simulated risk factors, $\widetilde g\big(\{V\}_{i_1,i_2},{\bf X}(T_k)\big)$. The procedure described above needs to be repeated for all exposure dates: $T_k$, $k=1,\dots,N_T,$ in Equation~(\ref{eqn:proxySC}). Once function $\widetilde g(\cdot)$ is determined for every $T_k$ and evaluated for all Monte Carlo scenarios, it can be used to calculate risk measures as exposure, potential future exposure, etc.

%The objective now is to determine the approximating function $g(\cdot)$ in a most efficient manner. Details on determining this function will be given in the follow-up section, Section~\ref{sec:2_1}. Where [???] we will discuss different aspects of the function $g(\cdot)$ and...

%Now, assuming that function $g(\cdot)$ is determined, the exposure algorithm can be summarized as follows:
%\begin{itemize}
%\vspace{-0.2cm}\item[step 1:] We start with generation of Monte Carlo paths for the underlying stochastic differential equations (SDEs), ${\bf X}(t)$. \vspace{0.1mm}
%\vspace{-0.2cm}\item[step 2:] Given the paths compute optimal points $x_i$ for all the underlying dimensions
%\vspace{-0.2cm}\item[step 3:] Construct n-dimensional surface $\widetilde{V}(t)$ which is an approximation of $V(t)$- we construct n-dimensional polynomial given the collocation points.
%\vspace{-0.2cm}\item[step 4:] Once the payoff function $\widetilde{V}(t)$ is estimated, we use all the Monte Carlo paths and determine stochastic random variable $\widetilde{V}(t).$ Given this approximation we are able to compute expected exposures.
%\end{itemize}

\begin{figure}[h!]
  \centering
    \includegraphics[width=0.49\textwidth]{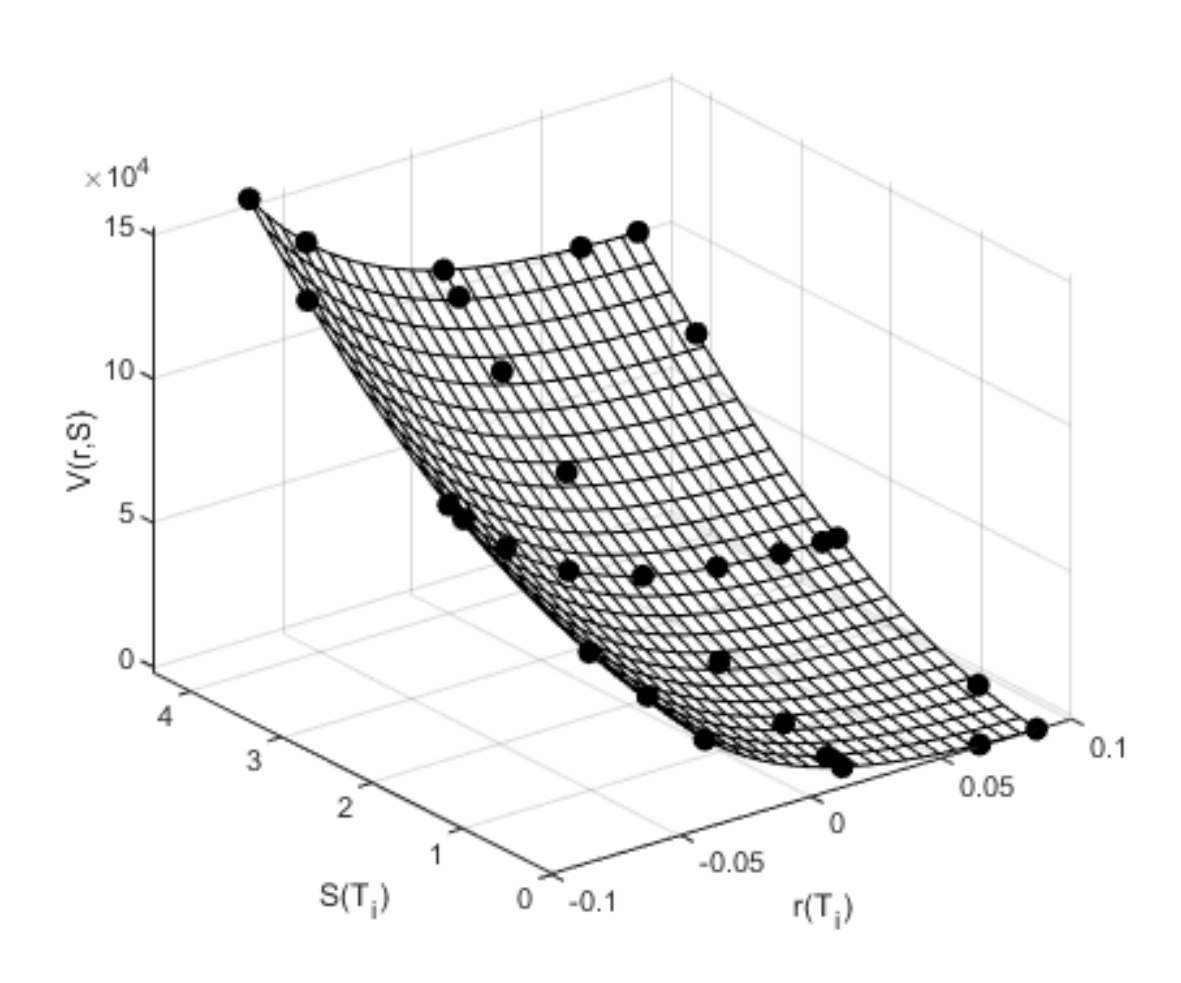}
    \includegraphics[width=0.49\textwidth]{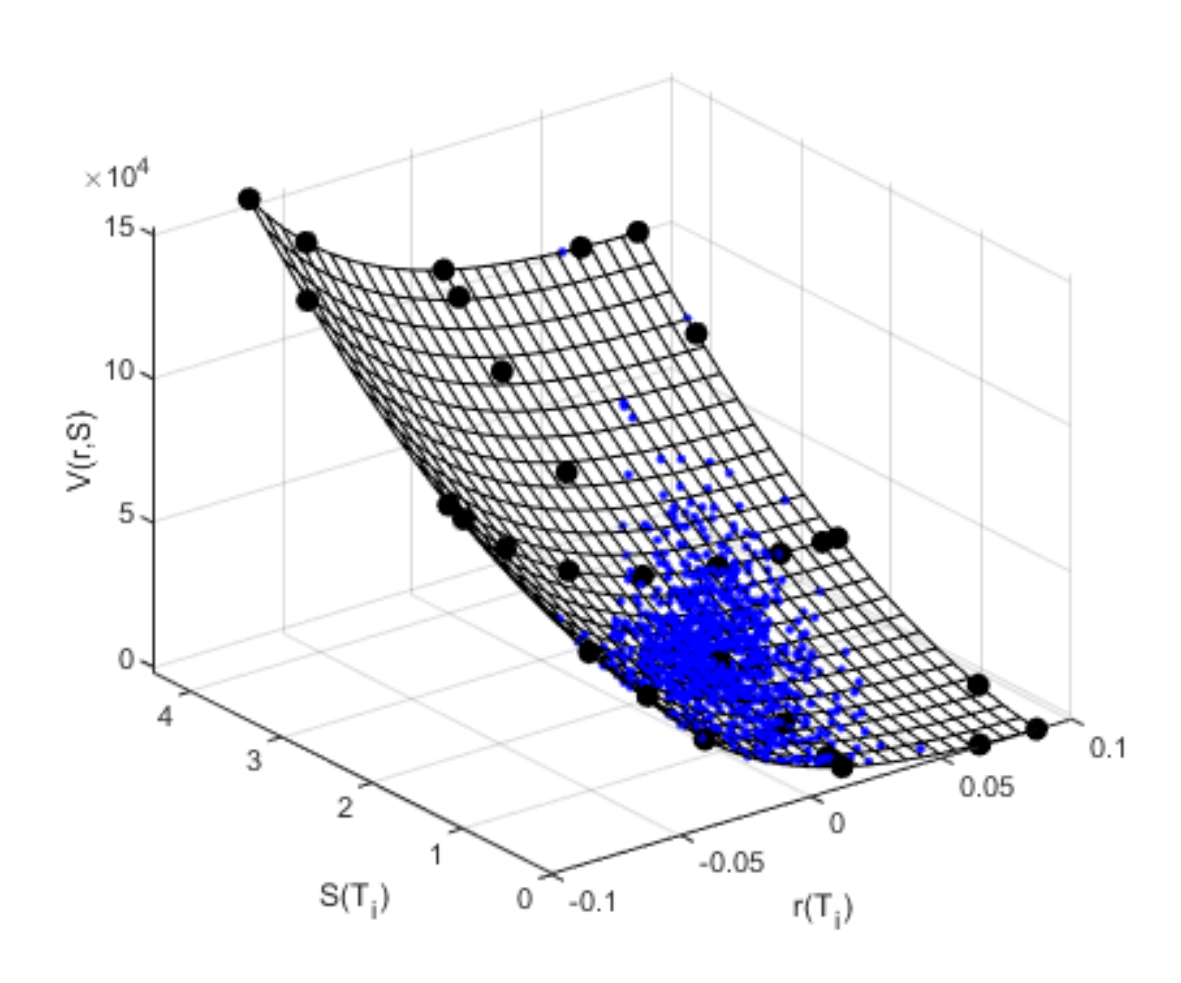}
      \caption{The construction of the sparse grid for portfolio evaluation. LHS: black dots indicate the collocation values and the grid constructed using Lagrange interpolation. RHS: blue dots illustrate the portfolio values based on the estimated grid for simulated risk factors ${\bf X}=[r(t),S(t)]$. }
      \label{fig:pathsValue}
\end{figure}

\begin{rem}[Correlation]
Since the proposed algorithm focuses on approximating the underlying portfolio, the method does not affect the structure of the underlying stochastic processes, i.e., Monte Carlo paths stay the same; thus, the model correlation structure stays intact.
\end{rem}

A detailed discussion on the error of the approximation, convergence, and the computational cost will be discussed in Section~\ref{sec:Error}, while the complete simulation algorithm is described in~\ref{sec:appendix}.

%\begin{figure}[h!]
%  \centering
%\includegraphics[width=0.48\textwidth]{Figures/Smolyak_grid_1.pdf}
% \includegraphics[width=0.48\textwidth]{Figures/Smolyak_grid_2.pdf}
%      \caption{}
%      \label{fig:pathsValue2}
%\end{figure}

%%%%%%%%%%%%%%%%%%%%%%%%%%%%%%%%%%%%%%%%%%%%%%%%%%%%%%%%%%%%%%%%%%%%%%%%%%%%%
%\section{Application of the Collocation Method in Exposure Computation}
\section{Practical Application and Numerical Experiments}
\label{sec:3}
%%%%%%%%%%%%%%%%%%%%%%%%%%%%%%%%%%%%%%%%%%%%%%%%%%%%%%%%%%%%%%%%%%%%%%%%%%%%%
%This section provides a detailed description of using the SC method for approximating portfolios and therefore getting exposure profiles efficiently. We start with a single currency portfolio comprising multiple interest rate swaps, thus a 1D case with a single risk factor. Later we will discuss portfolios depending on multiple risk factors, thus containing derivatives like interest rate swaps in different currencies. \textcolor{blue}{This section mainly focuses on products with no optionality, commonly known as linear products.~\footnote{In the numerical experiment, the pricing has been performed using Python 3.9 on Intel(R) Core (TM) i5-6800K CPU @ 3.6Ghz. }}

This section provides a detailed description of using the SC method to approximate portfolios and, therefore, get exposure profiles efficiently. We start with a single currency portfolio comprising multiple interest rate swaps, thus a 1D case with a single risk factor. Later we will discuss portfolios depending on multiple risk factors, i.e., containing derivatives like interest rate swaps in different currencies. We also focus on multi-factor models, like the Hull-White 2 factor model, and we assess the performance when valuing portfolios containing swaptions.~\footnote{In the numerical experiment, the pricing has been performed using Python 3.9 on Intel(R) Core (TM) i5-6800K CPU @ 3.6Ghz. }

%%%%%%%%%%%%%%%%%%%%%%%%%%%%%%%%%%%%%%%%%%%%%%%%%%%%%%%%%%%%%%%%%%%%%%%%%%%%%
%\subsection{Portfolio with a Single IR Risk Factor}
\subsection{Portfolio with Single Currency Interest Rate Products, 1D Case}
\label{sec:1D}
%%%%%%%%%%%%%%%%%%%%%%%%%%%%%%%%%%%%%%%%%%%%%%%%%%%%%%%%%%%%%%%%%%%%%%%%%%%%%
In this subsection, we analyze a portfolio that solely depends on a single risk factor; it can be, for example, a portfolio consisting only of EUR swaps~\footnote{Swaps are financial products that enable their holders to swap two sets of interest rate payments.}. As the industry-standard, we consider the 1D Hull-White (HW) model as the driver of the short rate process $r(t),$ with the dynamics given by:
\begin{eqnarray}
\label{eqn:HW}
\d r(t) = \lambda (\theta(t)-r(t))\dt + \eta \dW^\Q(t),
\end{eqnarray}
and parameters $\lambda$ and $\eta$, term structure function $\theta(t)$ and a Brownian motion under the risk-neutral measure $W^\Q(t).$ Function $\theta(t)$ is defined in terms of the zero-coupon bonds, $P_{mrkt}(0,t)$, quoted in the market and is given by:
\begin{eqnarray}
\label{eqn:theta}
\theta(t)= \frac{1}{\lambda}\frac{\partial}{\partial t}f^r(0,t)+f(0,t)+\frac{\eta^2}{2\lambda^2}\left(1-\e^{-2\lambda t}\right),
\end{eqnarray}
where the instantaneous forward rate is given by:
\begin{equation*}
f^r(0,t)=-\frac{\partial}{\partial t}\log P_{mrkt}(0,t).
\end{equation*}

The Hull-White process in~(\ref{eqn:HW}) allows for a large time step Monte Carlo simulation with the following {\it exact} formula:
\begin{equation}
\label{eqn:HWSteps}
r(t)=r(s)\e^{-\lambda (t-s)}+\lambda\int_s^t\theta(u)\e^{-\lambda(t-z)}\d z + \frac{\eta}{\sqrt{2\lambda}}\e^{-\lambda(t-s)}W^\Q(\e^{2\lambda(t-s)}-1),
\end{equation}
for any $t>s.$

The number of trades does not affect the method discussed, so we consider generic $\bar{M}$ interest rate (IR) swaps in this portfolio, i.e.:
\begin{eqnarray}
\label{eqn:portfolioIR}
V(t,r(t))=\sum_{i=1}^{\bar M}V_i(t,r(t)),
\end{eqnarray}
with $V_i(t,r(t))$ being the $i$'th IR swap evaluated at time $t$.
The benefit of using the HW model~(\ref{eqn:HW}) is that at any future time $t>t_0,$ the value of the $i$'th IR swap with payments at $\{T_{j+1},T_{j+2},\dots,T_{\bar{M_i}}\}$ is known analytically and is given by:
\begin{eqnarray}\label{swapSingleCCy}V_i(t,r(t))=N_i\sum_{k=j+1}^{\bar M_i}\tau_kP(t,T_k)\left(\ell_k(t)-K_i\right)\1_{t\leq T_k},\;\;\;\text{with}\;\;\; \ell_k(t)=\frac{1}{\tau_k}\left(\frac{P(t,T_{k-1})}{P(t,T_k)}-1\right),\end{eqnarray}
where $N_i$ stands for the notional amount, $K_i$ is the fixed rate and where the ZCB, $P(t,T)$, is defined as~\cite{OosterleeGrzelakBook}:
\begin{equation}
\label{eqn:HW:ZCB}
P(t,T)=\E^\Q\left[\e^{-\int_t^Tr(z)\d z}\Big|\F(t)\right]=\e^{A(t,T)+B(t,T)r(t)},
\end{equation}
with,
\begin{eqnarray}
\label{eqn:B}
B(t,T)&=&\frac1\lambda\left(\e^{-\lambda(T-t)}-1\right),\\
\label{eqn:A}
A(t,T)&=&\lambda\int_t^T\theta(z)
B(z,T)\d z+\frac{\eta^2}{4\lambda^3}\left[\e^{-2\lambda
(T-t)}\left(4\e^{\lambda
(T-t)}-1\right)-3+2\lambda(T-t)\right].
\end{eqnarray}
Because of the affinity of the Hull-White model, the evaluation of a swap at any time $t$ only depends on the interest rate, $r(t)$, which is highly beneficial.

Given the $\text{xVA}(t_0)$ discretization in~(\ref{eqnxVA}), the portfolio in~(\ref{eqn:portfolioIR}) needs to be simulated at every time $T_1,\dots,T_{N_T}$. In a standard approach to get the exposures at each exposure date, $T_k$, portfolio $V(T_k,r(T_k))$ needs to be evaluated at each path of the stochastic process $r(T_i)$. Typically, to get convergence of the solution, many simulated Monte Carlo paths are used. Often the number of paths is limited by the hardware, especially for large portfolios.

We define the approximating function $\widetilde g(V(T_k,\{r_{j}(T_k)\}_{j=1}^{n_1}),r(T_k)),$ which will only require portfolio evaluation at the $n_1$ collocation points, \begin{eqnarray}
\label{eqn:collR}\{r_{j}(T_k)\}_{j=1}^{n_1}:=\{r_{j}(T_k),\dots,r_{n_1}(T_k)\}.
\end{eqnarray}

One needs to decide which technique for the computation of the SC points will be used. %Given the results from Table~~\ref{Tab:diemsionality}, it is  beneficial to use the standard SC method for $d=1,2$ and sparse grid approach for $d>2$.
Determining collocation points for a Gaussian process is a straight-forward exercise, i.e., the collocation points, $r_{i,j}$, are given by:
$$r_j(T_k) = \E[r(T_k)]+\sqrt{\Var[r(T_k)]}x_j^{\mathcal{N}(0,1)},\;\;\;j=1,\dots,n_1,$$
where $x_j^{\mathcal{N}(0,1)}$'s are the collocation points from a
standard normal and where the mean and the variance are given by,
\[\E\left[r(t)|\F(t_0)\right]=r_0\e^{-\lambda
t}+\lambda\int_0^t{\theta(z)\e^{-\lambda(t-z)} \d z},\;\;\;\text{and}\;\;\;
\var\left[r(t)|\F(t_0)\right]=\frac{\eta^2}{2\lambda}\left(
1-\e^{-2\lambda t}\right).\]
The collocation points for a standard normal,
$x_j^{\mathcal{N}(0,1)}$, are known analytically and are tabulated
in~\cite{grzelak2015stochastic}. These points are also related to
the abscissas of the Gauss-Hermite quadrature. The difference between
the Gauss-Hermite abscissas $\{x_j^H\}_{j=1}^{n_1}$ and the collocation
points $\{x_j^{\mathcal{N}(0,1)}\}_{j=1}^{n_1}$, with
$X\sim\mathcal{N}(0,1)$, is the weight function, i.e., we use the
normal distribution while Gauss-Hermite quadrature is based on the
function $\e^{-x^2}$. The relation between the abscissas is given by $x_j^H=x_j^{\mathcal{N}(0,1)}/\sqrt{2},$ so once we have
the abscissas from the Gauss-Hermite quadrature, they can be used in
the collocation method.

Once we employ the standard collocation method, the approximating portfolio function could be generated using Lagrange polynomial:
\begin{equation}
\label{eqn:LagrangePortfoli}
\widetilde g(V(T_k,\{r_{j}(T_k)\}_{j=1}^{n_1}),r(T_k))=\sum_{j=1}^{n_1}V(T_k,r_{j}(T_k))\psi_j(r(T_k)),%\;\;\;
%\psi_j(x)=\prod_{l=1,\\j\neq
%l}^{n_1}\frac{x-r_{j}(T_k)}{r_{l}(T_k)-r_{j}(T_k)},
\end{equation}
where $r_{j}(T_k),r_{l}(T_k)$ are the collocation points defined in~(\ref{eqn:collR}),  $V(T_k,r_{j}(T_i))$ are the portfolio evaluations for a given expiry date, $T_k$, evaluated at the collocation points, $r_j(T_k)$, and where $\psi_j(r(T_k))$ is defined as in~(\ref{eqn:lagrange_1}).

Given the representation in~(\ref{eqn:LagrangePortfoli}), we can evaluate function $\widetilde g(\cdot)$ at all stochastic paths of the underlying process, $r(T_i).$ The Lagrange polynomial's particular choice is not obligatory; the reader can decide employing different interpolation routines. However, a benefit of Lagrange representation is that its monomial basis representation allows very efficient polynomial evaluation, even for a large number of samples.

%================ REPEATED
In the numerical experiment, we consider discounted expected (positive) exposures and potential future exposures defined for the exposure date, $T_k$, specified as:
\begin{eqnarray*}
\text{EE}(t_0,T_k)&=&\E^\Q\left[\frac{M(t_0)}{M(T_k)}E^+(T_k,r(T_k))\big|\F(t_0)\right],\\
\text{PFE}(t_0,T_k) &=& \inf\{x\in\R:p\leq F_{E^+(T_k,r(T_k))}(x)\},
\end{eqnarray*}
with positive exposures, $E^+(T_k,r(T_k)),$ defined in~(\ref{eqn:exposure}) for which the approximating function $\widetilde g(\cdot),$ reads:
\begin{eqnarray*}
E^+(T_k,r(T_k))\approx \max(\widetilde g(V(T_k,\{r_{j}(T_k)\}_{j=1}^{n_1}),r(T_k)),0),
\end{eqnarray*}
and where $M(t)$ under the HW model represents the money-savings account defined by $\d M(t)=r(t)M(t)\dt.$

%Another measure used here to measure the quality is the so-called {\it Potential Future Exposure} (PFE) which measures the maximum credit exposure calculated at some confidence level. The measure can be associated with measuring the quality of the approximating function $\widetilde g(\cdot)$ in the tails of the distribution of $V(t,r(t))$, at $t=T_k$. PFE, at time $t$, i.e. $\text{PFE}(t_0,t)$, is defined as a quantile of the positive exposure $E^+(t,r(t))$,
%where $F_{E^+(t,r(t))}(x)$ is the CDF of positive exposures observed at time $t$. Coefficient $p$ represents the  certainty level, the quantile level.
%================ REPEATED

To illustrate the approximation quality, we consider two portfolios: one with a single IR swap with a maturity of 10 years and one with 25 different IR swaps with varying  maturities. In the experiment we consider 50 equally spaced exposure dates, i.e. $T_1=0.5,T_2=1,\dots,T_{50}=25$.

In Figure~\ref{fig:SingleMultiSwap} (first figure), the numerical results for the computer simulation are presented~\footnote{The reported numerical results were obtained with typical model parameters observed in 2020.}. The figure shows EE's and PFE's for different confidence levels, $p$. We consider only $n_1=3$ collocation points in the experiment, thus only 3 portfolio evaluations per exposure date.  The results are highly satisfying. Moreover, for 10k Monte Carlo paths and 50 exposure dates, we have reduced the number of portfolio evaluations from 500k to only 150. Thus, the number of portfolio evaluations is reduced by more than a factor of 3,000.

Figure~\ref{fig:SingleMultiSwap} (second figure) shows that in the second portfolio case, the quality in the tail for PFE's, represented with $p=0.99$, can be further improved. As illustrated in Figure~\ref{fig:SingleMultiSwap} (third figure), the error for a $99\%$ PFE quantile can be reduced to almost zero once the number of collocation points is increased to $n_1=4$. The increase of collocation points has also increased the number of portfolio evaluations from 150 to 200. Then, the reduction of the number of portfolio evaluations is 2500. Discussion on the error propagation depending on the number of the collocation points takes place in Section~\ref{sec:Error}.
\begin{figure}[h!]
  \centering
    \includegraphics[width=0.49\textwidth]{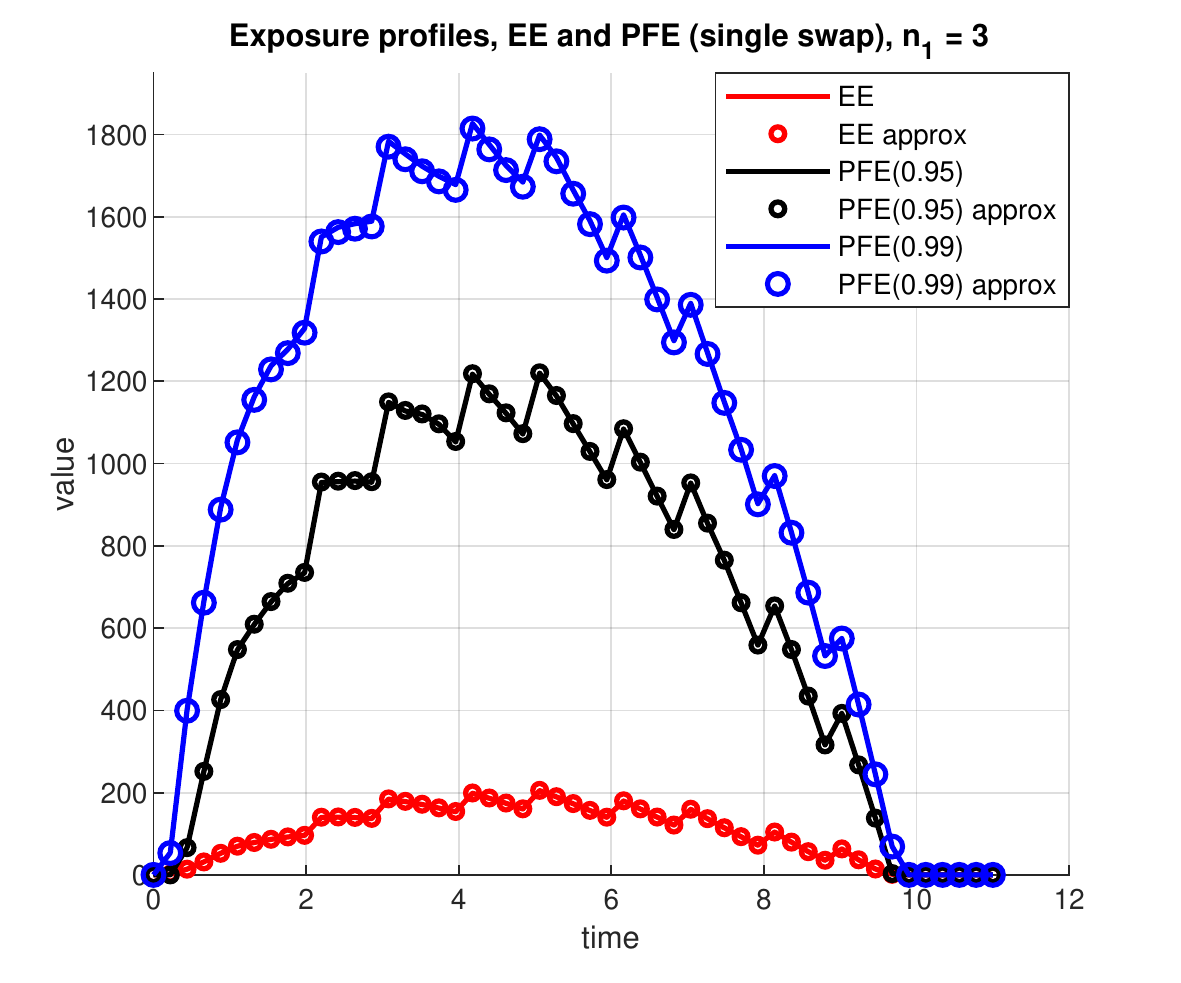}
    \includegraphics[width=0.49\textwidth]{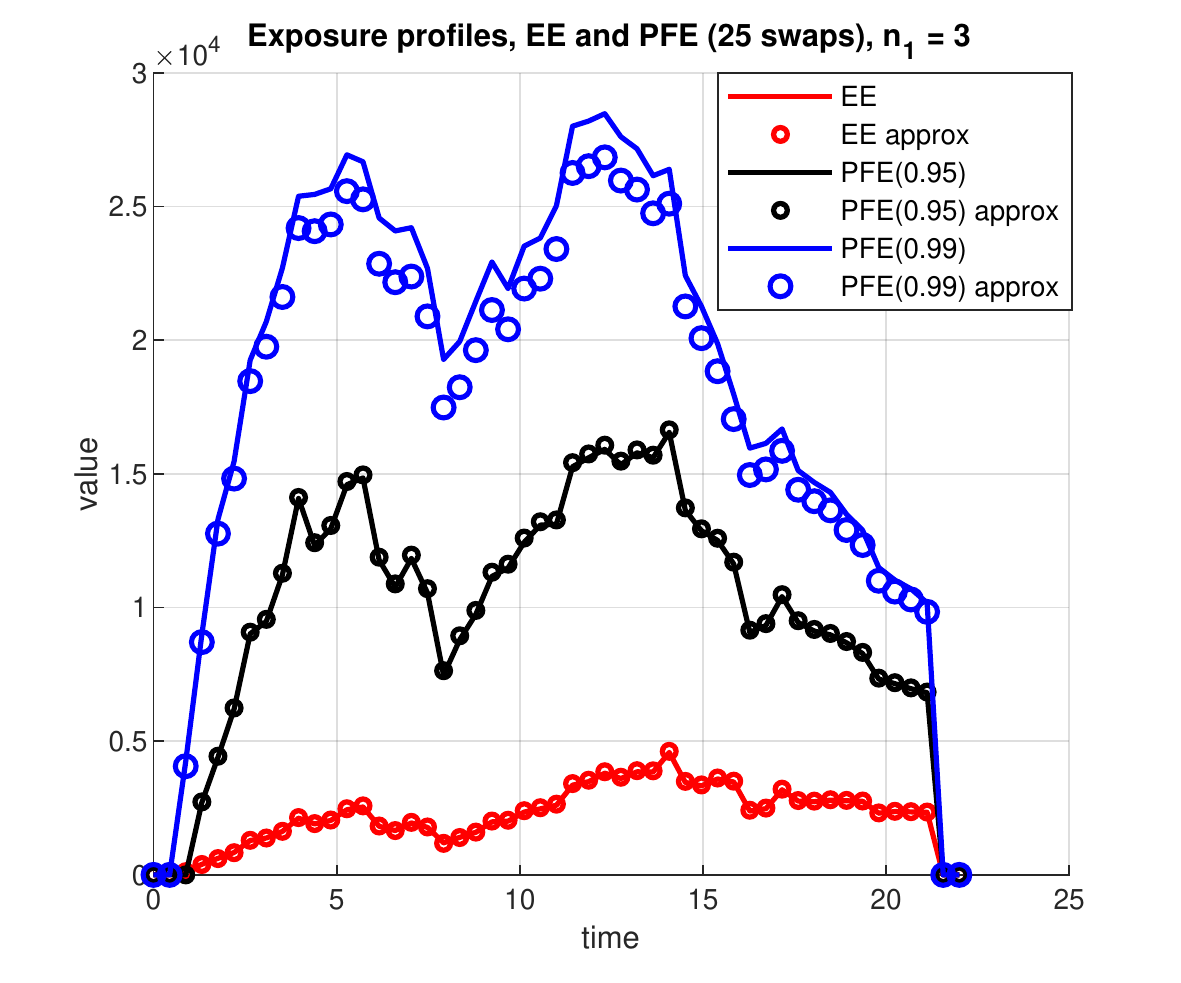}\\
    \includegraphics[width=0.49\textwidth]{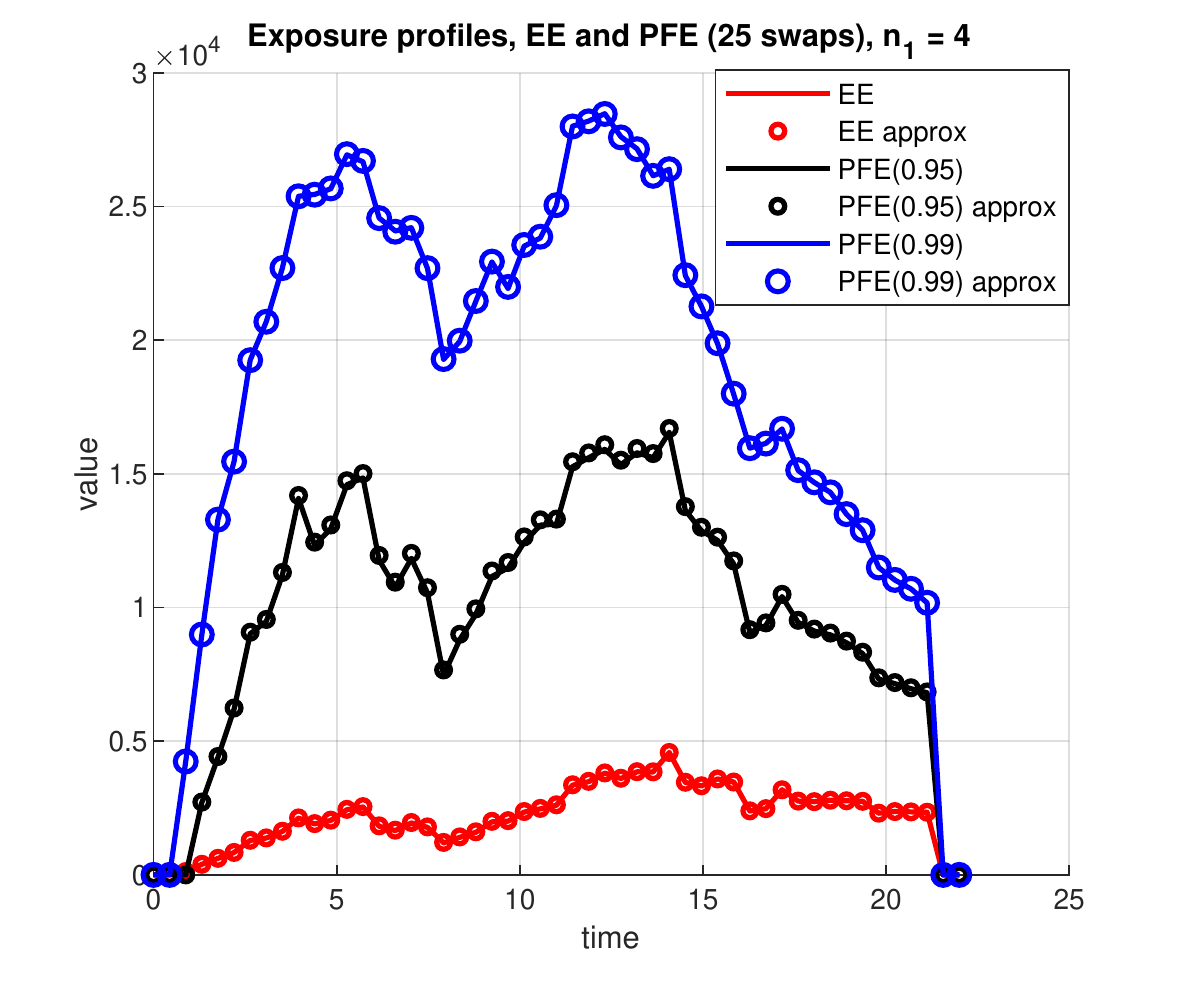}
      \caption{Portfolio with a single IR swap for $n_1=3$ (first figure). Portfolio with 25 swaps for $n_1=3$ (second figure) and portfolio with $n_1=4$ (third figure).}
      \label{fig:SingleMultiSwap}
\end{figure}

%%%%%%%%%%%%%%%%%%%%%%%%%%%%%%%%%%%%%%%%%%%%%%%%%%%%%%%%%%%%%%%%%%%%%%%%%%%%%%%%%%%%%
\subsection{Multi-D Case: Portfolios with Multiple Risk Factors}
\label{sec:multid}
%%%%%%%%%%%%%%%%%%%%%%%%%%%%%%%%%%%%%%%%%%%%%%%%%%%%%%%%%%%%%%%%%%%%%%%%%%%%%%%%%%%%%
Although the example in Section~\ref{sec:1D} is illustrative and confirms the high potential of the collocation method, portfolios typically involve multiple risk factors, for example, swaps in different currencies, foreign exchange forwards, stocks in different currencies, or other products.

Here, we consider a representative portfolio consisting of swaps in $d_c$ different currencies. This implies a problem of $d=2d_c+1$ dimensions: $d_c$ processes for foreign interest rates, $r_k(t)$, for $k=1,\dots,d_c $, $d_c$ processes for foreign exchange process that exchange a foreign amount of money to a base currency, $y_k(t)$, and 1 for an interest rate in a base currency, $r_b(t)$. This setting can be easily further extended with stochastic volatility or multi-factor processes for interest rates. However, we note that although advanced models for modeling foreign exchange rates with implied volatility smile and skew exist~\cite{GrzelakOosterlee:2011:CrossCurrency}, here, we consider the FX model to be based on a hybrid, the Black-Scholes Hull-White model~\cite{OosterleeGrzelakBook}.

The system of SDEs for the FX model, under the risk-neutral measure of a base currency, reads:
\begin{eqnarray}
\label{eqn:FXmulti}
\d y_k(t) &=&\left(r_b(t)-r_k(t)\right)y_k(t)\dt + \sigma_{y,k}y_k(t)\dW_{k}^{y}(t),
%\d r_b &=&\lambda_b\left(\theta_b(t)-r_b(t)\right)\dt + \eta_b\dW^{b}(t),\nonumber\\
%\d r_k &=&\lambda_k\left(\hat\theta_k(t)-r_k(t)\right)\dt + \eta_k\dW_k(t),\nonumber
\end{eqnarray}
with $r_b$ representing the base currency~\ref{eqn:HW} and with $k=1,\dots, d_c$ indicating foreign currencies, where we consider a full matrix of correlations with the following correlation coefficients,
\begin{eqnarray*}
\dW_{k}^{y}(t)\dW^{b}(t)=\rho_{k,b}^y\dt,\;\;\;\dW_{k}^{y}(t)\dW_j(t)=\rho_{k,j}^y\dt,\;\;\;\dW^{b}(t)\dW_k(t)=\rho_{k}\dt,
\end{eqnarray*}
and where the convexity adjusted term structure for foreign short rate processes is given by $\hat\theta_k(t)=\theta_k(t)-\eta_k\sigma_{y,k}\rho_{k,k}^y$ with $\theta(t)$ defined in~(\ref{eqn:theta}). Although the short-rate process, $r_k(t),$ contains a convexity correction in the drift, the model belongs to the affine class of processes, and the zero-coupon bonds can be priced as in~(\ref{eqn:HW:ZCB}). Thus, all the linear interest rate products can be expressed as linear combinations of zero-coupon bonds.
The representation of the system in~(\ref{eqn:FXmulti}) allows for large time step Monte Carlo simulation with the following solution for each FX process:
\begin{eqnarray*}
y_k(t) &=&y_k(s)\exp\left({\int_s^t\left(r_b(z)-r_k(z)-\frac12\sigma_{y,k}^2\right)\dt + \sigma_{y,k} (W_{k}^{y}(t)-W_{k}^{y}(s))}\right),
\end{eqnarray*}
and large-time steps for $r_b(t)$ and $r_k(t)$, as presented in~(\ref{eqn:HWSteps}).

A cross-currency portfolio consisting of swaps in $d_c$ different currencies is then given by:
\begin{eqnarray}
\label{ref:porftolioMulti}
V(t,{\bf X}(t)) =\sum_{i=1}^{M_b}V_i^b(t,r_b(t)) + \sum_{k=1}^{d_c}y_k^b(t)\sum_{i=1}^{M_k} V_i^k(t,r_k(t)),
\end{eqnarray}
with the following state vector, ${\bf X}(t) = [r_b(t),y_1^b(t),\dots,y_{d_c}^b(t),r_1(t),\dots,r_{d_c}(t)]^\T$,
and where the first sum indicates the sum over all swaps under the base currency, the second term indicates the summation over all swaps under foreign currencies that are {\it exchanged} with $y_k^b(t)$ to the base currency. The value of the portfolio in~(\ref{ref:porftolioMulti}) is expressed in the base currency. In a realistic scenario involving multiple currencies, the complexity of the exposure computation has increased. A portfolio with $d_c$ foreign currencies would require a tensor grid of $d=2d_c+1$ dimensions. This, as presented in Table~\ref{Tab:diemsionality}, may lead to a higher number of portfolio evaluations than one would expect with the Monte Carlo simulation. We will apply the SC method with Smolyak's sparse grid algorithm to reduce the number of computations.

Following the procedure described in Section~\ref{sec:generalConstruction}, we approximate the portfolio, for each exposure date, $T_k$, with:
\begin{eqnarray}
\label{eqn:approxMulti1}
V(T_k,{\bf X}(T_k))\approx \widetilde g\big(\{V\}_{i_1,\dots,i_d},{\bf X}(T_k)\big),
\end{eqnarray}
where function $\widetilde g(\cdot)$ is built based on Smolyak's sparse grid method presented in Equation~(\ref{eqn:SmolyakGrid2}) and where $\{V\}_{i_1,\dots,i_d}$ are the portfolio values evaluated at the collocation points based on Smolyak's sparse grid.

To illustrate the performance of the method, we consider the base currency to be EUR and a portfolio consisting of many swaps in Euros (EUR), British Pounds (GBP), and Polish Z\l oty (PLN). Such a problem requires a $7-$dimensional system of SDEs in~(\ref{eqn:FXmulti}):
\begin{equation}
    \label{eqn:X}
{\bf X}(t) = [r_{\text{\euro}}(t),y_{\$}^{\text{\euro}}(t),y_{\text{\pounds}}^{\text{\euro}}(t),y_{\text{z\l}}^{\text{\euro}}(t),r_{\$}(t),r_{\text{\pounds}}(t),r_{\text{z\l}}(t)]^\T,\end{equation}
where each process is calibrated to individual market and is then used in the system's simulation of SDEs in~(\ref{eqn:FXmulti}). In the experiment, we consider the case where all the processes are correlated. The correlation coefficients can be estimated based on historical data or estimated from correlation products that are occasionally present in the OTC market.

In the experiment, we consider $N_T=75$ equally spaced exposure dates, i.e., $T_1,T_2,\dots,T_{N_T}$, and a Monte Carlo simulation with 25k stochastic paths. Details regarding model parameters used in the experiments are presented in~\ref{sec:appendix2}.

As previously, we compare the results from a full-blown Monte Carlo simulation against the approximation in~(\ref{eqn:approxMulti1}). The graphical representation of the numerical results for the expected positive exposures and PFEs for different significance levels is presented in Figure~\ref{fig:PortfolioSmolyak}. The results vary depending on the sparse grid parameter $\mu$, which defines the grid-level.

\begin{figure}[h!]
  \centering
    \includegraphics[width=0.49\textwidth]{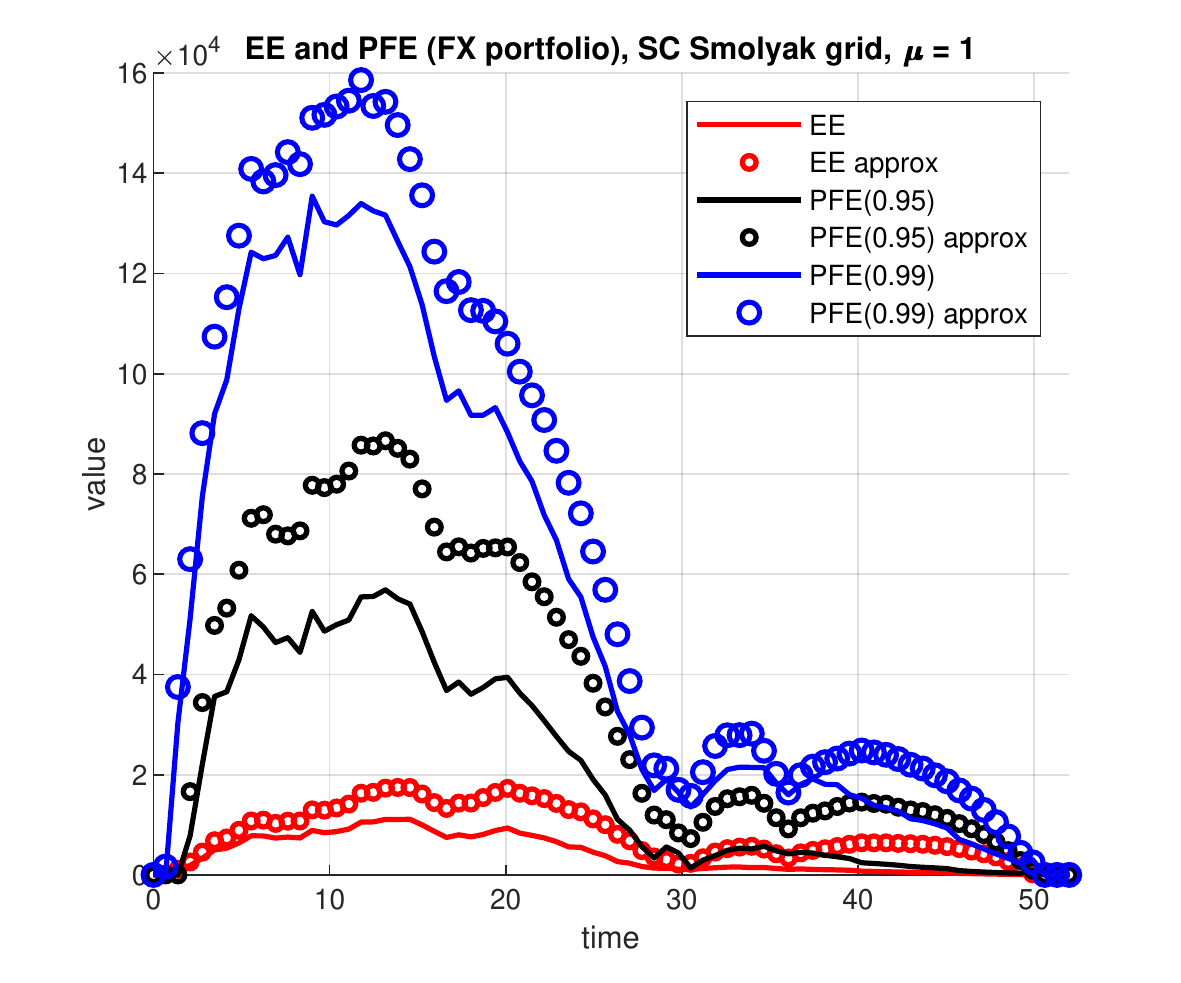}
    \includegraphics[width=0.49\textwidth]{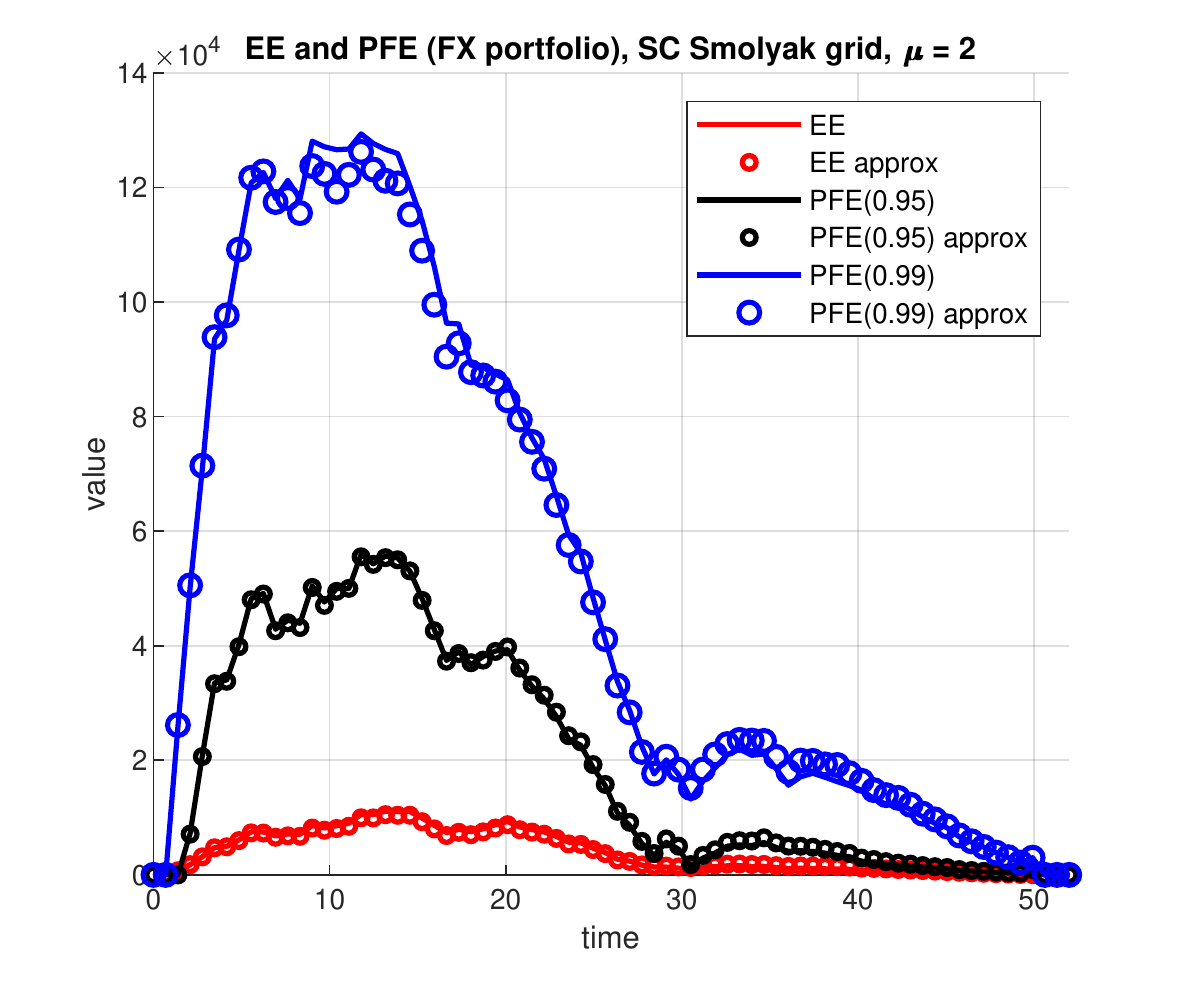}\\
    \includegraphics[width=0.49\textwidth]{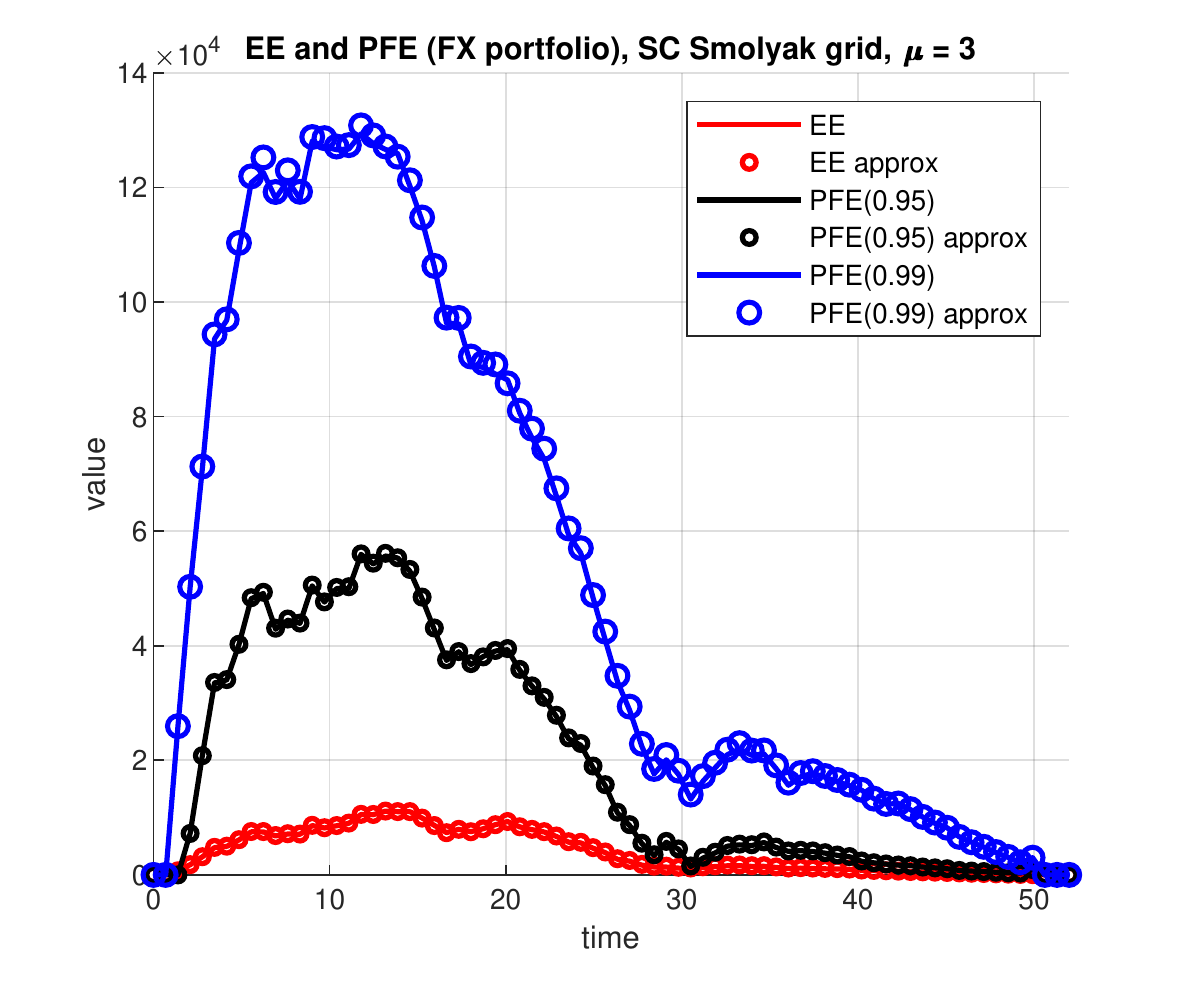}
      \caption{Expected Exposures, PFE for different significance level $(0.95,0.99)$ depending on the parameter $\mu=1,2,3.$}
      \label{fig:PortfolioSmolyak}
\end{figure}
The visual inspection shows that highly satisfactory results are obtained, already for $\mu=2$. This particular choice requires only $113$ portfolio evaluations per expiry date, $T_k.$ Details regarding the error defined as a relative difference are presented in Table~\ref{table:NDCase}. We conclude that the method performs well, and a significant computational gain is guaranteed. In the next section, further improvements for the methods are discussed.

%\begin{eqnarray}\text{error}=\frac{1}{N_T}\sum_{k=1}^{N_T} \frac{\left| V(T_k,{\bf X}(T_k))- \widetilde g\big(\dots,{\bf X}(T_k)\big)\right|}{V(T_k,{\bf X}(T_k))}.
%\end{eqnarray}

%\begin{eqnarray}\text{error}=\frac{1}{N_T}\sum_{k=1}^{N_T}| f- %\widetilde f|/f.
%\end{eqnarray}

\begin{table}[!h]
\begin{center}\footnotesize
\caption{The total number of portfolio evaluations depending on the parameter $\mu.$ The speed-up number reflects the reduction of portfolio evaluations against Monte Carlo simulation with $25k$ paths. In the experiment, we considered $N_T=75$ exposure dates, and where the error is defined as $\text{error}=\frac{1}{N_T}\sum_{k=1}^{N_T}| f(T_k)- \widetilde f(T_k)|/f(T_k)1_{f(T_k)\neq0}$, for $f(T_k)$ being  EE or PFE and where $\widetilde f(T_k)$ indicates the SC approximation.}
 \begin{tabular}{c|c |c| c  c  c }
  \hline
   & \multirow{2}{*}{total $\#$ of portfolio eval.}&speed-up& \multicolumn{3}{c}{error} \\\cline{4-6}
   $d=7$&  &vs. Monte Carlo &EE&$\text{PFE}_{0.95}$&$\text{PFE}_{0.99}$\\ \hline
   $\mu=1$& $75\times 15$ &$1666$ &3.9686 &   2.6203 &   0.3865\\
   $\mu=2$&$75\times113$&$221$ &0.1023   &0.0720   &0.0248 \\
   $\mu=3$&$75\times589$&$42$ & 0.0313   &0.0101 &  0.0166\\\hline
   Monte Carlo&$75\times25000$&-& - & -&- \\\hline
 \end{tabular}
 \label{table:NDCase}
 \end{center}
 \end{table}

%%%%%%%%%%%%%%%%%%%%%%%%%%%%%%%%%%%%%%%%%%%%%%%%%%%%%%%%%%%%%%%%%%%%%%%%%
\subsection{Divide and Conquer: Further Reduction of Portfolio Evaluations}
\label{sec:divideand}
%%%%%%%%%%%%%%%%%%%%%%%%%%%%%%%%%%%%%%%%%%%%%%%%%%%%%%%%%%%%%%%%%%%%%%%%%
In Section~\ref{sec:2_1}, in Table~\ref{Tab:diemsionality}, we have shown a relation between the dimensionality and the number of portfolio evaluations. The application of the sparse grid techniques facilitates an improvement compared to the tensor product. However, it is also clear that with the increasing dimensionality, the benefits of the SC method are reduced, i.e., although Table~\ref{Tab:diemsionality} shows that the number of collocation points grows only polynomially in the dimension, $d$ for $d>10$, the number of grid points may exceed the typical number of Monte Carlo paths, especially for high parameter value for $\mu.$

Although some portfolios may depend on multiple risk factors, individual trades typically do not. Even trades such as cross-currency swaps will highly unlikely depend on many currencies- this is an important observation as it allows us to the number of reduce portfolio evaluations further.
By dividing the portfolio into sub-portfolios, we can {\it linearize} the problem and further reduce the numerical complexity. For example, a grid with $d=8$ and level $\mu=3$, according to Table~\ref{Tab:diemsionality}, would require $849$ grid points. If we divide the portfolio into two sub-portfolios for $d=4$, we only will need $137$ evaluations per trade, and with four sub-portfolios with $d=2$, only $29$  evaluations yielding a huge computational gain~\footnote{Note that reported evaluations are per sub-portfolio, however since these portfolios are mutually disjoint none of the trades is evaluated more than once.}.

\begin{rem}[Sub-portfolios and Netting]
The division into sub-portfolio has no significant impact on the netting effects. The SC method is only used to approximate the value of a sub-portfolio, which later is evaluated for all the underlying risk factors.
\end{rem}

Let us look at a realistic portfolio discussed in Section~\ref{sec:multid}, which consists of swaps in different currencies. The strategy proposed above would yield the following portfolio decomposition~\footnote{We consider here $n_1$ collocation points for every dimension, $j=1,\dots,d_c$.}:
\begin{eqnarray}
\label{ref:porftolioMulti2}
V(T_k,{\bf X}(T_k))&=&\overline V^b(T_k,r_b(T_k)) + \sum_{j=1}^{d_c}y_j^b(t)\overline V^j(T_k,r_j(T_k))\nonumber\\
&\approx&g_b\big(\{\overline V^b\}_{i},r(T_k)\big)+\sum_{j=1}^{d_c}y_j^b(t)g_j\big(\{\overline V^j\}_{i},r_j(T_k)\big),\;\;\;i=1,\dots,n_1,
\end{eqnarray}
where ${\bf X}(t)$ is defined in~(\ref{eqn:X}) and where $\overline V^b(\cdot)$ and $\overline V^j(\cdot)$ indicate portfolios in the base currency $b$ and foreign currency $j$, respectively, and are defined as:
\begin{eqnarray}
\label{eqn:VbVj}
\overline V^b(T_k,r_b(T_k)):=\sum_{i=1}^{M_b}V_i^b(T_k,r_b(T_k))
\;\;\;\text{and}\;\;\;\overline V^j(T_k,r_j(T_k)):=\sum_{l=1}^{M_j} V_l^j(T_k,r_j(T_k)).
\end{eqnarray}

The representation above makes up a further improvement of the SC method for computing exposures. By approximating each of the $d_j$ sub-portfolios, we can separate the value of a portfolio from the cross-currency variable, $y_j^b$, for $j=1\dots,d_c$. This procedure has reduced the dimensionality from a 7D problem to a sum of 1D problems.

In Table~\ref{table:1DCaseSUB}, the numerical results are presented. We observe that by dividing the portfolio into sub-portfolios, a significant valuation improvement is achieved, i.e., we only require $4$ to $6$ evaluations per exposure date to guarantee errors that are even smaller than those reported in Table~\ref{table:NDCase}. Compared with a brute-force Monte Carlo pricing, the speed improvement varies from a factor of 6000 for $n_1=4$ to 4000 for $n_1=6$ collocation points.

\begin{table}[!h]
\begin{center}\footnotesize
\caption{The total number of trade evaluations, $\#$, depending on the number of the collocation points, $n_1$. The speed-up number reflects the reduction of portfolio evaluations against Monte Carlo simulation with $25k$ paths. Number of valuations represents the product of the number of exposure dates: $T_k$, $N_T=75,$ times the number of collocation points, $n_1$.}
 \begin{tabular}{c|c |c| c  c  c }
  \hline
   & \multirow{2}{*}{total $\#$ of sub-portfolio eval.}&speed-up& \multicolumn{3}{c}{error} \\\cline{4-6}
  &&vs. Monte Carlo &EE&$\text{PFE}_{0.95}$&$\text{PFE}_{0.99}$\\ \hline
   $n_1=2$&$75\times 2 $&$12500$& 0.6048   &0.5346 &  0.5841\\
   $n_1=3$&$75\times 3 $&$8333$& 0.1432   &0.0011 &  0.1605\\
   $n_1=4$&$75\times 4$&$6250$& 0.0074   &0.0757 &  0.0088\\
   $n_1=5$&$75\times 5$&$5000$& 0.0071   &0.0289 &  0.0223\\
   $n_1=6$&$75\times 6$&$4167$ &  0.0014   &0.0024 &  0.0126\\\hline
   Monte Carlo&$75\times25000$&-& - & -&- \\\hline
 \end{tabular}
 \label{table:1DCaseSUB}
 \end{center}
 \end{table}

The decomposition into a sum of 1D problems, as described above, may not always be possible in an actual portfolio as it depends on the type of derivatives, e.g., with cross-currency swaps depending on three currencies, we would expect sums of 3D portfolios.

\begin{rem}[Parallelization]
The encouraging results presented in this section can be further improved: each sub-portfolio can be evaluated independently, allowing for parallelization at the sub-portfolio level. Although relatively cheap, one can also parallelize the evaluation of function $\widetilde g(\cdot)$ at each Monte Carlo path facilitating even further speed gains.
\end{rem}
%%%%%%%%%%%%%%%%%%%%%%%%%%%%%%%%%%%%%%%%%%%%%%%%%%%%%%%%%%%%%%%%%%%%%%%%%
\subsection{Portfolios with non-linear products driven by multi-factor models}
%%%%%%%%%%%%%%%%%%%%%%%%%%%%%%%%%%%%%%%%%%%%%%%%%%%%%%%%%%%%%%%%%%%%%%%%%
In the previous section, we have presented the benefits of dividing the pricing problems into pieces. Unfortunately, the same strategy cannot be applied to problems that depend on multiple factor models. When determining risk limits, it is desired to determine PFEs based on multi-factor models. Such models mimic the yield curve dynamics time time much better than single-factor models. In this section, we consider such a case.\\
\indent This section considers an extension of the case discussed in Section~\ref{sec:divideand}- a portfolio consisting of linear and non-linear derivatives driven by multi-factor processes. In particular, we consider the portfolio defined in~(\ref{ref:porftolioMulti2}) with additional, non-linear products, namely- several interest rate swaptions with the varying strike, maturity and notional values. Furthermore, swaptions are added for each currency, and the portfolio is ensured to be balanced, i.e., there is no dominance of either interest rate swaps or swaptions that potentially could skew the results.\\
%In this section, we consider an extension of the case discussed in Section~\ref{sec:divideand}- a portfolio consisting of both: linear and non-linear derivatives, driven by the multi-factor processes. In particular, we consider the portfolio defined in~(\ref{ref:porftolioMulti2}) with additional, non-linear products, namely- the number of interest rate swaptions with the varying strike, maturity and notional. Furthermore, swaptions are added for each currency, and the portfolio is ensured to be balanced, i.e., there is no dominance of either interest rate swaps or swaptions that potentially could skew the results. \\
\indent In the experiment, each of the interest rate processes is driven by the 2-factor Hull-White model~\cite{Hull1994}, $r(t)=x^r(t)+y^r(t)+\psi^r(t)$ with $x^r(t)$ and $y^r(t)$ being OU processes with the initial and long-term mean equal to $0$ and where $\psi^r(t)$ is a term-structure (details are presented in~\cite{BrigoMercurio:2007}). Swaptions, contrary to interest rate swaps, are computationally more expensive as they require numerical integration and optimization for every Monte Carlo realization $\{x^r(t_i),y^t(t_i)\}$ (the details regarding the pricing of swaptions under the two-factor models are, for reader's convenience, in~\ref{sec:appendix3})).
%In the experiment, each of the interest rate processes is driven by the 2-factor Hull-White model~\cite{Hull1994}, $r(t)=x^r(t)+y^r(t)+\psi^r(t)$ with $x^r(t)$ and $y^r(t)$ being the OU processes with the initial and long-term mean equal to $0$ and where $\psi^r(t)$ is a term-structure (details are well presented in~\cite{BrigoMercurio:2007}). Swaptions, contrary to interest rate swaps, are computationally more expensive as they require numerical integration and optimization for every Monte Carlo realization $\{x^r(t_i),y^t(t_i)\}$ (the details regarding the pricing of swaptions under the two-factor models are, for reader's convenience, in~\ref{sec:appendix3}).}

Under the two-factor model the approximating formula is as follows,
\begin{eqnarray}
V(T_k,{\bf X}(T_k))
\approx g_b\big(\{\overline V^b\}_{i_1,i_2},x^r(T_k),y^r(T_k)\big)+\sum_{j=1}^{d_c}y_j^b(t)g_j\big(\{\overline V^j\}_{i_1,i_2},x^r_j(T_k),y^r_j(T_k)\big),
\end{eqnarray}
where $i_1=1,\dots,n_1$, $i_2=1,\dots,n_2$, the state vector ${\bf X}(t)$ now includes a pair of $x^r(t)$, $y^r(t)$ per each currency.
%\begin{equation}
%    \label{eqn:X2F}
%{\bf X}(t) = [r_{\text{\euro}}(t),y_{\$}^{\text{\euro}}(t),y_{\text{\poun%ds}}^{\text{\euro}}(t),y_{\text{z\l}}^{\text{\euro}}(t),r_{\$}(t),r_{\tex%t{\pounds}}(t),r_{\text{z\l}}(t)]^\T,\end{equation}
Both $\overline V^b(\cdot)$ and $\overline V^j(\cdot)$ are the portfolios in each currency, as defined similarly as in~(\ref{eqn:VbVj}), with the exception that these portfolios explicitly depend on the underlying stochastic factors, $x^r(t)$ and $y^r(t)$, and additionally include swaptions.  \\
\indent As in the single factor Hull-White interest rate model, the model's performance is excellent. In Table~\ref{table:non_linear_portfolio}, convergence results are presented. We report that already, for $\mu=2$, high-quality results are obtained. For two-dimensional problems, $d=2$, with $\mu=2$, require only $13$ portfolio evaluations (see Table~~\ref{Tab:diemsionality})), thus resulting in a drastic reduction in the computational time. \\
%Both $\overline V^b(\cdot)$ and $\overline V^j(\cdot)$ are the portfolios in each currency, as defined similarly as in~(\ref{eqn:VbVj}) with an exception that these portfolios explicitly depend on the underlying stochastic factors, $x^r(t)$ and $y^r(t)$, and additionally include swaptions.  \\
%\indent Similarly, as in the single factor interest rate model by Hull-White, the model's performance is excellent. In Table~\ref{table:non_linear_portfolio} the convergence results are presented. We report that already, for $\mu=2$, high-quality results are obtained. For two-dimensional problems, $d=2$, $\mu=2$ is equivalent with $13$ portfolio evaluations (see Table~\ref{Tab:diemsionality}), thus a drastic reduction in portfolio evaluations. \\
\indent Figure~\ref{fig:PortfolioSmolyakSwaptions} illustrates the quality for the approximation for different exposure dates. Again, already for $\mu=2$ we report no significant errors along the lifetime of the portfolio.
\begin{table}[!h]
\begin{center}\footnotesize
\caption{The error associated with the sparse approximation of the ``non-linear'' portfolio of derivatives. The settings of the experiment are chosen described in Table~\ref{table:NDCase}. The number of portfolio evaluations is as in Table~\ref{table:1DCaseSUB}.}
 \begin{tabular}{c| c  c  c }
  \hline
   & \multicolumn{3}{c}{error} \\\cline{2-4}
  &EE&$\text{PFE}_{0.95}$&$\text{PFE}_{0.99}$\\ \hline
   $\mu=1$&   0.2518 &	 0.2553& 	 0.2932   \\
   $\mu=2$&    0.0266& 	 0.0111& 	 0.0177  \\
   $\mu=3$&    0.0213& 	 0.0178& 	 0.0250 \\
 \end{tabular}
 \label{table:non_linear_portfolio}
 \end{center}
 \end{table}
 \begin{figure}[h!]
  \centering
    \includegraphics[width=0.49\textwidth]{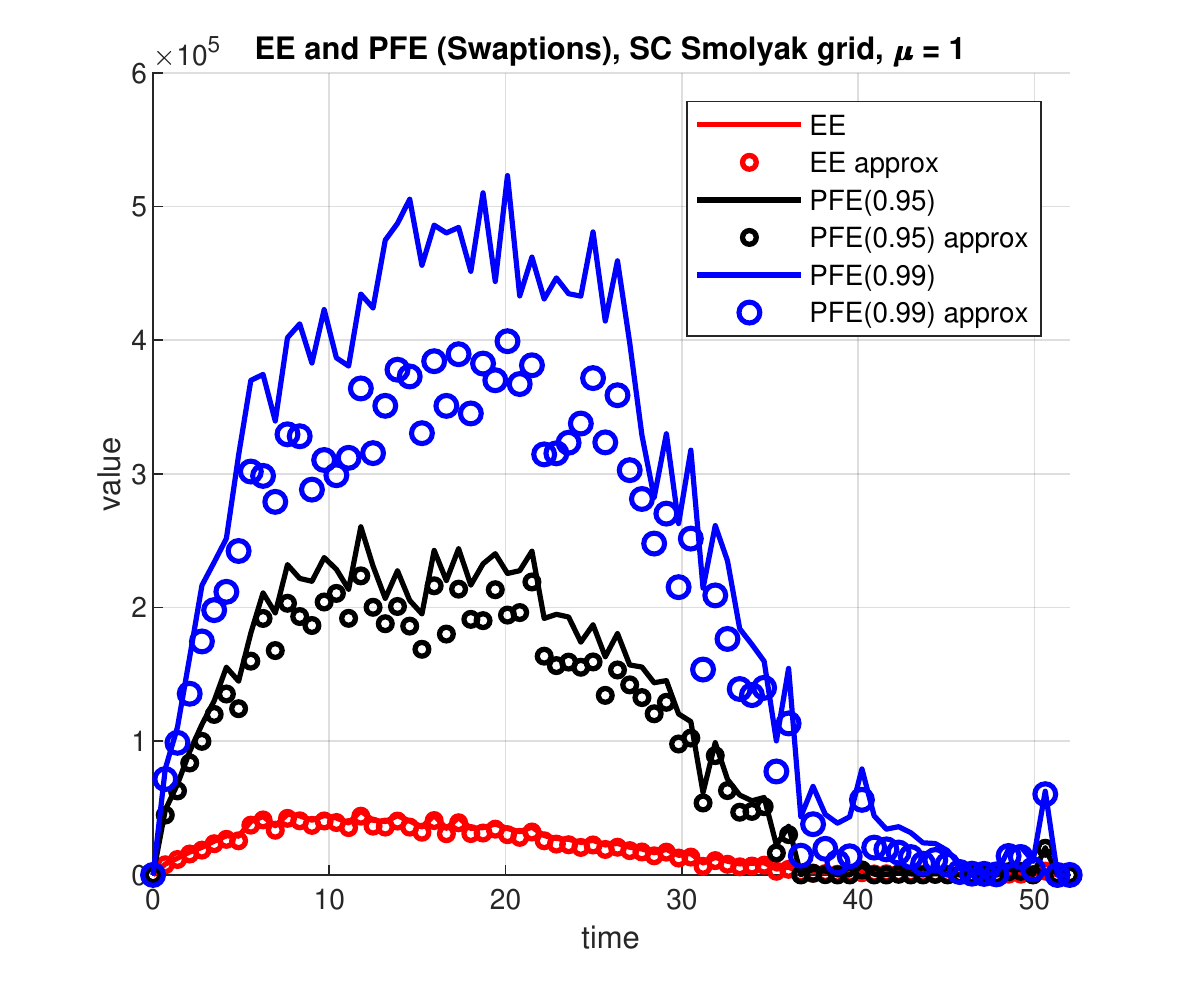}
    \includegraphics[width=0.49\textwidth]{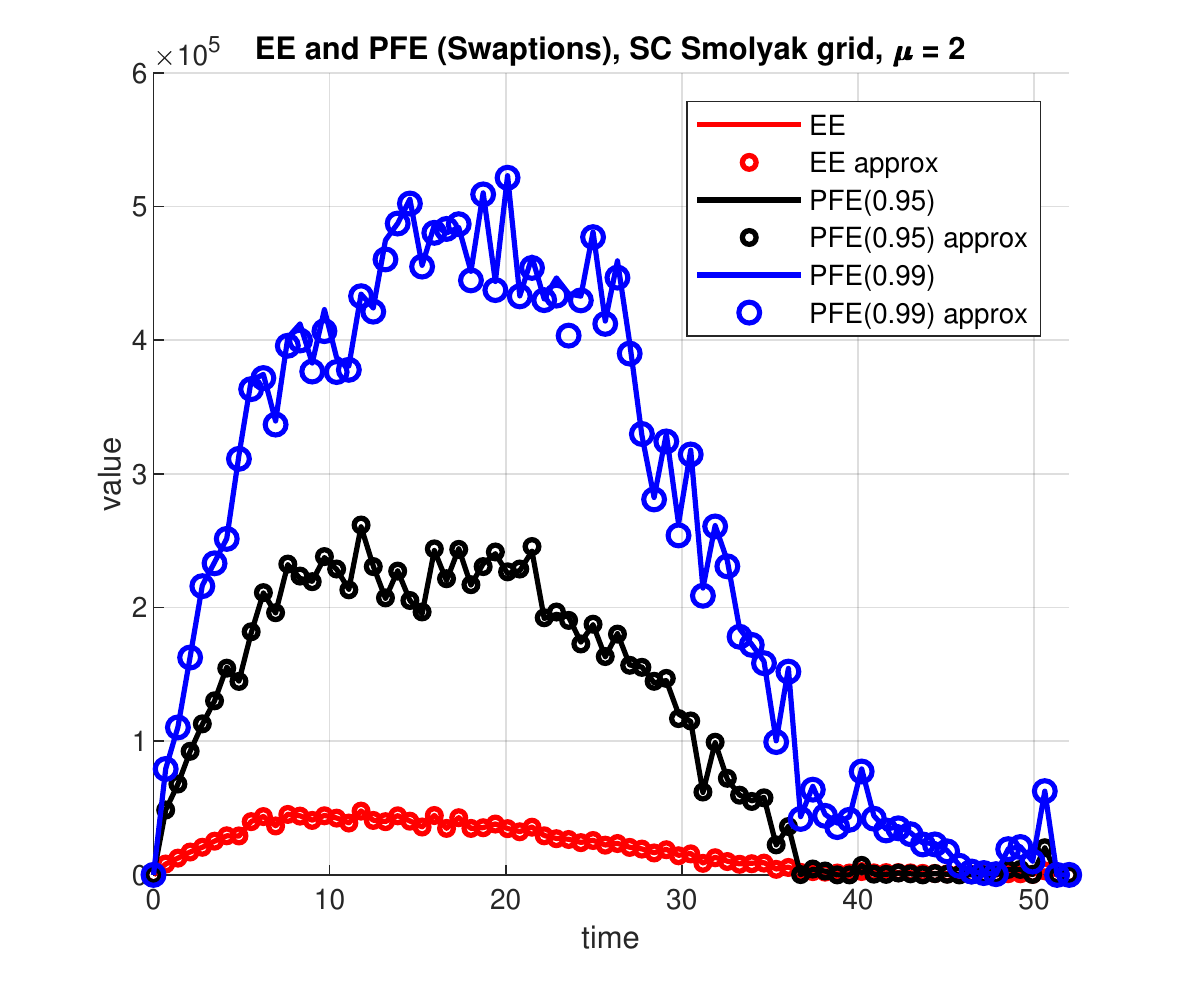}
      \caption{Portfolio with interest rate swaps and swaptions in different currencies. Expected Exposures, PFE for different significance level $(0.95,0.99)$ depending on the parameter $\mu=1,2.$}
      \label{fig:PortfolioSmolyakSwaptions}
\end{figure}

We conclude that under a multi-factor setting, the sparse grid methodology for approximating a portfolio consisting of linear and non-linear derivatives, there is no deterioration observed compared to a portfolio solely consisting of linear products. Furthermore, our experiments suggest that the method can be applied to even more exotic, possibly callable, derivatives.
%We conclude that under a multi-factor setting, the sparse grid methodology for approximating a portfolio consisting of linear and non-linear derivatives, no deterioration in the results is observed compared to a portfolio solely consisting of linear products. Furthermore, our experiments suggest that the method can be applied to even more exotic, possibly callable, derivatives.

%%%%%%%%%%%%%%%%%%%%%%%%%%%%%%%%%%%%%%%%%%%%%%%%%%%%%%%%%%%%%%%%%%%%%%%%%
\section{The SC Method: Implementation Details and Improvements}
\label{sec:4}
%%%%%%%%%%%%%%%%%%%%%%%%%%%%%%%%%%%%%%%%%%%%%%%%%%%%%%%%%%%%%%%%%%%%%%%%%
This section focuses on details and improvements regarding the SC method when applied to portfolio evaluations. Here, we discuss a domain scaling for Smolyak's grid, optimal collocation points, and adaptive grids for multi-D cases.
%%%%%%%%%%%%%%%%%%%%%%%%%%%%%%%%%%%%%%%%%%%%%%%%%%%%%%%%%%%%%%%%%%%%%%%%%
\subsection{Domain Scaling}
\label{sec:domainScaling}
%%%%%%%%%%%%%%%%%%%%%%%%%%%%%%%%%%%%%%%%%%%%%%%%%%%%%%%%%%%%%%%%%%%%%%%%%
Since Smolyak's grid is constructed in a $d$-dimensional cube, $[-1,1]^d$, a transformation of the grid needs to take place. It is a distinct feature compared to the SC method for low dimensions where the grid is built from the optimal points based on the quadrature points (see Section~\ref{sec:2_1}). Each dimension, $i$, is stretched with the following transformation:
\begin{equation*}
x_{i} = \frac12u_i (ub_i-lb_i) + lb_i+ \frac12(ub_i-lb_i),
\end{equation*}
where $u_i$ is the initial grid point in $[-1,1]$ and $lb_i$ and $ub_i$ are the lower and upper bound, respectively. Both quantiles, $lb_i$ and $ub_i$, can be computed either based on Monte Carlo paths or by using analytical properties of the underlying distribution of $X(t).$
Using Monte Carlo paths for a process $X(t),$ one can either obtain $lb_i=\min X(T_k)$ and $lb_i=\max X(T_k)$ at the exposure date $T_k$ or compute quantiles at a certain level $\alpha$:
\begin{eqnarray}
\label{eqn:lb}
lb_i&=&\min\{y \in\R:F_{X(T_k)}(y)\geq 1-\alpha\},\\
\label{eqn:ub}
ub_i&=&\min\{y \in\R:F_{X(T_k)}(y)\geq \alpha\}.
\end{eqnarray}

Figure~\ref{fig:GridStretch} illustrates the stretched grid using~(\ref{eqn:lb}) and~(\ref{eqn:ub}) applied to a 2D case. The figure shows the grid stretched for correlated variables $X_1$ and $X_2$. We have, for $\alpha=0.95$, obtained the following lower and upper bounds: $lb_1=-0.003$, $lb_2=1.055$, $ub_1=0.023$ and $ub_2=1.356.$ A similar effect is achieved for other processes in the underlying system of SDEs.
\begin{figure}[h!]
  \centering
    \includegraphics[width=0.45\textwidth]{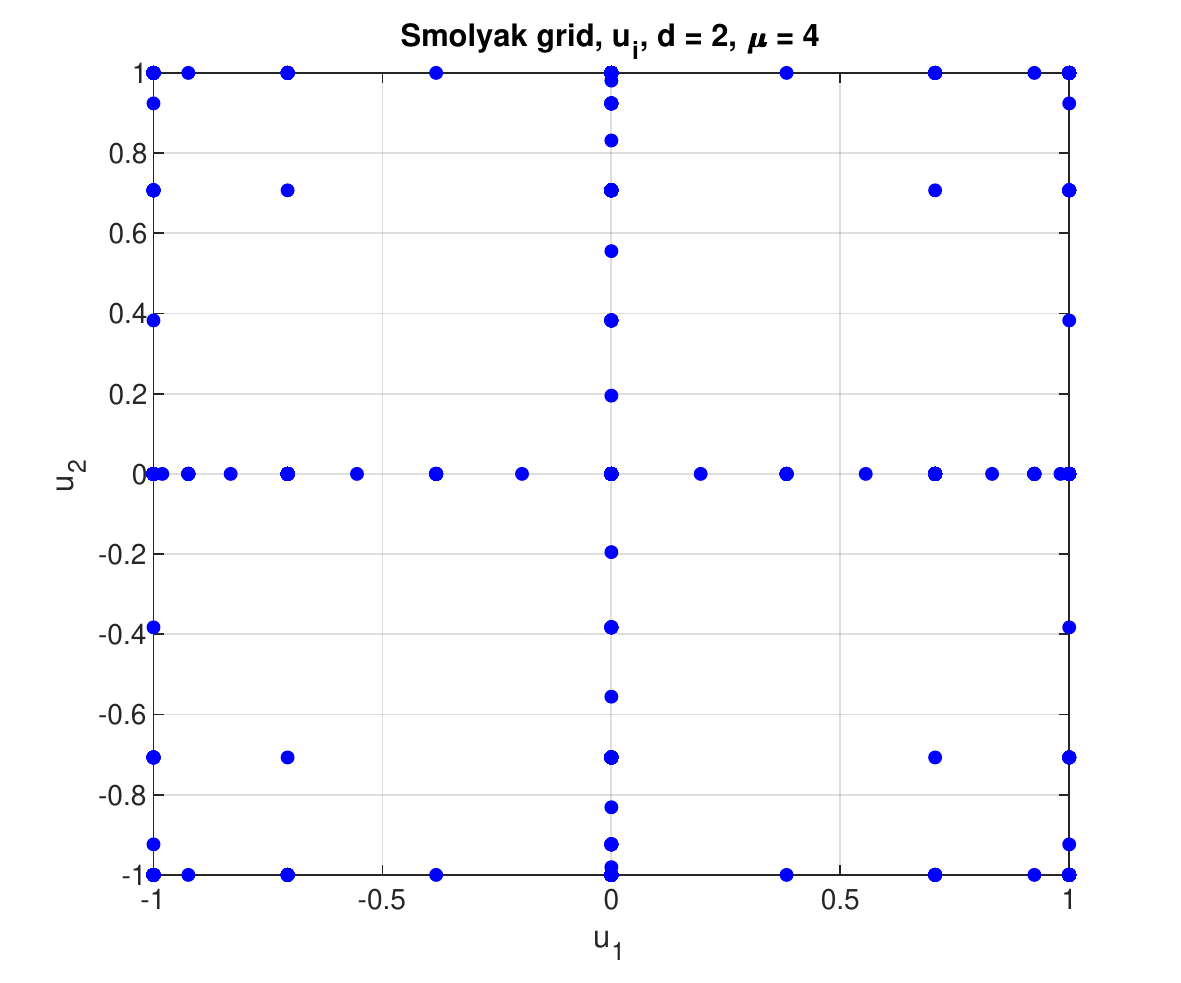}
    \includegraphics[width=0.45\textwidth]{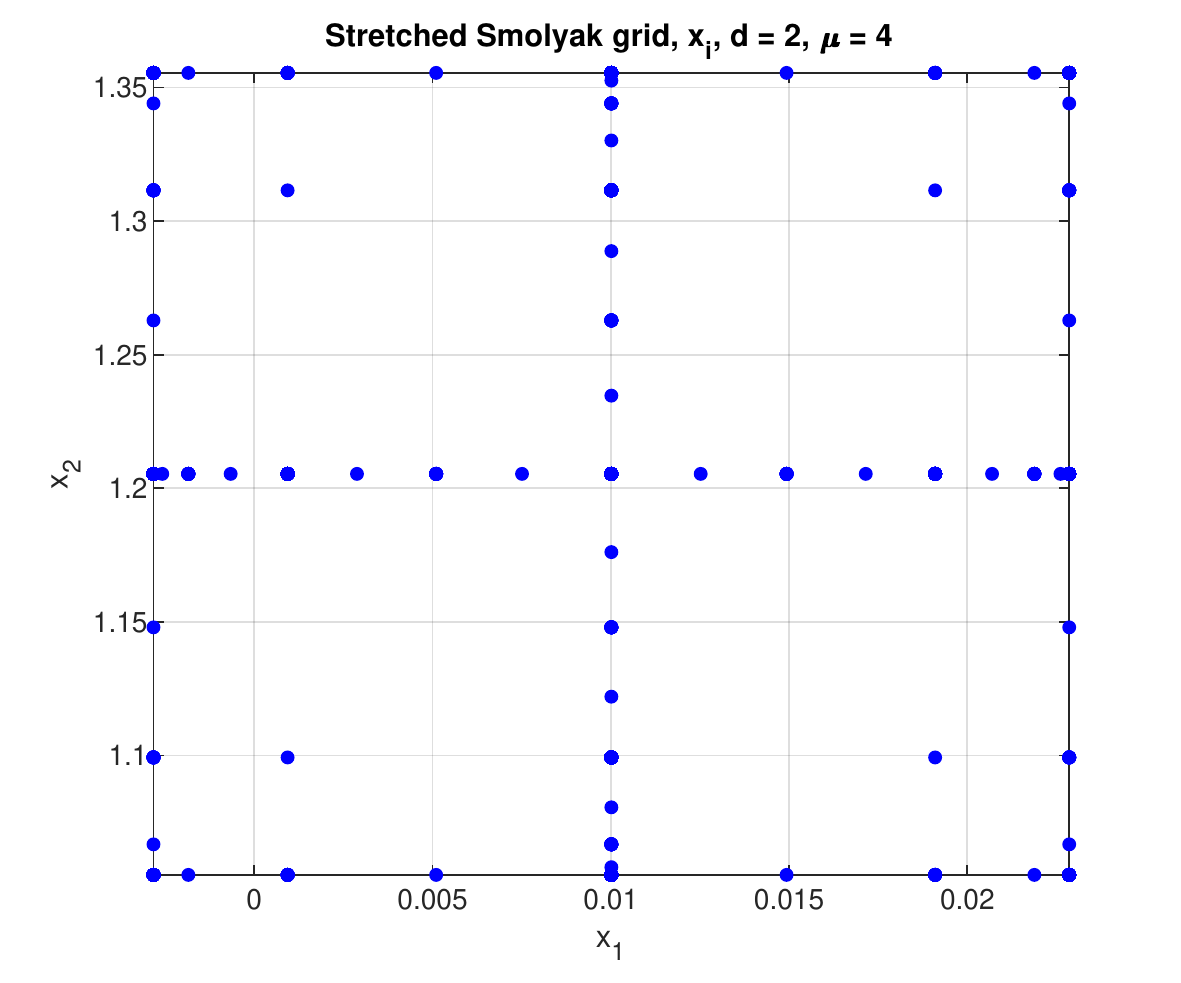}\\
      \caption{Smolyak's grid stretching based on quantiles. LHS: original grid. RHS: stretched grid.}
      \label{fig:GridStretch}
\end{figure}
%%%%%%%%%%%%%%%%%%%%%%%%%%%%%%%%%%%%%%%%%%%%%%%%%%%%%%%%%%%%%%%%%%%%%%%%%
\subsection{Sub-Optimal Collocation Points}
%%%%%%%%%%%%%%%%%%%%%%%%%%%%%%%%%%%%%%%%%%%%%%%%%%%%%%%%%%%%%%%%%%%%%%%%%
Although the collocation points under the Smolyak's algorithm are given and do not depend on the properties of the underlying processes~\footnote{Except for the domains, as discussed in Section~\ref{sec:domainScaling}.}, this is not the case for the {\it standard} SC method, described in Section~\ref{sec:2_1}.
These collocation points are determined based on the moments of the underlying stochastic model. This article focused on the most commonly used process for assets. We dealt with either normal or lognormal distributions; however, when the underlying processes are complex, one would need to compute  moments of such a process to establish the corresponding collocation points. This could be troublesome for two reasons: the moments are not always available in closed form, and moments obtained from Monte Carlo simulation may give rise to inaccurate results (the collocation method requires an LU decomposition of the so-called Gram matrix constructed by $n_i^2$ moments, $n_i$ being the number of collocation points in the $i$'th dimension).

To address the issue described above, one can calculate the optimal collocation points based on the normally distributed kernel process, i.e., for a random variable $\xi$ for which the moments are not available, the collocation points are computed by:

\[\xi_i = F^{-1}_\xi(F_{\mathcal{N}(0,1)}(x_i^{\mathcal{N}(0,1)})).\]

If we consider the CDF, $F_\xi(\cdot)$, of a continuous random variable $\xi$, the mapping $y=F_\xi(x)$ is bijective and $F_\xi(x)$ is strictly increasing, so is $F^{-1}_\xi(y)$. This implies that the argument $x$ can be obtained by the inverse interpolation of $F_\xi(y)$ against $y$, which can be done at essentially {\it no cost}. When the random variable $\xi$ is obtained from a Monte Carlo simulation, $F_\xi(\cdot)$ can be estimated from the empirical cumulative distribution function.
%%%%%%%%%%%%%%%%%%%%%%%%%%%%%%%%%%%%%%%%%%%%%%%%%%%%%%%%%%%%%%%%%%%%%%%%%
\subsection{Anisotropic and Adaptive Grids}
%%%%%%%%%%%%%%%%%%%%%%%%%%%%%%%%%%%%%%%%%%%%%%%%%%%%%%%%%%%%%%%%%%%%%%%%%
When dealing with sparse grids and valuation of portfolios, each of the dimensions is treated equally (the same number of collocation points in each dimension). This may not be desired, as some of the risk factors may be dominant in a portfolio risk profile. Anisotropic sparse grids address this problem and allow for unequal distribution of the collocation points over different dimensions. Discussions on this subject are well-covered in~\cite{Judd1998}.

Now, let us look at another aspect of the grid construction process.
Figures~\ref{fig:2d}, \ref{fig:3d}, and~\ref{fig:pathsValue}, show that the SC method generates a rectangular grid of collocation points. This may be sub-optimal, especially for correlated processes where the realizations would cluster. In such a multidimensional case where the underlying model consists of $d$ correlated risk factors, we use a multidimensional version of the normal collocation points. Thus, the grid of points would be determined by a multidimensional Gaussian copula estimated based on the risk factors ${\bf X}(t)$.

%For a standard normal multidimensional case the corresponding collocation points are given by:
%\begin{eqnarray}

%x_{1,i} &=&x_i^{\mathcal{N}(0,1)}\\

%x_{2,j}|x_{1,i} &=&\rho_{1,2}x_{1,i} +\sqrt{1-\rho_{1,2}^2} x_{2,j}^{\mathcal{N}(0,1)}\\

%x_{3,k}|x_{1,i},x_{2,j} &=&\rho_{1,3}x_{1,i} + \frac{\rho_{2,3}-\rho_{1,2}\rho_{1,3}}{\sqrt{1-\rho_{1,2}^2}}x_{2,j}+...

%\end{eqnarray}

%%\begin{eqnarray}
%%\nonumber
%%\hat x_{1,i} &=&{\bf L}_{1,1} x_i^{\mathcal{N}(0,1)}\\\nonumber
%%\hat x_{2,j}|\hat x_{1,i} &=&{\bf L}_{2,1}\hat x_{1,i} +{\bf L}_{2,2} x_{2,j}^{\mathcal{N}(0,1)}\\
%%\label{eqn:collNormal}
%%\hat x_{3,k}|\hat x_{1,i},\hat x_{2,j} &=&{\bf L}_{3,1}\hat x_{1,i} + {\bf L}_{3,2}\hat x_{2,j}+{\bf L}_{3,3} x_{3,k}^{\mathcal{N}(0,1)}\\\nonumber
%%\hat x_{4,l}|\hat x_{1,i},\hat x_{2,j},\hat x_{3,k} &=&{\bf L}_{4,1}\hat x_{1,i} + {\bf L}_{4,2}\hat x_{2,j}+{\bf L}_{4,3}\hat x_{3,k}+{\bf L}_{4,4}x_{4,l}^{\mathcal{N}(0,1)}\\\nonumber
%%\dots&=&\dots
%%\end{eqnarray}

%\begin{eqnarray}
%\nonumber
%\hat x_{1,i_1} &=&{\bf L}_{1,1} x_i^{\mathcal{N}(0,1)}\\\nonumber
%\hat x_{2,i_2}|\hat x_{1,i_1} &=&{\bf L}_{2,1}\hat x_{1,i_1} +{\bf L}_{2,2} x_{2,i_2}^{\mathcal{N}(0,1)}\\
%\label{eqn:collNormal}
%\hat x_{3,i_3}|\hat x_{1,i_1},\hat x_{2,i_2} &=&{\bf L}_{3,1}\hat x_{1,i_1} + {\bf L}_{3,2}\hat x_{2,i_2}+{\bf L}_{3,3} x_{3,i_3}^{\mathcal{N}(0,1)}\\\nonumber
%\hat x_{4,i_4}|\hat x_{1,i_1},\hat x_{2,i_2},\hat x_{3,i_3} &=&{\bf L}_{4,1}\hat x_{1,i_1} + {\bf L}_{4,2}\hat x_{2,i_2}+{\bf L}_{4,3}\hat x_{3,i_3}+{\bf L}_{4,4}x_{4,i_4}^{\mathcal{N}(0,1)}\\\nonumber
%\dots&=&\dots
%\end{eqnarray}

For $d$ risk factors, ${\bf X}(t)=[X_1(t),X_2(t),\dots,X_d(t)]^\T$, with their corresponding number of collocation points $n_1,n_2,\dots,n_d$, we have:
\begin{eqnarray}
\hat x_{i_1} &=&{\bf L}_{1,1} x_{i_1}^{\mathcal{N}(0,1)},\nonumber\\
\hat x_{i_2}|\hat x_{i_1} &=&{\bf L}_{2,1}\hat x_{i_1} +{\bf L}_{2,2} x_{i_2}^{\mathcal{N}(0,1)},\nonumber\\
\hat x_{i_3}|\hat x_{i_1},\hat x_{i_2} &=&{\bf L}_{3,1}\hat x_{i_1} + {\bf L}_{3,2}\hat x_{i_2}+{\bf L}_{3,3} x_{i_3}^{\mathcal{N}(0,1)},
\label{eqn:collNormal}
\end{eqnarray}
where $x_{i_j}^{\mathcal{N}(0,1)}$ indicates the $j$'th collocation point from a standard normal and where ${\bf L}_{i,j}$ is the $(i,j)$'th element of the Cholesky decomposition, ${\bf C}={\bf L}{\bf L}^\T$ with ${\bf C}$ being the correlation matrix.
In order to compute all the collocation points, we require the following inversions:
\begin{eqnarray}
\nonumber
x_{i_1}&=&F^{-1}_{X_1}(F_{\mathcal{N}(0,1)}(\hat x_{i_1})),\\\nonumber
x_{i_2}|x_{i_1}&=&F^{-1}_{X_2}(F_{\mathcal{N}(0,1)}(\hat x_{i_2}|\hat x_{i_1})),\\
x_{i_3}|x_{i_1},x_{i_2}&=&F^{-1}_{X_3}(F_{\mathcal{N}(0,1)}(\hat x_{i_3}|\hat x_{i_1},\hat x_{i_2})).
\label{eqn:collTarget}
\end{eqnarray}
In Figure~\ref{fig:2Dexample}, an illustrative example of the grid distribution in 2D is presented. The LHS figure shows the grid distribution given in~(\ref{eqn:collNormal}), while the RHS figure corresponds to the actual grid points from Equation~(\ref{eqn:collTarget}).

\begin{figure}[h!]
  \centering
    \includegraphics[width=0.49\textwidth]{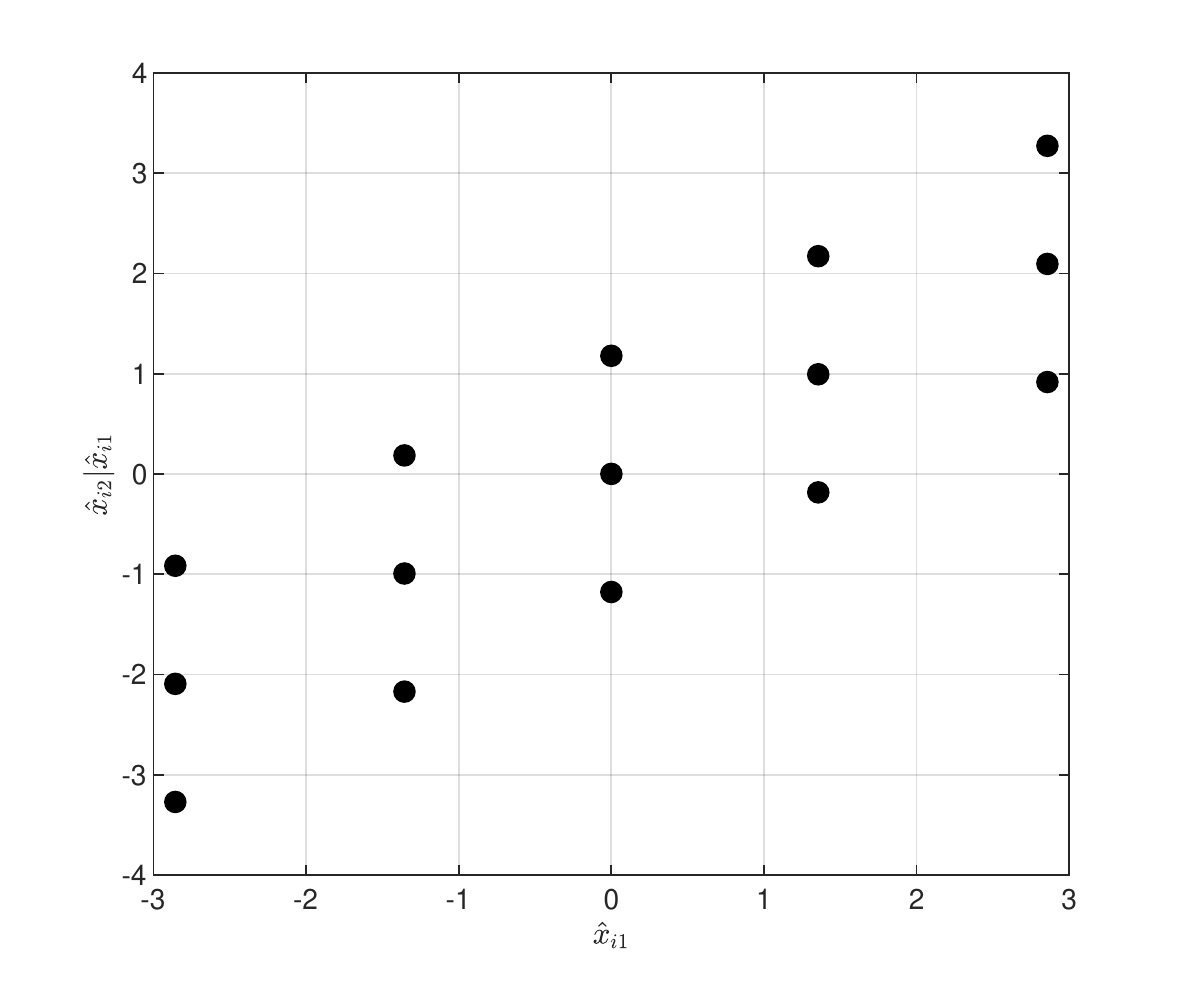}
    \includegraphics[width=0.49\textwidth]{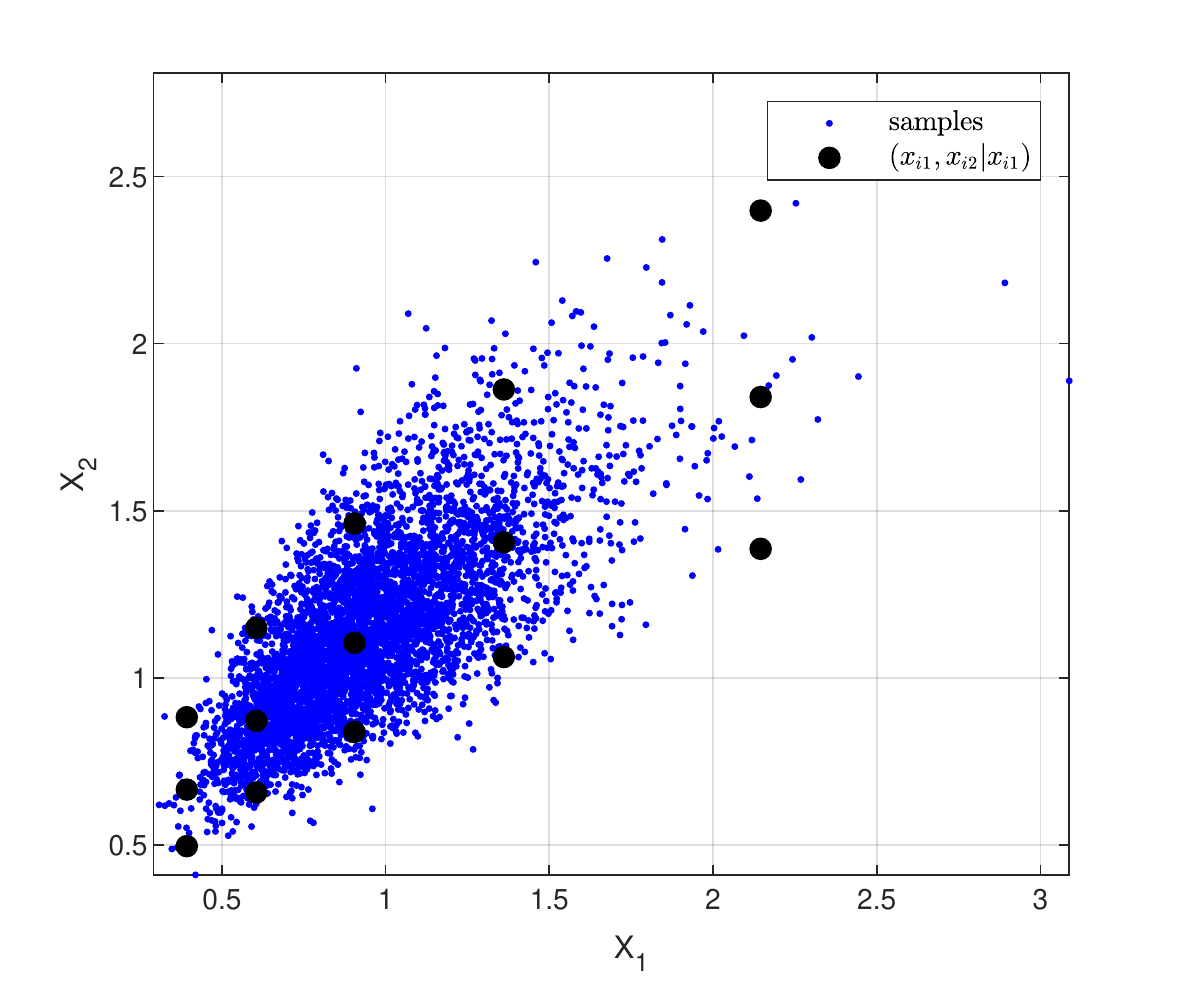}
      \caption{LHS: 2D collocation points $\hat x_{i1}$ and $\hat x_{i2}|\hat x_{i1}$ for $\rho=0.6.$ RHS: Scatter plot for two lognormal random variables $X_1\sim \e^{Z_1}$ and $X_2\sim\e^{Z_2}$ where $Z_1\sim\mathcal{N}(-0.1,0.3)$, $Z_2\sim\mathcal{N}(0.2,0.2)$, $N_1=5$, $N_2=3$ and $\rho_{Z_1,Z_2}=0.6$ and the corresponding collocation points $x_{i1}$ and $x_{i2}|x_{i1}.$}
      \label{fig:2Dexample}
\end{figure}

Given the representation above the approximation for the portfolio value is expressed by:
\begin{eqnarray*}
\widetilde g(V(T_k,\{{\bf X}\}_{i_1,\dots,i_d}),{\bf X}(T_k))=\sum_{i_1=1}^{n_1}\cdots \sum_{i_d=1}^{n_d}V(T_i,x_{i_1};x_{i_2}|x_{i_1};\dots;x_{i_d}|x_{i_1},\dots,x_{i_{d-1}})\hat\psi({\bf X}(T_k)),
\end{eqnarray*}
where
$\hat\psi({\bf X}(T_k))$ is defined in~(\ref{eqn:LagrangePortfoli}).

%Given the representation above we conclude that for $d$ stochastic risk factors $X_i$, $i=1,\dots,d$ we need to determine $N^d$ grid points and therefore $N^d$ inversions of the underlying CDFs. From a numerical perspective, these inversions are cheap as they purely depend on the simulated Monte Carlo paths and not on the portfolio itself.

%%%%%%%%%%%%%%%%%%%%%%%%%%%%%%%%%%%%%%%%%%%%%%%%%%%%%%%%%%%%%%%%%%%%%%%%%%%%%
%\subsection{Numerical Experiments: Large Portfolios Traded in Different Currencies}
%%%%%%%%%%%%%%%%%%%%%%%%%%%%%%%%%%%%%%%%%%%%%%%%%%%%%%%%%%%%%%%%%%%%%%%%%%%%%
%%%%%%%%%%%%%%%%%%%%%%%%%%%%%%%%%%%%%%%%%%%%%%%%%%%%%%%%%%%%%%%%%%%%%%%%%%%%%
\section{Error Analysis and Convergence}
%%%%%%%%%%%%%%%%%%%%%%%%%%%%%%%%%%%%%%%%%%%%%%%%%%%%%%%%%%%%%%%%%%%%%%%%%%%%%
\label{sec:Error}
This section focuses on the error analysis of the approximating function $\widetilde g(\cdot)$. In particular, we analyze the impact of the approximation on a portfolio containing interest rate swaps, possibly in different currencies and driven by the BSHW model defined in~(\ref{eqn:FXmulti}). Such a portfolio is of particular interest in the industry.
%Let us firstly investigate the distributional properties of such a portfolio.
%The result below gives us an important relation between the portfolio studied in this article and a lognormal distribution.
The result below gives us essential insight into the distributional properties of such a portfolio.
\begin{lem}[Distribution of a portfolio of multi-currency swaps]
\label{lem:1}
Under the multi-currency BSHW model, defined in~(\ref{eqn:FXmulti}), a portfolio, at the exposure date $T_k,$ consisting of interest rate swaps in $d_c$ foreign currencies, is distributed as a linear combination of lognormally distributed random variables
\begin{equation}
{\bf V}(T_k,{\bf X}(T_k))\sim \sum_{\ell\in\Omega} c_\ell\e^{{Z}_\ell(T_k)},
\end{equation} with a constant $c_\ell\in\R$, where ${Z}_\ell(T_k)~\sim \mathcal{N}({ m_\ell},{ \Sigma}^2_\ell)$, for a certain constant mean, $m_\ell$, variance, ${\Sigma}^2_\ell$, and where $\Omega$ represents the set of all ZCBs under all underlying currencies.
\end{lem}
\begin{proof}
Under the Hull-White model, every zero-coupon bond is log-normally distributed, i.e.,
\begin{eqnarray}
P(t,T)=\exp\left({A(t,T)+B(t,T)r(t)}\right)\sim \log\mathcal{N}\left(\mu_{P},\sigma_{P}^2\right),
\end{eqnarray}
with a certain mean parameter, $\mu_{P}$, variance, $\sigma^2_{P}$, where $r(t)$ is the short-rate at time $t$ and $A(t,T)$ and $B(t,T)$ are time-dependent functions defined in~(\ref{eqn:A}) and~(\ref{eqn:B}).
%where $\mu_P=A(t,T)+B(t,T)\E[r(t)]$ and %$\sigma_P^2=\left(\e^{B^2(t,T)\Var[r(t)]}-1\right)\e^{2(A(t,T)+\E[r(t)]+B^2(t,T)\Var[r(t)]}.$
As a result, an interest rate swap, as presented in~(\ref{swapSingleCCy}) is, after some simplifications, given as a linear combination of zero-coupon bonds,
\begin{eqnarray}
\label{eqn:swap1}
V(t,r(t))=N\left(P(t,T_{j})-P(t,T_{M})\right)-NK\sum_{k=j+1}^{M}\tau_kP(t,T_k),
\end{eqnarray}
for a certain strike, $K$, a notional amount, $N$, and swap payment dates $T_{j+1},\dots, T_M$. The interest rates swap defined in~(\ref{eqn:swap1}) can be recognized as a linear combination of lognormally distributed random variables:
\begin{eqnarray}
V(t,r(t))=\sum_{k=j}^Ma_k\e^{\mu_k+\sigma_kr_b(T_k)}=:\sum_{k=j}^Ma_kY_k,
\end{eqnarray}
for some constant parameter $a_k$ and where $Y_k\sim \log\mathcal{N}\left(\mu_k,\sigma_k^2\right)$. The same analogy holds for a portfolio consisting of interest rate swaps, even in the case of a portfolio involving multiple currencies, driven by the SDEs in~(\ref{eqn:FXmulti}),
\begin{eqnarray}
\label{ref:porftolioMultiLognormal}
V(t,{\bf X}(t)) &=&\sum_{i=1}^{M_b}V_i^b(t,r_b(t)) + \sum_{k=1}^{d_c}y_k^b(t)\sum_{i=1}^{M_k} V_i^k(t,r_k(t))\stackrel{\d}{=}\sum_{\ell\in\Omega} c_\ell\e^{{Z}_\ell(T_k)},%&=&\sum_{i=1}^{M_b}\sum_{k=1}^{M_i}a_kY_k + %\sum_{k=1}^{d_c}y_k^b(t)\sum_{i=1}^{M_k} \sum_{k=1}^{M_i}b_kZ_k\nonumber\\
%&=:&\sum_l c_l\xi_l,
\end{eqnarray}
where $c_l$ is a constant in $\R$ and $Z_l(T_k)\sim \log\mathcal{N}\left(m_\ell,\Sigma_\ell^2\right)$ for some $m_\ell$ and $\Sigma_\ell^2.$
\end{proof}

From the representation above, we conclude that under the BSHW model in~(\ref{eqn:FXmulti}), a portfolio consisting of swaps under different currencies can be represented as a linear combination of correlated log-normally distributed random variables. This also means that to estimate the error, we need to look at the SC method and its quality in approximating a linear combination of lognormals in~(\ref{ref:porftolioMultiLognormal}). Unfortunately, even the problem of the sum of lognormals is unresolved. We refer to an overview of attempts and approximations in~\cite{Dufresne2008SUMSOL}. Let us now proceed with the error estimates for the 1D and multi-D cases.

%Motivated by the research of Fenton~\cite{Fenton1960} and Schwartz and Yeh~\cite{Schwartz1982} who shown that a sum of lognormally distributed random variables can be well approximated by another log-normal distribution we start with the error estimate for the 1D case.

%The hart of the SC method lies in projecting the variable of interest on an approximating polynomial that mimics the variable of interest in such a way that the CDFs of both agree at the specified collocation points, $x_i$:
%\begin{eqnarray}
%\sum_l c_l\xi_l\sim a_0+a_1X +a_2X^2+a_3X^3+\dots,
%\end{eqnarray}
%where the coefficients $a_0,a_1,a_2,\dots$ are inferred from

In the 1D or 2D instances in which the collocation points $x_i$ correspond to the zeros of an orthogonal polynomial, the following equality, because of the connection to Gauss quadrature, in $L^2$, holds:
\begin{eqnarray}
\label{eqn:error_Epsilon}
\int_\R\left(V(T_k,x)-\widetilde g(T_k,x)\right)^2f_X(x)\dx =\sum_{i=1}^{n_1}\left(V(T_k,x_i)-\widetilde g(T_k,x_i)\right)^2\omega_i +\epsilon_{n_1} = \epsilon_{n_1},
\end{eqnarray}
where $V(T_k,x)-\widetilde g(T_k,x)$ represents the difference between the exact portfolio and the SC approximated function, $f_X(x)$ is the weight function, and where $\omega_i$, for $i>0$, are the quadrature points.
When the Gauss-Hermite quadrature is used with $n_1$ collocation points, the approximation error of the CDF can be estimated as,
\begin{eqnarray*}
\epsilon_{n_1}=\frac{n_1!\sqrt{\pi}}{2^{n_1}}\frac{\Psi^{(2n_1)}(\hat\xi)}{(2n_1)!},
\end{eqnarray*}
where
\begin{eqnarray}
\label{eqn:psi_x}
\Psi(x):=\left(V(T_k,x)-\widetilde g(T_k,x)\right)^2=\left(\frac{1}{n_1!}\frac{\d^{n_1}g(T_k,x)}{\d x^{n_1}}\big|_{x=\hat\xi}\prod_{i=1}^{n_1}(x-x_i)\right)^2.
\end{eqnarray}

Since the approximations proposed in this article are also used to estimate PFEs at different significance levels, we also need to assess the so-called {\it tail risk}. Since functions
$V(T_k,x),$ and $\widetilde g(T_k,x)$ agree at the collocation points, the upper bound for a risk factor $X(T_k)$ is given by~\cite{grzelak2015stochastic}:
\begin{eqnarray}
\E\left[(V(T_k,X(T_k))-\widetilde g(T_k,X(T_k)))^2|X>x_*\right]&\leq&\frac{1}{\P[X(T_k)>x_*]}\frac{n_1!\sqrt{\pi}}{2^{n_1}}\frac{\Psi^{(2{n_1})}(\hat\xi)}{(2{n_1})!},
\end{eqnarray}
with $\Psi(x)$ defined in~(\ref{eqn:psi_x}).
In a special case, when the risk factor is normally distributed,  $X(T_k)\sim\mathcal{N}(0,1),$ and by performing integration
by parts twice, one can show that for $x_*>0$:
\begin{eqnarray}
\P[X(T_k)>x_*]\geq
\frac{1}{\sqrt{2\pi}}\e^{-x_*^2/2}\left(\frac{1}{x_*}-\frac{1}{x_*^3}\right),
\end{eqnarray}
so that the error for $x_*>1$ is bounded by:
\begin{eqnarray*}
\E\left[(V(T_k,X(T_k))-\widetilde g(T_k,X(T_k)))^2|X(T_k)>x_*\right]\leq\pi\sqrt{2}\e^{x_*^2/2}\frac{x^3_*}{x^2_*-1}\frac{n_1!}{2^{n_1}(2n_1)!}\Psi^{(2n_1)}(\hat\xi).
\end{eqnarray*}
The factorial $(2n_1)!$ in the expression above
is dominant and thus, for a smooth function $\Psi(\xi)$, we obtain:
\[\lim_{n_1\rightarrow\infty}\E\left[(V(T_k,X(T_k))-\widetilde g(T_k,X(T_k)))^2|X>x_*\right]=0.\]

%??? TO BE FINISHED: Chebyshev + Smolyak error
%Using Chebychev's inequality and the error estimate, $\epsilon_N$, from~(\ref{epsilon_v01}) we can also show that for a fixed number of collocation points, $N$, we have for $a\rightarrow\infty$:
%\begin{eqnarray}
%\P((Y-Y_N)^2\geq a)&\leq&
%\frac{1}{a}\E((Y-Y_N)^2)=\frac{\epsilon_N}{a}\rightarrow0.
%\end{eqnarray}

Now, let us investigate the error propagation under sparse grid methods. The quality of sparse grid methods comes with strict requirements on the smoothness of the high-D functions and bounded mixed derivatives, i.e., the convergence of the method is specified for a particular {\it regularity} parameter $r$, which is defined for function spaces as follows,
%\begin{equation}
%\label{eqn:regularity}
%H^r(\omega^d)=\Big\{f:\omega^d\rightarrow \R;\frac{\partial^\alpha %f(x)}{\partial x^\alpha}\;\;\text{exists and is bounded in}\;\; %\omega^d\;\;\text{for all}\;\;|\alpha|_\infty\leq r \Big\},
%\end{equation}
\begin{equation}
\label{eqn:regularity}
F_d^r:=\Big\{f:\omega^d\rightarrow \R;\frac{\partial^\alpha f({\bf x})}{\partial {\bf x}^\alpha}\;\;\text{exists and}\;\; \Big|\Big|\frac{\partial^\alpha f({\bf x})}{\partial {\bf x}^\alpha}\Big|\Big|_\infty \leq 1 ;\;\text{if }\;\; \alpha\leq r\Big\}.
\end{equation}

%and the corresponding norm:
%\begin{equation}
%\big|\big| f\big|\big|:=\max\left\{\Big|\Big|\frac{\partial^\alpha %f(x)}{\partial x^\alpha}  \Big|\Big|_\infty\;\;\text{for}\;\; %|\alpha|_\infty\leq r\right\}.
%\end{equation}
Then, according to~\cite{Barthelmann2000} for $V(T_k,{\bf x})\in F_d^r$ and a particular finite domain, $\omega$, the corresponding Smolyak quadrature rule of degree $\mu$ has the asymptotic convergence rate of
\begin{equation}
\label{eqn:error}
\sup_{V(T_k,{\bf x})\in F_d^r}\big|\big| V(T_k,{\bf x})-\widetilde g(T_k,{\bf x})\big|\big|_\infty\leq \frac{c_{d,r}}{n^{r}(d,\mu)}(\log n(d,\mu))^{(r+2)(d-1)+1},
\end{equation}
where $r$ is the regularity parameter defined in~(\ref{eqn:regularity}), $n(d,\mu)$ represents the total number of grid points used in the grid construction (see Table~\ref{Tab:diemsionality}), and $c_{d,r}$ is a constant that only depends on dimension, $d$, and the regularity parameter, $r$. The error bound in~(\ref{eqn:error}) shows the relation between the number of grid points and smoothness of the pricing function.

%https://link.springer.com/content/pdf/10.1023/A:1018977404843.pdf

The condition defined for the partial derivatives in~(\ref{eqn:regularity}) requires that a continuous function $f$ is contracting and is associated with a satisfied Lipschitz condition. Lemma~\ref{eqn:regularity} shows that the value of a multi-currency portfolio with interest rate swaps driven by the BSHW model in~(\ref{eqn:FXmulti}) is of an exponential form. The exponential function, $\e^{\bf x}$, is not globally Lipschitz continuous as it becomes arbitrarily steep for ${\bf x}\rightarrow+\infty$ or ${\bf x}\rightarrow-\infty$. However, any continuously differentiable function is locally Lipschitz, as continuous functions are locally bounded, therefore its gradient is locally bounded.

In Figure~\ref{fig:Cont}, the base currency portfolio, defined in Section~\ref{sec:multid}, and its derivatives to the underlying risk factor are presented. We observe a rather smooth profile of the value for different exposure dates, $T_k$. However, the portfolio's value increases exponentially for negative interest rates, especially for $T_k\rightarrow0$. Therefore, one would expect more grid points to be needed to compensate for the exponential growth of the portfolio. However, from a practical perspective, such a scenario is unrealistic as it would imply highly negative rates for short maturities. On the other hand, in the computation of xVA, exposures are multiplied with default probabilities that for $T_k\rightarrow0$ are very low; therefore, the overall impact, even for significant error, is minimal.

%On the other hand, high quality of the accuracy for exposures for short %maturities is, from practical prospective, less of importance i.e. in the xVA %calculations exposures are

\begin{figure}[h!]
  \centering
    \includegraphics[width=0.49\textwidth]{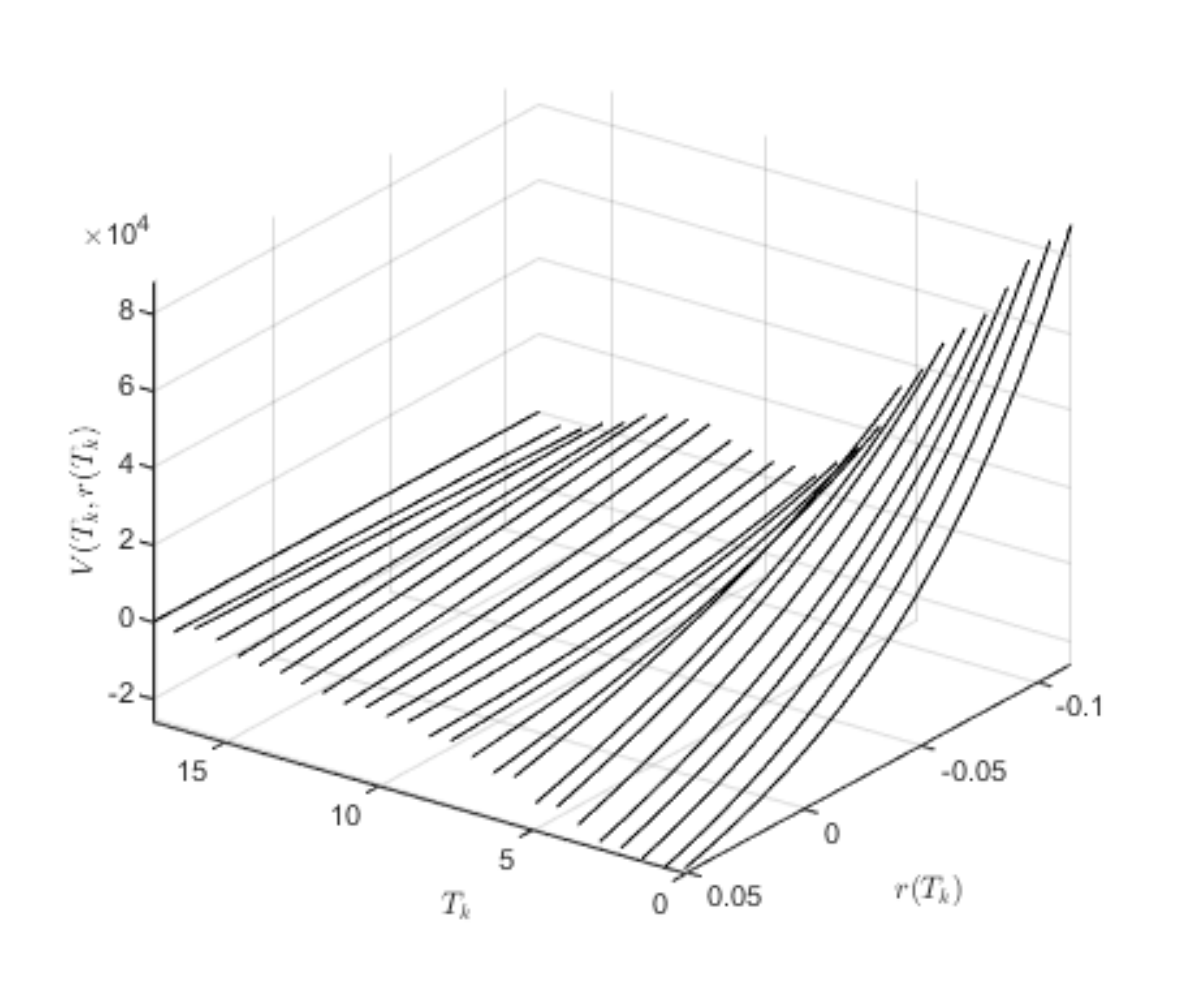}
    \includegraphics[width=0.49\textwidth]{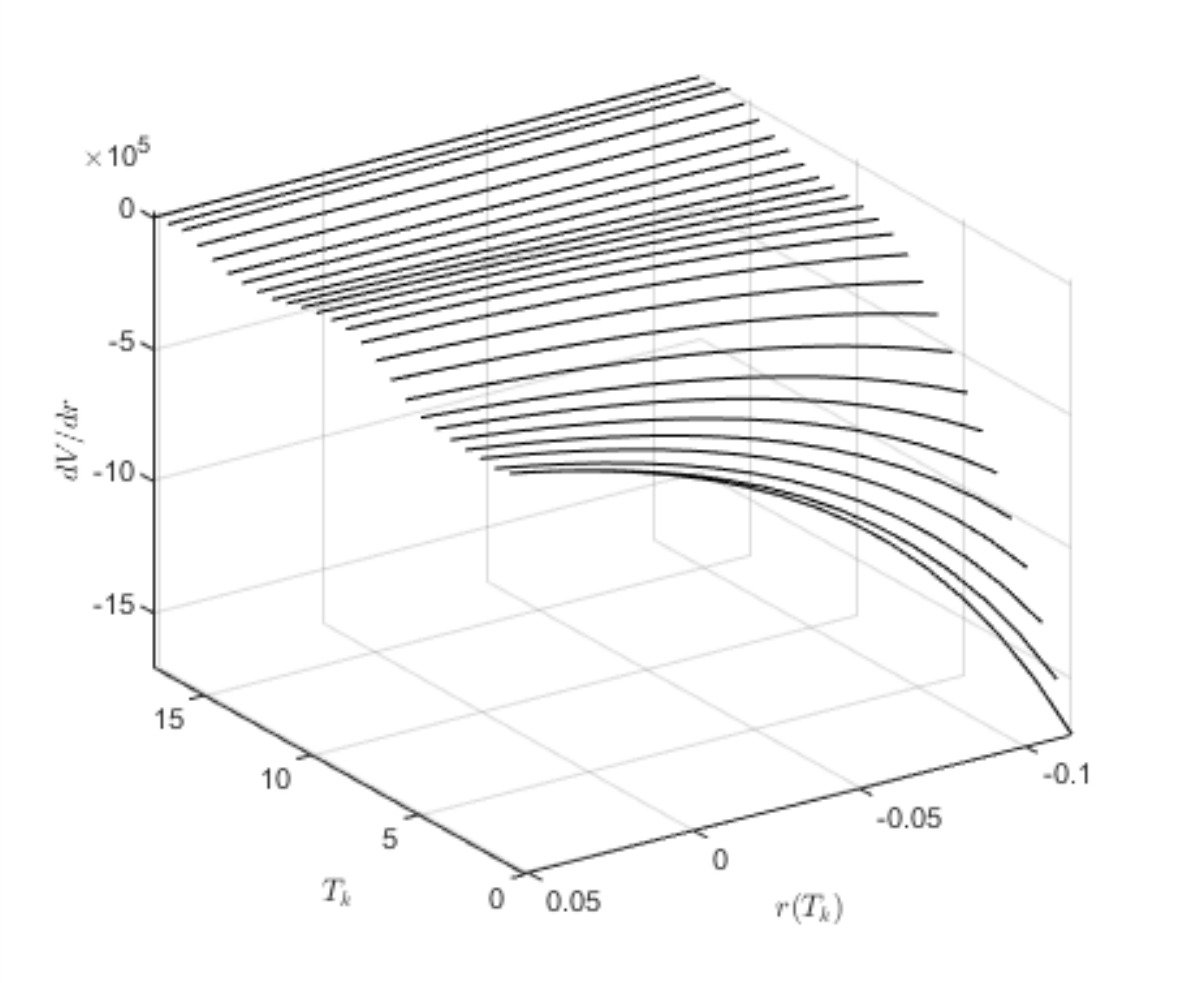}
      \caption{LHS: The portfolio of swaps, $\overline V^b(T_k,r_b(T_k))$ for $k=1,\dots,N_T$, in the base currency as a function of exposure date, $T_k$, defined in~(\ref{ref:porftolioMulti2})  and interest rate, $r(T_k)$. RHS: Derivative of the portfolio with respect to interest rates.}
      \label{fig:Cont}
\end{figure}

%%%%%%%%%%%%%%%%%%%%%%%%%%%%%%%%%%%%%%%%%%%%%%%%%%%%%%%%%
\subsection{Convergence with a Numerical Experiment}
%%%%%%%%%%%%%%%%%%%%%%%%%%%%%%%%%%%%%%%%%%%%%%%%%%%%%%%%%
This section analyzes the convergence of the SC method depending on the number of collocation points, $n_1$ (1D case), and the level parameter, $\mu$ (7D case), based on the portfolio cases discussed in Section~\ref{sec:3}. The dependence of the parameter $\mu$ on the number of the collocation points is presented in Table~\ref{Tab:diemsionality}.

Figure~\ref{fig:Error1DCase} illustrates the convergence results for the portfolio based on the Lagrange interpolation with the collocation points determined based on the quadrature rule. We report excellent results: already $n_1=4$ guarantees high accuracy. In Figure~\ref{fig:ErrornDCase}, the multi-D case is considered. As expected, the fastest convergence is obtained for $\text{EE}$, while a higher number of grid points is necessary for a satisfactory convergence in tails of the distribution, $\text{PFEs}$.

\begin{figure}[h!]
  \centering
    \includegraphics[width=0.49\textwidth]{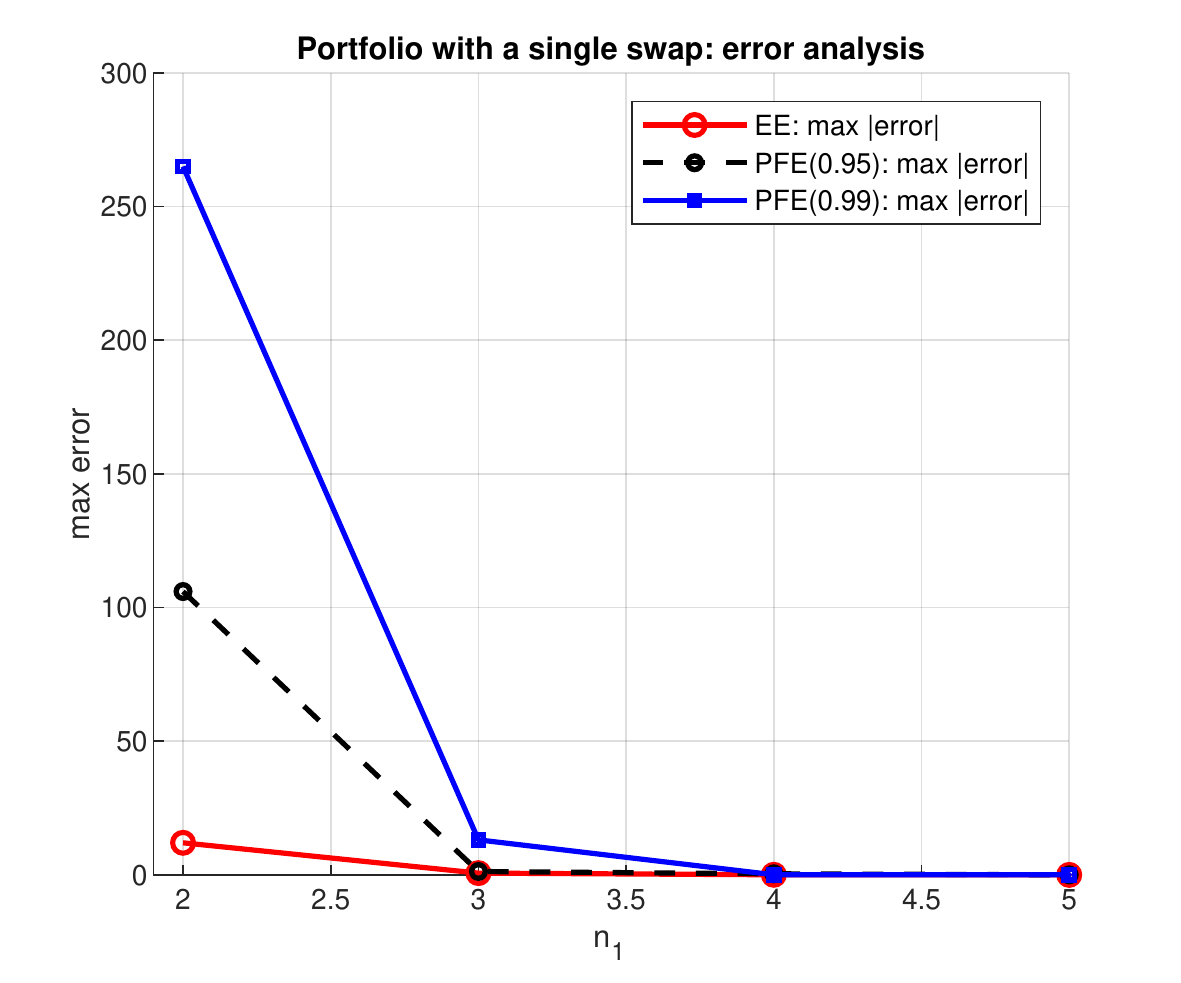}
    \includegraphics[width=0.49\textwidth]{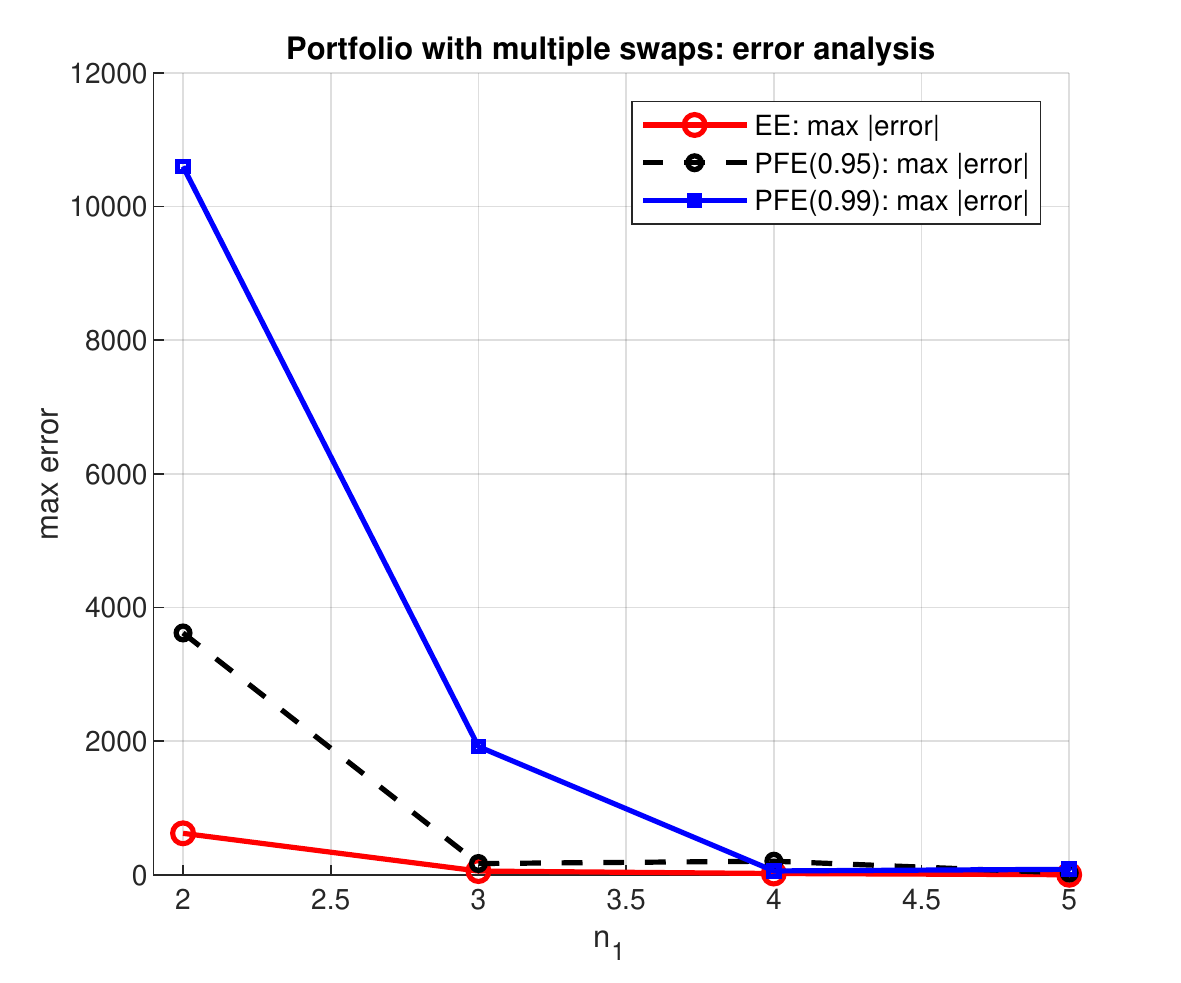}
      \caption{Error is defined as $\max_k |f(T_k)-\widetilde f(T_k)|$, where $f(T_k)$ is EE or PFE, and where $\widetilde f(T_k)$ corresponds to the SC approximation. The details regarding the portfolios under consideration are defined in Section~\ref{sec:1D}. LHS: single swap portfolio. RHS: multi-swap case.}
      \label{fig:Error1DCase}
\end{figure}

\begin{figure}[h!]
  \centering
    \includegraphics[width=0.49\textwidth]{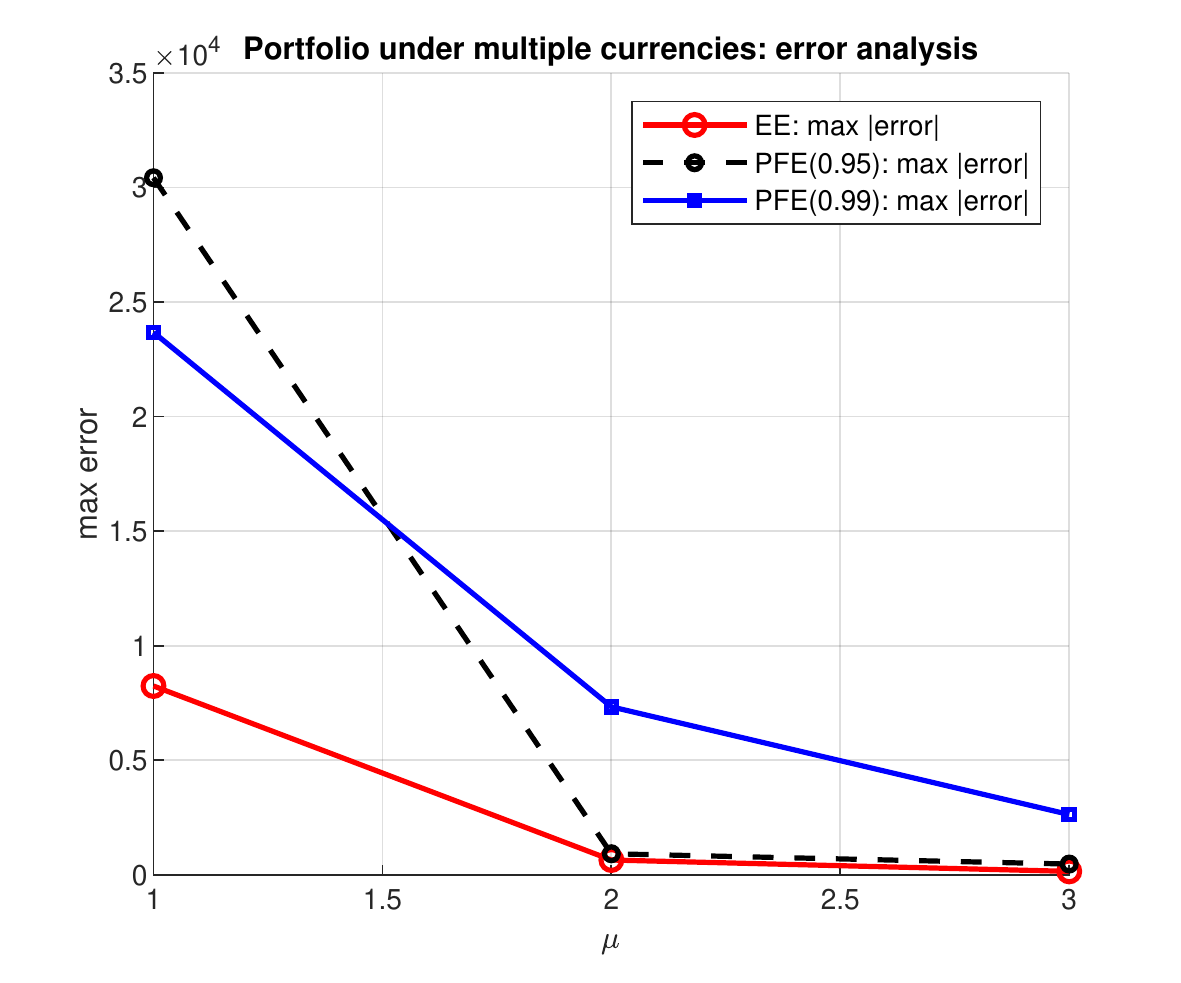}
    \includegraphics[width=0.49\textwidth]{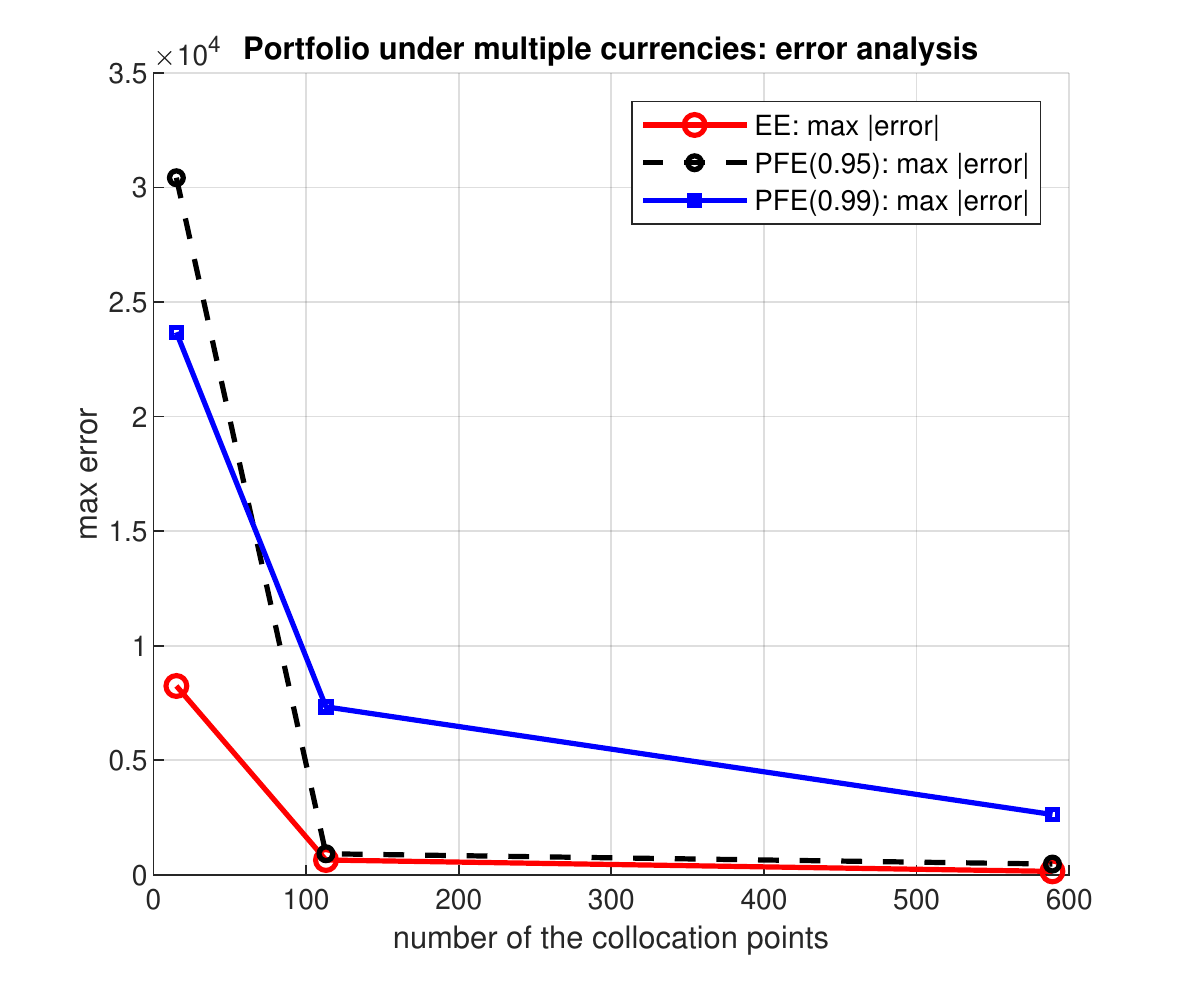}
      \caption{Error as a function of level, $\mu$, is defined as $\max_k |f(T_k)-\widetilde f(T_k)|$ where $f(T_k)$ is EE or PFE, and where $\widetilde f(T_k)$ corresponds to the SC approximation. The details regarding the portfolios under consideration are defined in Section~\ref{sec:multid}.}
      \label{fig:ErrornDCase}
\end{figure}

%\begin{table}[!h]
%\begin{center}
% \caption{Error analysis for a portfolio described in Section~\ref{sec:1D}.}
% \begin{tabular}{ c|c | c | c | c }
%  \hline
%   & total $\#$ of portfolio& EE   &PFE1&PFE2  \\
%   &evaluations  &&p=0.95&p=0.99       \\ \hline
%   $n_1=2$& 100 &???&??? \\
%   $n_1=3$&150& ??? & ??? \\
%   $n_1=4$&200& ??? & ??? \\
%   $n_1=5$&250& ??? & ??? \\
%   $n_1=6$&300& ??? & ???\\
%   $n_1=7$&350& ??? & ???\\
%   $n_1=8$&400& ??? & ???\\
%   $n_1=9$&450& ??? & ???\\
%   $n_1=10$&500& ??? & ???\\\hline
%
% \end{tabular}
% \label{table:1DCase}
% \end{center}
% \end{table}
%%%%%%%%%%%%%%%%%%%%%%%%%%%%%%%%%%%%%%%%%%%%%%%%%%%%%%%%%%%%%%%%%%%%%%%%%%%%%
%\subsection{Timing Results}
%%%%%%%%%%%%%%%%%%%%%%%%%%%%%%%%%%%%%%%%%%%%%%%%%%%%%%%%%%%%%%%%%%%%%%%%%%%%%
%time needed for grid construction\\
%time needed for portfolio evaluation\\\

%Phenomenal results...

%%%%%%%%%%%%%%%%%%%%%%%%%%%%%%%%%%%%%%%%%%%%%%%%%%%%%%%%%%%%%%%%%%%%%%%%%%%%%
\section{Conclusions}
\label{sec:conclusions}
%%%%%%%%%%%%%%%%%%%%%%%%%%%%%%%%%%%%%%%%%%%%%%%%%%%%%%%%%%%%%%%%%%%%%%%%%%%%%
In this article, we have discussed the SC method extended with the sparse grid algorithm of Smolyak and its application to efficient exposure computation in financial risk management. We can drastically reduce the number of portfolio evaluations for multi-currency portfolios. The proposed algorithm is beneficial for large portfolios involving many risk factors. The model can be easily generalized to any portfolio and size. Our numerical experiments have shown that for a realistic portfolio consisting of linear and non-linear derivatives, the expected reduction in the portfolio evaluations may exceed 6000 times, depending on the number of risk factors and required accuracy. We have discussed the convergence aspects, and several realistic examples were given.

\small
\bibliographystyle{plain}
%{ \bibliography{BibTexFile/bibtexfile}}
{ \bibliography{bibtexfile}}

\appendix
\section{The SC Algorithm for Exposure Computation}
\label{sec:appendix}
\begin{algorithm}
\caption{Construction and Evaluation of the Approximating Portfolio}\label{algo:SC}
\begin{algorithmic}[1]
\Procedure{SC}{}
\State Calibrate the underlying SDEs, ${\bf X}(t)=[X_1(t),\dots,X_d(t)]^\T$ , as defined in~(\ref{eqn:FXmulti}).
\State Simulate multi-dimensional process ${\bf X}(t)$ using Monte Carlo.
\For {$T_k$ with $k \in \{1,...,N_{T}\}$}
\State Given the simulated Monte Carlo paths determine, for each $X_i(T_k)$, the interpolation domain, as described in Section~\ref{sec:domainScaling}.
\State Build the SC grid using the Smolyak's algorithm, see Section~\ref{sec:2_1}.
\State Evaluate portfolio $V(T_k,V({\bf X}(T_k))$ at the grid points $\{X\}_{i_1,\dots,i_d}$ obtained in the previous step and obtain $\{V\}_{i_1,\dots,i_d}.$
\State Using Smolyak's interpolation build the approximating function,
$\widetilde g\big(\{V\}_{i_1,\dots,i_d},{\bf X}(T_k)\big)$, as described in~(\ref{eqn:proxySC}).
\State Evaluate function $\widetilde{g}(\{V\}_{i_1,\dots,i_d},{\bf X}(T_k))$ for all the Monte Carlo paths.
\State Compute EEs, PFEs etc. as defined in~(\ref{eqn:EE}) and~(\ref{eqn:PFE}).
%\For {$i \in \{1,...,M\}$}
%\State Obtain all the bond prices $P(T_i,T_j)$ with $j=i,i+1,...,M$ from %(\ref{eqn:bond_analytical}) or (\ref{eqn:BondCIR++}).
%\State Compute $L(T_{i-1};T_{i-1},T_i)$ with (\ref{eqn:ForwardLiborRate}).
%\State Calculate the new rate $S_{T_i,T_M}$ using (\ref{eqn:swap_rate}).
%\State Evaluate the difference $K - \kappa(T_i)$ .
%\State Estimate the prepayment via (\ref{eqn:RefinancingIncentiveFunctionalForm}).
%\State Get the correct value of the notional $N(T_i)$ considering %(\ref{eqn:NotionalDependencyPsiFunction}).
%\State Integrate numerically on the fine mesh to obtain $M(T_i)$.
\EndFor
%\State Acquire the price of the IAS for the simulation $k$ using (\ref{eqn:IAS})
%\EndFor
%\State Average on the simulations to get $V_{\text{IAS}}$
\EndProcedure
\end{algorithmic}
\end{algorithm}

\section{Multi-D case: Details Regarding Model Configuration}
\label{sec:appendix2}
\begin{equation}
{\bf C}=\left[\begin{array}{ccccccc}
1&0.5&0.5&0.65&0.7&0.75&0.8\\
0.5& 1& 0.45& 0.35& 0.5& 0.5& 0.6\\
         0.5& 0.45& 1& 0.5& 0.5& 0.5& 0.7\\
         0.65& 0.35& 0.5& 1& 0.5& 0.5& 0.5\\
         0.7& 0.5& 0.5& 0.5& 1& 0.5& 0.58\\
         0.75& 0.5& 0.5& 0.5& 0.5& 1& 0.55\\
         0.8& 0.6& 0.7& 0.5& 0.58& 0.55& 1\\
         \end{array}\right],
\end{equation}
$P_1(0,t)=\exp(-0.01t)$, $P_2(0,t)=\exp(-0.015t)$,
$P_3(0,t)=\exp(-0.02t)$, $\lambda_b=0.003$,  $\lambda_1=0.003$, $\lambda_2=0.002$, $\lambda_3=0.001$, $\eta_b=0.01$, $\eta_1=0.01$,$\eta_2=0.02$,$\eta_3=0.003$, $y_1(t_0)=1.2$, $y_2(t_0)=0.86$, $y_3(t_0)=4.59.$ Considered portfolio consisted of about 30 randomly chosen interest rate swaps.

%\begin{table}[!h]
%\begin{center}\footnotesize
%\caption{????}
% \begin{tabular}{c|c|c|c|c|c||c|c|c|c|c|}
%  \hline
%   &\multicolumn{5}{c||}{Portfolio 1}&\multicolumn{5}{c}{Portfolio 2}  \\\hline
%   No. &P/R &Notional &K&$T_{start}$&$T_{end}$&P/R &Notional &K&$T_{start}$&$T_{end}$\\ \hline
%   1&\\
%\hline
% \end{tabular}
% \label{table:NDCase}
% \end{center}
% \end{table}

\section{Pricing of Swaptions under the Gaussian 2-Factor Model}
\label{sec:appendix3}
Under the Gaussian 2 Factor model, the short rate process, $r(t)$, is defined as: $r(t)=x^r(t)+y^r(t)+\psi^r(t)$, where processes $x^r(t)$ and $y^r(t)$ are defined by the following system of SDEs:
\begin{eqnarray*}
\d x^r(t)&=&-\lambda_1x^r(t)\dt + \eta_1\dW_1(t),\;\;\;x^r(t_0)=0,\\
\d y^r(t)&=&-\lambda_2y^r(t)\dt + \eta_2\dW_2(t),\;\;\;y^r(t_0)=0,
\end{eqnarray*}
with $\dW_1(t)\dW_2(t)=\rho\dt$, and where $\psi^r(t)$ is known in a closed form~\cite{BrigoMercurio:2007}, which is an explicit function of the zero-coupon bonds, $P_{Mrkt}(t_0,T)$, available in the market. Given simulated realizations $\{x^r(t_i),y^r(t_i)\}$ one can establish a yield curve, as a function of $T$, in terms of zero-coupon bonds,
\[P(t_i,T)=\exp\left({A(t_i,T)-\boxed{x^r(t_i)}B_1(t_i,T)-\boxed{y^r(t_i)}B_2(t_i,T)}\right),\]
where \[A(t_i,T)=\log\frac{P_{Mrkt}(t_0,T)}{P_{Mrkt}(t_0,t_i)}-\frac12\left(V^2(0,T)+V^2(0,t_i)\right),\] and $B_{j}(t_i,T)=\frac1\lambda_j(1-\e^{-\lambda_j(T-t_i)})$, $B_{1,2}(t_i,T)=\frac{1}{\lambda_1+\lambda_2}\e^{-(T-t_i)(\lambda_1+\lambda_2)}$, with $j=\{1,2\}$,
and where
{\footnotesize
\begin{eqnarray*}
V^2(t_i,T)&=&\frac{\eta_1^2}{\lambda_1^2}\left(T-t_i-B_1(t_i,T)-\frac12\lambda_2B_2^2(t_i,T)\right)+\frac{\eta_2^2}{\lambda_2^2}\left(T-t_i-B_2(t_i,T)-\frac12\lambda_1B_2^2(t_i,T)\right)\\
&&+\frac{2\eta_1\eta_2\rho}{\lambda_1\lambda_2}\left(T-t_i-B_1(t_i,T)-B_2(t_i,T)+B_{1,2}(t_i,T)\right).
\end{eqnarray*}}
\normalfont
Then the price of a European payer swaption with swaption expiry $T$ on an interest rate swap, with notional $N$, fixed rate, $K$ and a set of pay-dates $\mathcal{T}=\{T_{\alpha+1},\dots,T_{\beta}\}$:
\begin{eqnarray}
V(t_0,\mathcal{T},N,K)=NP_M(t_0,T)\int_\R\tilde\omega(x)\Big(\Phi(-h_1(x)-\sum_{i=\alpha+1}^\beta\kappa_i(x)\e^{\psi_i(x)}\Phi(-h_2(x)))\Big)\dx,
\end{eqnarray}
with the weight function, $\tilde\omega(x)$, given by:
\begin{eqnarray*}
\tilde\omega(x)=\frac{1}{\tilde\eta_1\sqrt{2\pi}}\e^{-\frac{(x-\hat\mu_1)^2}{2\hat\eta_1^2}},
\end{eqnarray*}
and where functions $h_1(x)$ and $h_2(x)$ are defined as:
\begin{eqnarray*}
h_1(x)=\frac{\bar{x}-\tilde\mu_2}{\tilde\eta_2\sqrt{1-\tilde\rho^2}}-\frac{\tilde\rho(x-\tilde\mu_1)}{\tilde\eta_1\sqrt{1-\tilde\rho^2}},\;\;h_2(x)=h_1(x)+B_2(T,T_i)\tilde\eta_2\sqrt{1-\tilde\rho^2}.
\end{eqnarray*}
Constant $\bar{x}$ is the solution of the following equation,
\[\sum_{i=\alpha+1}^\beta\kappa_i(x)\e^{-\bar x B_2(T,T_i)}=1,\]
with
\begin{eqnarray*}
\kappa_i(x)=c_i\e^{A(T,T_i)-xB_1(T,T_i)},\;\;c_i=K\tau_i,\;\;\text{for}\;\;\;\alpha<i<\beta,\;\;\;\text{and}\;\;c_\beta=1+K\tau_\beta.
\end{eqnarray*}
Finally, $\psi(x)$, and $\tilde\mu_1$ $\tilde\mu_2$ are:
\begin{eqnarray*}
\psi_i(x)&=&-B_2(T,T_i)\left(\tilde\mu_2-\frac12(\tilde\eta_2^2(1-\tilde\rho^2))B_2(T,T_i)+\tilde\rho\tilde\eta_2\frac{1}{\tilde\eta_1}(x-\tilde\mu_1)\right),\\
\tilde\mu_1&=&\frac{\eta_1^2}{2\lambda_1^2}(1-\e^{-2\lambda_1T})+\frac{\eta_1\eta_2\rho}{\lambda_2}B_{1,2}(0,T)-\left(\frac{\eta_1^2}{\lambda_1}+\frac{\eta_1\eta_2\rho}{\lambda_2} \right)B_1(0,T),\\
\tilde\mu_2&=&\frac{\eta_2^2}{2\lambda_2^2}(1-\e^{-2\lambda_2T})+\frac{\eta_1\eta_2\rho}{\lambda_1}B_{1,2}(0,T)-\left(\frac{\eta_2^2}{\lambda_2}+\frac{\eta_1\eta_2\rho}{\lambda_1} \right)B_2(0,T),
\end{eqnarray*}
with
\begin{eqnarray*}
\tilde\eta_j=\eta_j\sqrt{\frac{1-\e^{-2\lambda_jT}}{2\lambda_j}},\;\;\;\text{for}\;\;\;j\in\{1,2\},\;\;\;\text{and}\;\;\;\tilde\rho=\frac{\eta_1\eta_2\rho}{\tilde\eta_1\tilde\eta_2}B_{1,2}(0,T).
\end{eqnarray*}
\end{document}